\documentclass[a4paper]{article}

\pdfoutput=1

\usepackage[T1]{fontenc}
\usepackage[utf8]{inputenc}
\usepackage[margin=29mm]{geometry}
\usepackage{csquotes}
\usepackage[english]{babel}
\usepackage[style=numeric-comp,sorting=none,giveninits]{biblatex}
\usepackage[font=footnotesize]{caption}
\usepackage{amsmath}
\usepackage{amssymb}
\usepackage{amsthm}
\usepackage{mathtools}
\usepackage[affil-it]{authblk}
\usepackage{bm}
\usepackage{booktabs}
\usepackage{enumitem}
\usepackage{newfloat}
\usepackage{thmtools, thm-restate}
\usepackage[table]{xcolor}
\usepackage{tikz}
\usepackage{tikzscale}
\usepackage{graphicx}
\usepackage[colorlinks]{hyperref}
\usepackage[capitalise]{cleveref}
\usepackage{aligned-overset}

\DeclarePairedDelimiter{\bra}{\langle}{\rvert}
\DeclarePairedDelimiter{\ket}{\lvert}{\rangle}
\DeclarePairedDelimiterX\braket[2]{\langle}{\rangle}{#1\mathclose{}\delimsize\vert\mathopen{}#2}
\DeclarePairedDelimiterX\braxket[3]{\langle}{\rangle}{#1\mathclose{}\delimsize\vert\,\mathopen{}#2\mathclose{}\,\delimsize\vert\mathopen{}#3}
\DeclarePairedDelimiterX\ketbra[2]{\lvert}{\rvert}{#1\mathclose{}\delimsize\rangle\!\delimsize\langle\mathopen{}#2}
\DeclarePairedDelimiterX\proj[1]{\lvert}{\rvert}{#1\mathclose{}\delimsize\rangle\!\delimsize\langle\mathopen{}#1}

\usetikzlibrary{quantikz2}

\newtheorem{theorem}{Theorem}
\newtheorem{proposition}[theorem]{Proposition}
\newtheorem{lemma}[theorem]{Lemma}
\newtheorem{corollary}[theorem]{Corollary}
\theoremstyle{definition}
\newtheorem{definition}[theorem]{Definition}
\theoremstyle{remark}
\newtheorem*{remark}{Remark}

\Crefname{property}{Property}{Properties}
\Crefname{relation}{Relation}{Relations}

\AtBeginDocument{
    \let\Re\relax
    \DeclareMathOperator{\Re}{\mathfrak{Re}}

}

\newcommand{\approxe}[1]{\overset{#1}{\approx}}
\newcommand{\bigO}{O}
\DeclareMathOperator*{\E}{\mathbb{E}}
\newcommand{\hilb}[1]{\mathcal{#1}}
\newcommand{\mathrmpair}[2]{\mathrm{#1}, \mathrm{#2}}
\DeclareMathOperator{\poly}{poly}
\DeclareMathOperator{\rank}{rank}
\newcommand{\regular}[1]{\hat{#1}}
\DeclareMathOperator{\tr}{tr}

\DeclareFloatingEnvironment{protocol}
\Crefname{protocol}{\protocolname}{\protocolname s}

\begin{document}

\title{Parallel remote state preparation for fully \\ device-independent verifiable blind quantum computation}
\author{Sean A. Adamson\thanks{\href{mailto:sean.adamson@ed.ac.uk}{\texttt{sean.adamson@ed.ac.uk}}}}
\affil{School of Informatics, University of Edinburgh, \protect\\ 10 Crichton Street, Edinburgh EH8 9AB, United Kingdom}
\date{}

\maketitle

\begin{abstract}
We introduce a device-independent two-prover scheme in which a classical verifier can use a simple untrusted quantum measurement device (the client device) to securely delegate a quantum computation to an untrusted quantum server. To do this, we construct a parallel self-testing protocol to perform device-independent remote state preparation of $n$ qubits and compose this with the unconditionally secure universal verifiable blind quantum computation (VBQC) scheme of Fitzsimons and Kashefi [\href{https://doi.org/10.1103/PhysRevA.96.012303}{Phys. Rev. A \textbf{96}, 012303 (2017)}]. Our self-test achieves a multitude of desirable properties for the application we consider, giving rise to practical and fully device-independent VBQC. It certifies parallel measurements of all cardinal and intercardinal directions in the $XY$ plane as well as the computational basis, uses few input questions (of size logarithmic in $n$ for the client and a constant number communicated to the server), and requires only single-qubit measurements to be performed by the client device.
\end{abstract}

\section{Introduction}

With the advent of cloud-based quantum computing services (such as those now offered by IBM, Amazon, and Microsoft among others \cite{castelvecchi2017ibm,alsina2016experimental,devitt2016performing,ibmquantum,amazonbraket,azurequantum}), it is becoming increasingly important to allow the secure delegation of quantum computations to powerful remote quantum servers.
In such a scenario, a client wishes to securely delegate some computation to one of these remote servers.
Such a client desires that the remote server cannot learn about the computation (a property called ``blindness'') and that they can be sure that the computation was performed correctly by the server (called ``verifiability'').
This is known succinctly as verifiable blind delegated quantum computation (VBQC).
Ideally, the delegating party would only require strictly classical capabilities.
This has been done \cite{fitzsimons2018post,natarajan2017quantum,hayashi2018self}, but such protocols come with disadvantages such as not exhibiting blindness, requiring many provers, or assuming a tensor product structure of the provers' systems.
Instead, one may allow the client party to have some minimal quantum technological capabilities (such as the ability to perform single-qubit measurements) that are foreseeable of possible future personal quantum devices (e.g., those that might fit inside a mobile phone).
Unconditionally secure protocols for VBQC already exist provided that the client can initially ensure the preparation of a product of single-qubit states on the server side.
One of the most prevalent of these is the seminal protocol of \textcite{fitzsimons2017unconditionally}, which we henceforth refer to as the \emph{FK protocol}.
It requires input states to be prepared in any of the eight cardinal and intercardinal directions of the $XY$ plane or the computational basis, and is based on the measurement-based quantum computing (MBQC) model \cite{raussendorf2001one,danos2007measurement,broadbent2009parallelizing}.
Further improvements have since been made to the protocol such as reducing the overhead of qubits involved in verification \cite{kashefi2017optimised,xu2020improved}.
While possible deviations by the server are taken into account in such protocols, the level of trust given to client side devices must also be taken into consideration.
This is important since the devices held by a client are likely to be error-prone and could have been prepared by external parties.
The most general form of security for protocols in this context is known as ``device-independence'' \cite{mayers1998quantum,colbeck2011private}, in which no assumptions are made about the honesty of the devices (they may even be adversarially prepared).
In contrast, blind verifiable delegation protocols such as \cite{aharonov2017interactive,morimae2014verification,hayashi2015verifiable,fitzsimons2017unconditionally,broadbent2018verify,fujii2017verifiable,morimae2017verification,fitzsimons2018post} are not inherently device-independent \cite{fitzsimons2017private}.
One promising approach to the desired device-independent remote state preparation is that of ``self-testing'' \cite{mayers1998quantum,mayers2004self}.
The idea is for a classical verifier (in this case the client) to certify the existence of maximally-entangled states shared between two provers (in this case one being the quantum device of the client and the other being the server) from measurement statistics alone.
The client side prover then performs particular measurements on the entangled states which, depending on their outcomes, teleports particular states to the server side.
Most self-testing protocols, however, do not exhibit the qualities required for practical composability with FK-type VBQC protocols \cite{gheorghiu2015robustness}.
Typical approaches have thus far either prepared states sequentially \cite{hajdusek2015device} or have appealed to other verifiable protocols such as \cite{broadbent2018verify} and proven their blindness property separately as inherited from the use of self-testing itself \cite{coladangelo2019verifier}.

\subsection{Our contributions}

We exhibit a two-prover parallel self-testing-based scheme in which a classical verifier is able to delegate a quantum computation to an untrusted quantum server Bob (who is assumed to be in possession of a powerful universal quantum computer) using only a simple untrusted measurement device (which may only perform single-qubit measurements) and shared entanglement.
In our context, Alice acts both as the verifier and the client who is in possession of the measurement device.
The computation is performed blindly by the server and its correctness is verifiable by the client.
The protocol proceeds for the verifier in the following way:
\begin{enumerate}
    \item Certify $n$ EPR pairs of entanglement shared between Alice and Bob, and measurements of each in the $XY$ plane and computational bases.
    \item Prepare a suitable $n$-qubit state on Bob's side by measurements performed on Alice's side.
    \item Perform the unconditionally secure interactive protocol of \textcite{fitzsimons2017unconditionally} (or another FK-type protocol) with Bob.
\end{enumerate}
The blindness and verifiability properties of our protocol are inherited from the FK protocol, since using an entanglement-based approach to the remote state preparation ensures that blindness is not compromised \cite{gheorghiu2015robustness,dunjko2016blind}.
Our main result is the parallel device-independent certification of remotely prepared states that can be used for the first two steps of the protocol.
\begin{theorem}[Informal \cref{thm:state_preparation}]
\label{thm:state_preparation_informal}
    Suppose that the maximal quantum expectation values of all given Bell expressions are approximately attained by the measurement statistics collected from provers Alice and Bob, i.e., Alice and Bob pass all requested Bell tests with high probability (this is achievable by using the honest strategy of \cref{sec:honest_strat}).
    Then, with high probability, the state on Bob's side upon Alice performing a measurement is close (up to some isometry that is independent of the measurement or its outcome) in trace distance to an ideal valid input to the FK protocol or that with all dummy qubits flipped.
    The prepared state is known only to Alice.
\end{theorem}

The self-testing protocol on which we base \cref{thm:state_preparation_informal} (whose resulting statement is given in \cref{thm:protocol_isometry}) simultaneously exhibits many properties desirable for the VBQC application we consider.
The rationale behind all of these properties is discussed in further detail in \cref{sec:efficient_testing}.
We briefly summarize them here.

Our self-test is parallel, meaning that $n$ Bell states are certified at once with no prior assumption on the tensor product structure of the underlying state space.
Permutations of $n$ single-qubit measurements on Alice's side, each either in the computational basis or one of the four bases of the eight canonical states of the $XY$ plane (those corresponding to all cardinal and intercardinal directions) are also certified.
The possible correlated complex conjugation freedom that arises for measurement operators of this kind is accounted for and, moreover, is limited to measurements in the computational basis (so as only to affect the preparation of dummy qubits for the VBQC protocol).
The number of possible input questions in the test is small.
The client side measurement device is asked questions of size logarithmic in $n$, while questions of only a constant size need be communicated to the server.
In the honest case, Alice need only perform single-qubit measurements local to each of the EPR pairs she shares with Bob.
This reflects the minimal quantum capabilities she is given access to.
Moreover, despite Bob being assumed to possess a powerful universal quantum computer, he need only perform two-qubit Bell measurements in the self-testing protocol, and most of the time only measures single qubits.
The isometry on Bob's subsystem (and resulting reduced junk state) guaranteed by the self-test is independent of any string of $n$ measurement bases selected for state preparation.
This ensures composability with the FK protocol.
Classical processing of the gathered raw experimental data required for the self-test scales efficiently in $n$.

While we do not explicitly attempt to derive robustness bounds for our self-testing statement, we believe that our derivations are compatible with the many standard techniques that have been developed and used successfully for this purpose in other works \cite{mckague2016self,mckague2017self,coladangelo2017parallel,chao2018test,bowles2018self}.
Nevertheless, we do phrase all of the subtests comprising our full self-test in terms of an error tolerance $\varepsilon$, and derive all relations between operators in terms of this quantity.
We therefore expect that our self-testing statement would exhibit analytic robustness at worst $\bigO(\sqrt{\varepsilon} n^{2})$, thereby giving a trace distance that is $\bigO(\varepsilon^{1/3} n^{4/3})$ in \cref{thm:state_preparation_informal}.
Large improvements in robustness having been shown achievable using numerical optimization techniques such as semidefinite programming \cite{wu2016device,navascues2007bounding,navascues2008convergent,navascues2015almost,yang2014robust,bancal2015physical,supic2021device,wang2016all,wu2014robust,agresti2021experimental}.

The information aggregated from experimental outcomes that is used in each subtest is \emph{local} in the sense that it corresponds to measurements of only individual or pairs of the $n$ Bell states in the honest strategy (despite being conditioned on the measurements that are asked for at other positions).
The quantity $\varepsilon$ for each subtest then does not refer to the noise and statistical uncertainty present over all $n$ Bell states, but rather for constant sized chunks of the experimental resources.
Thus, in a (possibly noisy) physical implementation of the honest strategy, the error estimate $\varepsilon$ would not typically increase with the number of qubits being prepared $n$.
It is also for this reason that the exponential number of possible outcomes in $n$ associated with each question does not lead to the estimation of probabilities requiring exponential time; local consideration of outcomes effectively transforms exponentially many probabilities with one distribution per question to a linear number of distributions per question each with a constant number of probabilities to estimate.

It is important to acknowledge that while we demonstrate several practical advantages over other device-independent VBQC schemes, it is still infeasible with current hardware to achieve a sufficiently small robustness for large enough, computationally interesting $n$ \cite{nadlinger2022experimental}.
In practice, we would likely require fault-tolerant devices, in which many physical qubits are error corrected to perform local measurements of required logical qubits.
Fortunately the present protocol is compatible with fault-tolerance \cite{gheorghiu2015robustness}.
It should also be stated that in the present work we do not consider the well-known problem of reusing devices in device-independent protocols.
That is, devices containing an internal memory could be prepared that record inputs and outputs and then (if reused) leak information about them in later runs.
This issue was first identified by \textcite{barrett2013memory}.
Further information on this issue can be found in the review by \textcite{portmann2022security}.

Many of our results are not specific to VBQC alone.
In particular, our self-test and remote state preparation protocol have properties that may be desirable for other quantum delegation applications.
Furthermore, due to the range of measurements we are able to certify, our tests could be easily adapted to the remote preparation of other states.

\subsection{Overview of techniques}

The basis for our self-testing is a careful consideration of the statistics one would expect to find from parallel Clauser--Horne--Shimony--Holt (CHSH) measurements of $n$ maximally entangled Bell states shared between provers Alice and Bob \cite{clauser1969proposed}.
We take Alice to be the one performing measurements of Pauli observables from the standard strategy for the CHSH game.
By appropriately chunking the raw data received by the verifier into outcomes for the different questions of a local CHSH inequality conditioned on the different possible fixed input questions asked at other position, it is possible to construct sets of CHSH-type inequalities for each of the individual Bell states, all of which would be saturated with honest behavior.
We proceed to show the opposite of this---that the saturation of these inequalities is sufficient to prove all operator relations required and achieve a self-testing statement for the Bell states and CHSH measurements.
Moreover, we show that after removing many of the requested inequalities, the remaining tests are still sufficient.
Enough of the tests can be removed that the total number of remaining tests (and thus input questions) is reduced to scale quadratically (rather than exponentially) in $n$, and there are only a constant number of possible actions that Bob need take.
Intuitively, the players cannot cheat in the tests by sharing fewer than $n$ Bell states since Alice cannot be sure which of $n$ positions of Bob she is being tested against, while at the same time Bob does not know how correlated different positions in Alice's input question are with one another.
The reduction in the number of our questions comes from the fact that only pairwise correlations in Alice's question strings must be hidden (see \cref{sec:protocol} for further details).

We show that the ``triple CHSH'' inequalities introduced by \textcite{acin2016optimal} can be used to extend our technique to include certification of all Pauli observables $\sigma_{\mathrm{x}}$, $\sigma_{\mathrm{y}}$, and $\sigma_{\mathrm{z}}$.
We construct our isometry such that the complex conjugation ambiguity appears in $\sigma_{\mathrm{z}}$ rather than the usual $\sigma_{\mathrm{y}}$ measurements.
We then introduce further tests (also efficient), based on perfect correlations between further measurements for Alice and those present already for Bob, that ensure that these additional untrusted measurements for Alice certify reference measurements of the intercardinal directions of the $XY$ plane, as required to generate input states to the FK VBQC protocol.
Finally, we augment our self-test thus far with a test that expects Bob to perform Bell measurements on two sets of pairs of his qubits, in order to ensure that any possible correlated complex conjugation of the provers' measurement operators occurs globally across all $n$ of their registers (this is similar to techniques used in \cite{coladangelo2019verifier,bowles2018self} for the same purpose).

One drawback of the technique we use to reduce Alice's questions to quadratic order is that the resulting local isometry is only able to certify the measurement operators for a constant number of choices of bases for Alice's measurement of the $n$ EPR pairs.
The greater structure present in her restricted set of possible inputs may leak some information about this choice of bases, which in turn would allow Bob to gain some knowledge of the states prepared for him and cheat in the subsequent VBQC interactive protocol.
To remedy this, we instead use polynomially many different sets of our quadratically many questions (polynomially many questions in $n$ overall) and perform the certification for each of these.
This results in a polynomial-sized subset $\mathcal{S}$ of questions for Alice (which we call ``special'' questions), for each of which a different local isometry certifies a different string of bases measured on Alice's reference system.
In order for our remote state preparation protocol to be composable with the FK interactive protocol, it must be the case that states are prepared up to an isometry that is independent of the choice of bases in $\mathcal{S}$ (otherwise one could not assume that the physical state of Bob originates from local quantum operations applied to his ideal reference state without also assuming that he has knowledge of the bases chosen \cite{gheorghiu2015robustness}).
This does not present a problem for the use of our self-test as, despite each question in $\mathcal{S}$ requiring a different local isometry, we show that the isometry local to Bob's subsystem is the same in all of these cases.
While it is the case that the resulting security parameter for the FK protocol will go as the reciprocal of $\lvert \mathcal{S} \rvert$, which is inverse polynomial in $n$ for our choice of questions, this trade-off between question size and security scaling is an inescapable feature of any remote state preparation protocol used for FK-type protocols.

Since our protocol must perform remote state preparation, we are interested in self-testing statements that estimate the closeness of (normalized) physical postmeasurement states from their (normalized) ideal counterparts.
The robustness guarantees usually given by self-testing statements estimate this distance for observables acting on states, which naively lead to similar estimates for (subnormalized) measurement projectors acting on states.
This is acceptable for protocols that prepare states sequentially (such as in \cite{reichardt2013classical1,reichardt2013classical2,gheorghiu2015robustness,hajdusek2015device}), however, for parallel protocols of $n$ states (which have exponentially many outcomes per measurement) would lead to robustness estimates that scale exponentially in $n$ for post-measurement states.
We overcome this using \cref{lem:general_robust_prob,thm:robust_prob} of \cref{sec:post-measurement_states} at the cost of relaxing the estimate by a factor that is polynomial in the original robustness and allowing acceptance with high probability.

\subsection{Related works}

General composability of delegated quantum computation was studied by \textcite{dunjko2014composable}.
The original two-prover protocol for VBQC that we make use of is that of \textcite{fitzsimons2017unconditionally}, which is a verifiable extension to the blind protocol of \textcite{broadbent2009universal}.
Other forms of resource states and repetition schemes have since been devised in order to improve the practicality and overhead of the protocol \cite{kashefi2017optimised,xu2020improved,kashefi2021securing}.
This VBQC protocol was proven to be robust and composable with device-independent state preparation protocols by \textcite{gheorghiu2015robustness}.
In this work, they also used the CHSH rigidity results of \textcite{reichardt2013classical1,reichardt2013classical2} to achieve such preparation sequentially, resulting in a device-independent VBQC protocol using total resources scaling like $g^{2048}$, where $g$ is the number of gates in the circuit to be delegated.
A more efficient scheme (again with sequential preparation but this time based on self-testing) was given by \textcite{hajdusek2015device} and uses $\Theta(g^{4} \log{g})$ resources.
This, however, requires that the server party is in possession of $n$ spacelike separated provers, which is difficult to achieve in practice and cannot be verified to be the case by the client.
Another many-prover protocol was presented by \textcite{mckague2016interactive} based on self-testing graph states.
A more recent approach is the so-called ``verifier-on-a-leash'' protocol of \textcite{coladangelo2019verifier}.
This device-independent protocol is based on the verifiable delegation approach of \textcite{broadbent2018verify} rather than the FK protocol, with the blindness property a result of its combination with self-testing.
The efficient resource usage $\Theta(g \log{g})$ is primarily a result of the self-testing protocol used: a modified version of the ``Pauli braiding test'' of \textcite{natarajan2017quantum}.
However, this protocol requires that the client possesses a more powerful quantum measuring device (being able to perform joint measurements of multiple qubits), has exponentially many questions of size $\bigO(n)$ bits, and is not fault-tolerant.
The robustness guarantee given by the underlying self-testing protocol is $\poly(\epsilon)$, where $\epsilon$ is the overall rejection probability in the test (a quantity that is not directly comparable with the local error tolerances $\varepsilon$ used in many other self-testing statements including our own).
The quantum steering scenario, in which the device of one party is entirely trusted, was considered in the context of VBQC by \textcite{gheorghiu2017rigidity}.
In the setting of \emph{computational} security, there exist recent single-prover protocols such as those of \textcite{mahadev2018classical,gheorghiu2019computationally,gheorghiu2023quantum}, the latter of which perform remote state preparation in this setting.

The concept of self-testing was first introduced by \textcite{mayers1998quantum} in a cryptographic context, with the first mention of the term ``self-testing'' appearing in \cite{mayers2004self}.
The question of which states can be self-tested was answered by \textcite{coladangelo2017all} in the bipartite case and later for the multipartite case \cite{supic2023quantum}.
Arbitrary parallel self-testing of EPR pairs was first introduced by \textcite{mckague2016self} and a host of self-tests for entangled states of arbitrary local dimension have now been proposed \cite{yang2013robust,mancinska2014maximally,mckague2016self,coladangelo2017parallel,coudron2016parallel,mckague2017self,ostrev2016structure,coladangelo2017all,coladangelo2018generalization,chao2018test,natarajan2017quantum,natarajan2018low,bowles2018self,natarajan2019neexp,mancinska2021constant,supic2021device,sarkar2021self,fu2022constant,adamson2022practical}.
Complex conjugation ambiguity in self-testing complex measurements was first recognized by \textcite{mckague2011generalized}.
The triple CHSH inequality was first introduced by \textcite{acin2016optimal} and subsequently used for self-testing by \textcite{bowles2018self}, whose results we make use of in the present work.
It was also used by \textcite{renou2021quantum} to devise an experiment to rule out quantum theory with real numbers.
Commutation-based measures were introduced by \textcite{kaniewski2017self} and also used to certify multiple anticommuting observables.
Works of \textcite{mckague2016self} and \cite{coladangelo2017parallel} gave general theorems for converting certain sets of approximate commutation relations between observables to robust self-testing statements (with real measurements).
The latter is based on work by \textcite{chao2018test} and also gives self-testing statements for parallel repeated CHSH games, while the former was later used as such by the same author \cite{mckague2017self}.
Parallel self-tests with some of the desirable properties for state preparation were given in \cite{adamson2022practical}, and are based on the family of ``magic rectangle'' nonlocal games \cite{adamson2020quantum}.
The possibility of self-testing with just a single prover by replacing nonlocal correlations by computational assumptions was examined by \textcite{metger2021self}.
More details on self-testing can be found in the excellent review by \textcite{supic2020self}.

\subsection{Organization of the paper}

In \cref{sec:prelims}, we give details of the notation and basic results of which we make frequent use.
In \cref{sec:efficient_testing}, we explain in further detail the desirable properties required of self-testing protocols that are to be used for remote state preparation in the context of VBQC.
The triple CHSH inequality of \textcite{acin2016optimal} is exhibited in \cref{sec:triple_chsh}, and we use it to show a single-copy self-testing statement that is used as a building block in later results.
\Cref{sec:parallel_self-testing} contains our results linking the existence of certain operator relations to a self-testing statement with desired properties, as well as a result for lifting parallel self-testing statements certifying only single-qubit observables to those with arbitrary tensor products of observables.
Our main protocol is outlined in \cref{sec:protocol}, where we give some intuition behind our self-testing protocol, state our main remote state preparation result (\cref{thm:state_preparation}), detail the construction of Alice's question set in \cref{sec:alice_questions}, define the measurement scenario and tests required in \cref{sec:requested_observations}, and give the honest strategy in \cref{sec:honest_strat}.
We prove in \cref{sec:operator_rels} that approximate acceptance in our tests yields approximate operator relations, and thus give our formal self-testing statements of \cref{thm:protocol_isometry,cor:protocol_isometry}.
We finish in \cref{sec:discussion} with some discussion and possible future directions.

As there are many interconnected results, let us summarize the main tree of dependencies.
Our main result regarding state preparation in the context of our VBQC application is stated in \cref{thm:state_preparation}.
Its proof is contained in \cref{sec:state_preparation_proof} and makes use of \cref{cor:protocol_isometry,thm:robust_prob}.
In turn, \cref{cor:protocol_isometry} is a special case of the more general self-testing statement \cref{thm:protocol_isometry} (giving us device-independence) after applying \cref{lem:single_to_multiple_observable} in order to certify simultaneous measurements of all qubits, rather than just one.
\Cref{thm:robust_prob} is used to ensure the distance of prepared states from ideal does not blow up exponentially in the number of qubits.
Our general self-testing statement, \cref{thm:protocol_isometry}, is a consequence of \cref{thm:single_observable_isometry} (a general result on self-testing) and relies on the Bell expressions and correlations described in \cref{sec:requested_observations} being experimentally satisfied.
To aid in its proof, \cref{thm:assumptions_satisfied} converts these particular expressions to a number of operator relations (of a standard form as found in the statement of \cref{thm:single_observable_isometry}).
This is done by combining \cref{prop:symmetry,prop:com_alice,prop:com_bob,prop:acomm_alice,prop:acomm_bob,prop:conj_rel} of \cref{sec:individual_relations}.

\section{Preliminaries}
\label{sec:prelims}

We begin by introducing in \cref{sec:notation} the notation that will be used throughout the paper and some elementary results thereof.
We then introduce some of the definitions for the self-testing of quantum systems in \cref{sec:self-testing}.
In \cref{sec:sos_decomposition} we discuss sum-of-squares decomposition for Bell operators.
In \cref{sec:regularization}, we detail the unitary regularization technique for operators and exhibit a convenient result on the anticommutativity of such operators.
Finally, in \cref{sec:vbqc}, we briefly outline how blindness and verifiability are achieved in the FK-type VBQC protocols.

\subsection{Notation and elementary results}
\label{sec:notation}

In this section, we introduce notation that will be used throughout.
We also state some useful properties of the mathematical objects presented.

\subsubsection{Pauli observables and important states}

The Pauli observables $\sigma_{\mathrm{x}}$, $\sigma_{\mathrm{y}}$, and $\sigma_{\mathrm{z}}$ will be denoted interchangeably by $\sigma_{1}$, $\sigma_{2}$, and $\sigma_{3}$, respectively.
In the computational basis, these have matrices
\begin{subequations}
\begin{align}
    \sigma_{\mathrm{x}} = \sigma_{1} & =
    \begin{pmatrix}
        0 & \phantom{-} 1 \\
        1 & \phantom{-} 0
    \end{pmatrix} , \\
    \sigma_{\mathrm{y}} = \sigma_{2} & =
    \begin{pmatrix}
        0 & -i \\
        i & \phantom{-} 0
    \end{pmatrix} , \\
    \sigma_{\mathrm{z}} = \sigma_{3} & =
    \begin{pmatrix}
        1 & \phantom{-} 0 \\
        0 & -1
    \end{pmatrix} .
\end{align}
\end{subequations}
We will define the $\ket{\pm_{\theta}}$ qubit states for $\theta \in \mathbb{R}$ as
\begin{equation}
    \ket{\pm_{\theta}} = \frac{\ket{0} \pm e^{i \theta} \ket{1}}{\sqrt{2}} .
\end{equation}
We define operators $\sigma_{\mathrm{x} + \mathrm{y}}$ and $\sigma_{\mathrm{x} - \mathrm{y}}$ (denoted interchangeable by $\sigma_{4}$ and $\sigma_{5}$) as
\begin{equation}
    \sigma_{\mathrm{x} + \mathrm{y}} = \sigma_{4} = \frac{\sigma_{\mathrm{x}} + \sigma_{\mathrm{y}}}{\sqrt{2}} ,\quad
    \sigma_{\mathrm{x} - \mathrm{y}} = \sigma_{5} = \frac{\sigma_{\mathrm{x}} - \sigma_{\mathrm{y}}}{\sqrt{2}} .
\end{equation}
The $\ket{\pm_{0}} = \ket{\pm}$ states are eigenvectors of $\sigma_{\mathrm{x}}$, and the $\ket{\pm_{2 \pi / 4}} = \ket{\pm i}$ states are eigenvectors of $\sigma_{\mathrm{y}}$.
Meanwhile, the $\ket{\pm_{\pi / 4}}$ states are eigenvectors of $\sigma_{\mathrm{x} + \mathrm{y}}$, and the $\ket{\pm_{3 \pi / 4}}$ states are eigenvectors of $\sigma_{\mathrm{x} - \mathrm{y}}$.
Another convenient notation that we make use of is writing $\ket{\sigma_{\chi}^{\lambda}}$ for the eigenstate of $\sigma_{\chi}$ having corresponding eigenvalue $\lambda \in \{+1, -1\}$, where $\sigma_{\chi}$ is any of the five operators just introduced with $\chi \in \{\mathrm{x}, \mathrm{y}, \mathrm{z}, \mathrm{x} + \mathrm{y}, \mathrm{x} - \mathrm{y}\}$.

We denote by $\ket{\Phi^{+}}$ the maximally entangled Bell state
\begin{equation}
    \ket{\Phi^{+}} = \frac{\ket{00} + \ket{11}}{\sqrt{2}} .
\end{equation}

\subsubsection{Complex conjugation and composition}

For a vector $\ket{v}$ or linear map $M$, we denote using a star symbol as $\ket{v}^{*}$ or $M^{*}$ the complex conjugation of their matrix entries with respect to a fixed basis.
This is not to be confused with the notation $\mathbb{N}^{*}$, which we use for the set of natural numbers excluding $0$.
To denote the composition of many linear maps, we adopt product notation, with the convention for the order in which they are applied given by
\begin{equation}
    \prod_{j=1}^{n} M_{j} = M_{1} \dots M_{n} .
\end{equation}
Given also some $\bm s \in \{0, 1\}^{n}$, we define the related notation denoted by an operator ``raised to the power'' of this string by
\begin{equation}
    M^{\bm s} = \prod_{j=1}^{n} M_{j}^{s_{j}} .
\end{equation}
This is an abuse of notation in which $M$ has not been defined on its own, but is nonetheless used in the notation on the left-hand side due to the same letter being used for all the given $M_{j}$ on the right-hand side.
It will always be clear in context to which set of operators $\{ M_{j} \mid 1 \leq j \leq n \}$ the notation $M^{\bm s}$ is associated with due to the choice of capital letter being used.

\subsubsection{Hilbert space notation}

Hilbert spaces will be denoted using calligraphic symbols, for example $\mathcal{H}$.
In the context of delegated computation, Alice will refer to the client party, while Bob will refer to the server party.
Hilbert spaces corresponding to the parties Alice and Bob will be denoted by variations of the symbols $\hilb{A}$ and $\hilb{B}$, respectively.
In certain situations, it will be useful to explicitly write to which space a state belongs.
To this end, we sometimes write the Hilbert space as a subscript $\ket{\psi}_{\hilb{H}} \in \hilb{H}$.
Similarly, for a linear map $M_{\hilb{H}}$, a subscript or superscript calligraphic symbol denotes its domain.
In the case of a joint space $\hilb{A} \otimes \hilb{B}$, we may omit the tensor product symbol in this notation, so that $\ket{\psi}_{\hilb{A} \hilb{B}} \in \hilb{A} \otimes \hilb{B}$.
Given a state $\ket{\psi}$ in a joint space $\hilb{A} \otimes \hilb{B}$ and a linear map $M$ defined on $\hilb{A}$, we often also denote by $M$ its extension to $\hilb{A} \otimes \hilb{B}$ which acts trivially on $\hilb{B}$.
That is, we adopt the notation
\begin{equation}
    M \ket{\psi} = (M \otimes I_{\hilb{B}}) \ket{\psi} ,
\end{equation}
where $I_{\hilb{B}}$ is the identity operator on $\hilb{B}$.

\subsubsection{Properties of norms}

The norm $\lVert \cdot \rVert$ will refer throughout to that induced by the inner product of the Hilbert space being considered as $\lVert \ket{v} \rVert = \sqrt{\braket{v}{v}}$.
In the case of a linear operator $O$ defined on this Hilbert space, $\lVert O \rVert_{p}$ will refer to its Schatten $p$-norm.
\begin{definition}[Schatten $p$-norm]
    Let $\hilb{H}_{1}$ and $\hilb{H}_{2}$ be Hilbert spaces.
    For $1 \leq p < \infty$, the \emph{Schatten $p$-norm} of a bounded linear operator $O \colon \hilb{H}_{1} \to \hilb{H}_{2}$ is given by
    \begin{equation}
        \lVert O \rVert_{p} = \tr(\lvert O \rvert^{p})^{\frac{1}{p}} ,
    \end{equation}
    where $\lvert O \rvert = (O^{\dagger} O)^{1/2}$.
    If, moreover, $O$ is compact and both $\hilb{H}_{1}$ and $\hilb{H}_{2}$ are separable then, equivalently,
    \begin{equation}
        \lVert O \rVert_{p} = \Biggl( \sum_{j} s_{j}^{p} \Biggr)^{\frac{1}{p}} ,
    \end{equation}
    where the $s_{j} \geq 0$ are the singular values of $O$ (i.e., the eigenvalues of $\lvert O \rvert$) given in descending order.
\end{definition}

Of particular importance are the trace class norm (the case where $p = 1$) and the operator norm (conventionally denoted with $p = \infty$).
We also denote the operator norm without any subscript.
The operator norm has the important properties that $\lVert O \ket{v} \rVert \leq \lVert O \rVert \lVert \ket{v} \rVert$ for all vectors $\ket{v}$, and that $\lVert U \rVert = 1$ if $U$ is unitary.
The trace class norm satisfies
\begin{equation}
\label{eq:traced_out_norm}
    \lVert \tr_{\hilb{B}}(O) \rVert_{1} \leq \lVert O \rVert_{1}
\end{equation}
if $O$ is defined on a joint Hilbert space $\hilb{A} \otimes \hilb{B}$ \cite{rastegin2012relations}.

\subsubsection{Tolerance relations}

We define a tolerance relation $\approxe{\varepsilon}$ to denote when two vectors are $\varepsilon$-close in the vector norm distance.
Given vectors $\ket{u}$ and $\ket{v}$ in the same Hilbert space, this relation is defined by
\begin{equation}
    \ket{u} \approxe{\varepsilon} \ket{v} \iff \lVert \ket{u} - \ket{v} \rVert \leq \varepsilon .
\end{equation}
By the triangle inequality, we can then succinctly state the property that
\begin{equation}
    \ket{u} \approxe{\varepsilon} \ket{v} \text{ and } \ket{v} \approxe{\delta} \ket{w}
    \implies \ket{u} \approxe{\varepsilon + \delta} \ket{w} .
\end{equation}
The following lemma will prove useful to estimate the action of unitary operators on some state.
\begin{lemma}
\label{lem:state_estimate}
    Let $\ket{\varphi}$ and $\ket{\chi}$ be normalized states belonging to the same Hilbert space and let $\varepsilon \geq 0$.
    The real part $\Re{\braket{\varphi}{\chi}} \geq 1 - \varepsilon$ if and only if
    \begin{equation}
        \ket{\varphi} \approxe{\sqrt{2 \varepsilon}} \ket{\chi} .
    \end{equation}
\end{lemma}
\begin{proof}
    Using the property $\lVert \ket{v} \rVert = \sqrt{\braket{v}{v}}$, we can expand
    \begin{equation}
    \begin{split}
        \lVert \ket{\varphi} - \ket{\chi} \rVert^{2}
        & = \braket{\varphi}{\varphi} + \braket{\chi}{\chi} - \braket{\varphi}{\chi} - \braket{\chi}{\varphi} \\
        & = 2 - \braket{\varphi}{\chi} - \braket{\varphi}{\chi}^{*} \\
        & = 2 (1 - \Re{\braket{\varphi}{\chi}}) .
    \end{split}
    \end{equation}
    Therefore, $\lVert \ket{\varphi} - \ket{\chi} \rVert^{2} \leq 2 \varepsilon$ if and only if $1 - \Re{\braket{\varphi}{\chi}} \leq \varepsilon$.
\end{proof}

\subsubsection{Number strings}

Given any string of length $n$, which we denote in bold by $\bm x = (x_{1}, \dots, x_{n})$, we sometimes find it convenient to consider the same string but of length $n-1$ and with its $i$th element removed.
We write this as the original symbol (in bold) for the string with a subscript as
\begin{equation}
    \bm{x}_{i} = (x_{1}, \dots, x_{i-1}, x_{i+1}, \dots, x_{n}) .
\end{equation}
To reiterate: $x_{i}$ is the $i$th element of $\bm x$, while $\bm{x}_{i}$ is $\bm x$ with its $i$th element removed.
An exception to this is given to the $i$th standard basis vector of length $n$, which we always denote unambiguously by
\begin{equation}
    \bm{e}_{i}^{n} = (\delta_{ij})_{j=1}^{n} .
\end{equation}
Let $[P]$ denote the Iverson bracket of a statement $P$.
We sometimes use this notation as an alternative way to express the Kronecker delta function.
That is, we will take
\begin{equation}
    \delta_{jk} = [j = k] .
\end{equation}

\subsubsection{Complexity}

We interchangeably express the functions with argument $n$ that are $2^{\bigO(\log{n})}$ by writing that they are $\poly(n)$.

\subsubsection{Bell expressions}

A (linear) Bell expression $\mathcal{I}$ is defined as a real-valued function taking experimental probabilities $\bm p = (p(a,b \mid x,y))_{a,b,x,y}$ to a linear combination
\begin{equation}
    \mathcal{I}[\bm p] = \sum_{a,b,x,y} \beta_{x,y}^{a,b} p(a,b \mid x,y) ,
\end{equation}
where $\beta_{x,y}^{a,b} \in \mathbb{R}$.
For notational convenience, we often write values of a Bell expression as $\mathcal{I}[p(a,b \mid x,y)]$, where it is understood that all probabilities (varying over $a$, $b$, $x$, and $y$) are arguments to $\mathcal{I}$.

\subsection{Self-testing (with complex measurements)}
\label{sec:self-testing}

In a self-testing scenario, two observers Alice and Bob share an unknown physical quantum state $\rho$ on $\hilb{A} \otimes \hilb{B}$.
For convenience, it is common to instead work with a purification of the physical state $\ket{\psi} \in \hilb{A} \otimes \hilb{B} \otimes \hilb{P}$ for some purifying space $\hilb{P}$ separate from the observers.
Since all operations accessible to the observers act trivially on this purifying space, we suppress it in our notation, and treat $\ket{\psi} \in \hilb{A} \otimes \hilb{B}$ as the physical state.
One may also assume that the measurements are projective (see \cite[Appendix~B]{supic2020self} for a detailed discussion of this topic).
The Born rule states that the probability of outcomes $a$ and $b$ upon being provided with inputs $x$ and $y$ is given by
\begin{equation}
\begin{split}
    p(a,b \mid x,y)
    & = \tr(\proj{\psi} M_{a \mid x} \otimes N_{b \mid y}) \\
    & = \bra{\psi} M_{a \mid x} \otimes N_{b \mid y} \ket{\psi} ,
\end{split}
\end{equation}
where $\{M_{a \mid x}\}_{a}$ and $\{N_{b \mid y}\}_{b}$ are the physical, projective, local measurements of Alice and Bob for questions $x$ and $y$, respectively.
We now state a first definition of what it means to robustly self-test some reference state and measurements.
\begin{definition}[Self-testing of states and real measurements]
\label{def:real_self-testing}
    The probabilities $p(a, b \mid x, y)$ are said to $\delta$-approximately \emph{self-test} the state $\ket{\psi^{\prime}} \in \hilb{A}^{\prime} \otimes \hilb{B}^{\prime}$ and measurement operators $M_{a \mid x}^{\prime} \in \mathcal{L}(\hilb{A}^{\prime})$ and $N_{b \mid y}^{\prime} \in \mathcal{L}(\hilb{B}^{\prime})$ if, for any state $\ket{\psi} \in \hilb{A} \otimes \hilb{B}$ and measurement operators $M_{a \mid x} \in \mathcal{L}(\hilb{A})$ and $N_{b \mid y} \in \mathcal{L}(\hilb{B})$ from which these probabilities may arise, there exists a junk state $\ket{\xi} \in \tilde{\hilb{A}} \otimes \tilde{\hilb{B}}$ and isometries $V_{\hilb{A}} \colon \hilb{A} \to \hilb{A}^{\prime} \otimes \tilde{\hilb{A}}$ and $V_{\hilb{B}} \colon \hilb{B} \to \hilb{B}^{\prime} \otimes \tilde{\hilb{B}}$ defining the local isometry $V = V_{\hilb{A}} \otimes V_{\hilb{B}}$ such that for all $a$, $b$, $x$, and $y$
    \begin{subequations}
    \begin{align}
        V \ket{\psi}
        & \approxe{\delta} \ket{\psi^{\prime}} \otimes \ket{\xi} , \\
        V (M_{a \mid x} \otimes N_{b \mid y}) \ket{\psi}
        & \approxe{\delta} (M_{a \mid x}^{\prime} \otimes N_{b \mid y}^{\prime}) \ket{\psi^{\prime}} \otimes \ket{\xi} .
    \end{align}
    \end{subequations}
\end{definition}
This definition is standard, and accounts for the unobservable possibilities of local unitary basis transformations applied to the state and measurements, as well as embedding of the state and measurement operators in a Hilbert space of larger dimension, or the existence of additional degrees of freedom (on which the measurement operators do not act).
We may assume without loss of generality that the reference state $\ket{\psi^{\prime}}$ is \emph{real}, meaning that $\ket{\psi^{\prime}}^{*} = \ket{\psi^{\prime}}$, since the Schmidt decomposition guarantees the existence of local orthonormal bases in which all entries to its matrix are real.
Unless it is also assumed that $(M_{a \mid x}^{\prime})^{*} = M_{a \mid x}^{\prime}$ and $(N_{b \mid y}^{\prime})^{*} = N_{b \mid y}^{\prime}$ (that the reference measurements are real in this basis), \cref{def:real_self-testing} does not account for the unobservable possibility that Alice and Bob actually implement complex conjugated versions of the reference measurements in a correlated fashion \cite{mckague2011generalized,kaniewski2017self}.
This is because probabilities are real numbers, and so
\begin{equation}
\begin{split}
    p(a,b \mid x,y)
    & = \tr(\proj{\psi^{\prime}} M_{a \mid x}^{\prime} \otimes N_{b \mid y}^{\prime}) \\
    & = \tr(\proj{\psi^{\prime}} (M_{a \mid x}^{\prime})^{*} \otimes (N_{b \mid y}^{\prime})^{*}) \\
    & = p(a,b \mid x,y)^{*} ,
\end{split}
\end{equation}
but complex conjugation is not a unitary transformation.

It is sufficient for our purposes to consider complex conjugation performed in some convenient fixed local orthonormal bases for which $\ket{\psi^{\prime}}$ is real.
This is because complex conjugation of projectors performed in arbitrary local orthonormal bases is equivalent (up to some local unitary transformation) to conjugation performed in the original fixed bases (the complex conjugate of a unitary matrix is also unitary).
In particular, we cannot use \cref{def:real_self-testing} to self-test a reference state $\ket{\Phi^{+}}$ and (on one side) the observables $\sigma_{\mathrm{x}}$, $\sigma_{\mathrm{y}}$, and $\sigma_{\mathrm{z}}$; there is no local orthonormal bases in which the state and corresponding projectors are all real.
In cases with complex measurements, the following definition allowing also for correlated complex conjugation can be used.
\begin{definition}[Self-testing of states and complex measurements]
\label{def:complex_self-testing}
    The probabilities $p(a, b \mid x, y)$ are said to $\delta$-approximately \emph{self-test} the state $\ket{\psi^{\prime}} \in \hilb{A}^{\prime} \otimes \hilb{B}^{\prime}$ and measurement operators $M_{a \mid x}^{\prime} \in \mathcal{L}(\hilb{A}^{\prime})$ and $N_{b \mid y}^{\prime} \in \mathcal{L}(\hilb{B}^{\prime})$ if, for any state $\ket{\psi} \in \hilb{A} \otimes \hilb{B}$ and measurement operators $M_{a \mid x} \in \mathcal{L}(\hilb{A})$ and $N_{b \mid y} \in \mathcal{L}(\hilb{B})$ from which these probabilities may arise, there exists a junk state $\ket{\xi} \in \tilde{\hilb{A}} \otimes \hilb{A}^{\prime\prime} \otimes \tilde{\hilb{B}} \otimes \hilb{B}^{\prime\prime}$ and isometries $V_{\hilb{A}} \colon \hilb{A} \to \hilb{A}^{\prime} \otimes \tilde{\hilb{A}} \otimes \hilb{A}^{\prime\prime}$ and $V_{\hilb{B}} \colon \hilb{B} \to \hilb{B}^{\prime} \otimes \tilde{\hilb{B}} \otimes \hilb{B}^{\prime\prime}$ defining the local isometry $V = V_{\hilb{A}} \otimes V_{\hilb{B}}$ such that for all $a$, $b$, $x$, and $y$
    \begin{subequations}
    \begin{align}
        V \ket{\psi}
        & \approxe{\delta} \ket{\psi^{\prime}} \otimes \ket{\xi} , \\
        V (M_{a \mid x} \otimes N_{b \mid y}) \ket{\psi}
        & \approxe{\delta} (\bar{M}_{a \mid x} \otimes \bar{N}_{b \mid y}) \ket{\psi^{\prime}} \otimes \ket{\xi} ,
    \end{align}
    \end{subequations}
    where
    \begin{subequations}
    \begin{align}
        \bar{M}_{a \mid x} &= M_{a \mid x}^{\prime} \otimes \proj{0}_{\hilb{A}^{\prime\prime}} + (M_{a \mid x}^{\prime})^{*} \otimes \proj{1}_{\hilb{A}^{\prime\prime}} , \\
        \bar{N}_{b \mid y} &= N_{b \mid y}^{\prime} \otimes \proj{0}_{\hilb{B}^{\prime\prime}} + (N_{b \mid y}^{\prime})^{*} \otimes \proj{1}_{\hilb{B}^{\prime\prime}} ,
    \end{align}
    \end{subequations}
    and the state $\ket{\xi}$ has the form
    \begin{equation}
        \ket{\xi}
        = \ket{\xi_{0}} \otimes \ket{0}_{\hilb{A}^{\prime\prime}} \ket{0}_{\hilb{B}^{\prime\prime}}
        + \ket{\xi_{1}} \otimes \ket{1}_{\hilb{A}^{\prime\prime}} \ket{1}_{\hilb{B}^{\prime\prime}}
    \end{equation}
    for some subnormalized $\ket{\xi_{0}}$ and $\ket{\xi_{1}}$ in $\tilde{\hilb{A}} \otimes \tilde{\hilb{B}}$ satisfying $\braket{\xi_{0}}{\xi_{0}} + \braket{\xi_{1}}{\xi_{1}} = 1$.
\end{definition}

It is often sufficient to deduce a self-testing statement such that, instead of a full set of probabilities $p(a, b \mid x, y)$, one need only observe certain combinations of them given by the maximal violation of some Bell inequality $\mathcal{I}[p(a,b \mid x,y)] = \beta$.
One may replace the probabilities in \cref{def:real_self-testing,def:complex_self-testing} with such a maximal violation.
One may also choose to self-test measurement operators on only Alice's subsystem by simply taking Bob's measurement operators in these definitions to be identity operators.
Similarly, one need not choose to make a self-testing statement certifying reference operators used to produce all of the probabilities $p(a, b \mid x, y)$, provided that the result is shown to hold for all compatible sets of physical measurements.

Let $M = M_{+} - M_{-}$ be a $\pm 1$-outcome observable on Alice's subsystem with corresponding projectors $M_{\pm}$, and suppose we have statements of the form
\begin{subequations}
\label{eq:certify_state_obs}
\begin{align}
    V \ket{\psi}
    & \approxe{\delta} \ket{\psi^{\prime}} \otimes \ket{\xi} , \\
    V M \ket{\psi}
    & \approxe{\delta} \bar{M} \ket{\psi^{\prime}} \otimes \ket{\xi} ,
\end{align}
\end{subequations}
where $\bar{M} = M^{\prime} \otimes \proj{0} + (M^{\prime})^{*} \otimes \proj{1}$ for the reference observable $M^{\prime} = M_{+}^{\prime} - M_{-}^{\prime}$.
In this case, one automatically obtains statements in terms of the projectors of the form of \cref{def:complex_self-testing}
\begin{equation}
\label{eq:certify_projectors}
    V M_{\pm} \ket{\psi}
    \approxe{\delta} \bar{M}_{\pm} \ket{\psi^{\prime}} \otimes \ket{\xi} ,
\end{equation}
where $\bar{M}_{\pm} = M_{\pm}^{\prime} \otimes \proj{0} + (M_{\pm}^{\prime})^{*} \otimes \proj{1}$.
This follows from the linearity of $V$, along with the facts
\begin{equation}
    M_{\pm} = \frac{I \pm M}{2} ,\quad
    M_{\pm}^{\prime} = \frac{I \pm M^{\prime}}{2} ,\quad
    (M_{\pm}^{\prime})^{*} = \frac{I \pm (M^{\prime})^{*}}{2} .
\end{equation}

\subsection{Sum-of-squares (SOS) decomposition}
\label{sec:sos_decomposition}

A useful tool for proving robust self-testing statements from Bell inequalities is the sum-of-squares (SOS) decomposition \cite{yang2013robust,bamps2015sum}.
Suppose that a state $\ket{\psi}$ achieves the maximal quantum value $\beta$ of some Bell operator $O$ to within an amount $\varepsilon \geq 0$.
That is, $\bra{\psi} O \ket{\psi} \geq \beta - \varepsilon$.
Suppose also that we can write the shifted Bell operator $\beta - O$ in the form
\begin{equation}
    \beta - O
    = \sum_{j} F_{j}^{\dagger} F_{j}
\end{equation}
for some linear operators $F_{j}$.
Then
\begin{equation}
\begin{split}
    \varepsilon
    & \geq \bra{\psi} (\beta - O) \ket{\psi} \\
    & = \bra{\psi} \sum_{j} F_{j}^{\dagger} F_{j} \ket{\psi} \\
    & = \sum_{j} \lVert F_{j} \ket{\psi} \rVert^{2} .
\end{split}
\end{equation}
Therefore, for all $j$,
\begin{equation}
\label{eq:sos_bound}
    \lVert F_{j} \ket{\psi} \rVert
    \leq \sqrt{\varepsilon} .
\end{equation}
As we also see in \cref{sec:triple_chsh}, the $F_{j}$ may be found to have a form such that \cref{eq:sos_bound} gives useful relations from which a self-testing statement may ultimately be deduced.

\subsection{Regularization of operators}
\label{sec:regularization}

We may wish to evolve a state by an operation that acts in the same way as a unitary operator, but is not itself known to be unitary.
Given $\ket{\psi} \in \hilb{A} \otimes \hilb{B}$ and a Hermitian linear operator $T$ on $\hilb{B}$ such that $T \ket{\psi} \approxe{\varepsilon} U \ket{\psi}$ for some unitary Hermitian operator $U$ on $\hilb{A}$ and $\varepsilon \geq 0$, it is possible to define a new operator $\regular{T}$ on $\hilb{B}$ that is unitary and acts on $\ket{\psi}$ almost identically to $T$ and $U$ \cite{yang2013robust,bamps2015sum,supic2020self}.
Note that the operators $U$ and $T$ commute on $\hilb{A} \otimes \hilb{B}$ since the former is defined to act nontrivially only on $\hilb{A}$, while similarly the latter acts nontrivially only on $\hilb{B}$.

This \emph{regularization} of $T$ is performed in two steps.
One first removes all zero eigenvalues from $T$ by defining a new operator $\tilde{T} = T + P$, where $P$ is the orthogonal projection onto $\ker{T}$ (although we sometimes explicitly write its adjoint $\tilde{T}^{\dagger}$, it should be noted that this operator is Hermitian and $\tilde{T}^{\dagger} = \tilde{T}$, since both $T$ and $P$ are Hermitian).
All nonzero vectors mapped to $\bm 0$ under $T$ remain unchanged under $\tilde{T}$ and so have eigenvalue $1$ instead.
We then have that $\tilde{T}^{\dagger} \tilde{T}$ is positive definite: that is $\bra{v} \tilde{T}^{\dagger} \tilde{T} \ket{v} = \lVert \tilde{T} \ket{v} \rVert^{2} > 0$ for all nonzero vectors $\ket{v}$ since $\tilde{T} \ket{v} \neq \bm 0$ by construction.
Thus, its principal square root $\lvert \tilde{T} \rvert = (\tilde{T}^{\dagger} \tilde{T})^{1/2}$ is also positive definite.
Since this $\lvert \tilde{T} \rvert$ is positive definite, it is therefore also invertible.

In the second step, one defines
\begin{equation}
    \regular{T} = \tilde{T} \lvert \tilde{T} \rvert^{-1} .
\end{equation}
The regularized operator $\regular{T}$ is unitary by construction.
Furthermore, by considering an eigenbasis of $T$, we see that $\regular{T} T = \lvert T \rvert$.
We also have that $\lvert T \rvert = \lvert U^{\dagger} T \rvert$ (since $U^{\dagger}$ is unitary), and the property that $\lvert U^{\dagger} T \rvert \geq U^{\dagger} T$.
This operator inequality is valid since the operator $U^{\dagger} T$ is Hermitian ($U$ commutes with $T$ and both are Hermitian by assumption).
Therefore,
\begin{equation}
\begin{split}
    \lVert \regular{T} \ket{\psi} - T \ket{\psi} \rVert
    & = \lVert \ket{\psi} - \regular{T} T \ket{\psi} \rVert \\
    & = \lVert \ket{\psi} - \lvert T \rvert \ket{\psi} \rVert \\
    & = \lVert \ket{\psi} - \lvert U^{\dagger} T \rvert \ket{\psi} \rVert \\
    & \leq \lVert \ket{\psi} - U^{\dagger} T \ket{\psi} \rVert \\
    & = \lVert U \ket{\psi} - T \ket{\psi} \rVert \\
    & \leq \varepsilon .
\end{split}
\end{equation}
In other words, $\regular{T} \ket{\psi} \approxe{\varepsilon} T \ket{\psi}$ and $\regular{T} \ket{\psi} \approxe{2 \varepsilon} U \ket{\psi}$.

Given a state-dependent anticommutation relation between two Hermitian operators, we may wish to make a similar statement about their regularized versions.

\begin{lemma}
\label{lem:regularized_acomm}
    Let $\ket{\psi} \in \hilb{A} \otimes \hilb{B}$.
    Suppose that $T_{1}$ and $T_{2}$ are Hermitian operators on $\hilb{B}$ such that $T_{1} \ket{\psi} \approxe{\varepsilon} U_{1} \ket{\psi}$ and $T_{2} \ket{\psi} \approxe{\varepsilon} U_{2} \ket{\psi}$ for some unitary Hermitian operators $U_{1}$ and $U_{2}$ on $\hilb{A}$ and $\varepsilon \geq 0$.
    Then the regularized operators $\regular{T}_{1}$ and $\regular{T}_{2}$ defined on $\hilb{B}$ satisfy
    \begin{equation}
        \{ \regular{T}_{1}, \regular{T}_{2} \} \ket{\psi} \approxe{c} \{ T_{1}, T_{2} \} \ket{\psi} ,
    \end{equation}
    where $c = (6 + \lVert T_{1} \rVert + \lVert T_{2} \rVert) \varepsilon$.
\end{lemma}
\begin{proof}
    As discussed earlier, the regularized operators satisfy $\regular{T}_{j} \ket{\psi} \approxe{2 \varepsilon} U_{j} \ket{\psi}$ and $\regular{T}_{j} \ket{\psi} \approxe{\varepsilon} T_{j} \ket{\psi}$.
    Thus,
    \begin{equation}
    \begin{split}
        \{ \regular{T}_{1}, \regular{T}_{2} \} \ket{\psi}
        & \approxe{4 \varepsilon} (\regular{T}_{1} U_{2} + \regular{T}_{2} U_{1}) \ket{\psi} \\
        & \approxe{2 \varepsilon} (T_{1} U_{2} + T_{2} U_{1}) \ket{\psi} ,
    \end{split}
    \end{equation}
    where we have used the unitarity of $\regular{T}_{1}$ and $\regular{T}_{2}$ for the first estimate, and the unitarity of $U_{1}$ and $U_{2}$ for the second estimate.
    In obtaining the second estimate, we used that both $\regular{T}_{1}$ and $T_{1}$ commute with $U_{2}$ and that both $\regular{T}_{2}$ and $T_{2}$ commute with $U_{1}$ on the Hilbert space $\hilb{A} \otimes \hilb{B}$.
    This is again due to each of the operators being defined to act nontrivially on only one of either $\hilb{A}$ or $\hilb{B}$.
    We also have that
    \begin{equation}
    \begin{split}
        & \lVert (T_{1} U_{2} + T_{2} U_{1}) \ket{\psi} - \{ T_{1}, T_{2} \} \ket{\psi} \rVert \\
        & \leq \lVert T_{1} (U_{2} - T_{2}) \ket{\psi} \rVert + \lVert T_{2} (U_{1} - T_{1}) \ket{\psi} \rVert \\
        & \leq (\lVert T_{1} \rVert + \lVert T_{2} \rVert) \varepsilon ,
    \end{split}
    \end{equation}
    and so the result follows.
\end{proof}

\subsection{Verifiable blind quantum computation}
\label{sec:vbqc}

The VBQC protocol introduced by \textcite{fitzsimons2017unconditionally} is built upon the measurement-based quantum computation (MBQC) model.
It enables a client to delegate a quantum computation to a remote server with the security properties of blindness and verifiability.

Following the formalism proposed by \textcite{danos2007measurement}, a single-party computation in MBQC is a sequence of entangling gates and single-qubit measurements in bases determined by previous measurement outcomes.
It is comprised of three ingredients:
\begin{enumerate}
    \item A \emph{graph state} prepared by applying controlled-$Z$ gates to pairs of $\ket{+}$ states according to the edges an undirected graph whose vertices correspond to the qubits.
    \item Values of \emph{measurement angles} $\phi_{i} \in A = \{ \frac{k \pi}{4} \}_{k=0}^{7}$ for all qubits $i$.
    \item A notion of order in which measurements are performed and their dependence on previous measurement outcomes, called the \emph{flow}.
\end{enumerate}
Measurements of the qubits are performed in bases $\ket{\pm_{\phi_{i}^{\prime}}}$, where \emph{corrected} measurement angles $\phi_{i}^{\prime} \in A$ depend in a fixed way on the $\phi_{i}$ and outcomes $s_{j} \in \{0,1\}$ of previous measurements according to the flow.

Imagine a delegated scenario in which Alice initially sends the unentangled $\ket{+}$ qubits to Bob.
Alice then interactively instructs Bob to perform the desired measurement pattern, requesting corrected measurement angles $\phi_{i}^{\prime}$ be performed based on Bob's previous output bits, which he communicates back to Alice at each step in the computation.

Blindness with respect to the angles $\phi_{i}$ and outcomes $s_{j}$ can be achieved by Alice initially sending qubits $\ket{+_{\theta_{i}}}$ to Bob instead, where the $\theta_{i} \in A$ are randomly chosen by Alice and kept secret.
Measuring $\ket{+_{\theta_{i}}}$ using angle $\phi_{i}^{\prime} + \theta_{i}$ is the same as measuring $\ket{+}$ using angle $\phi_{i}^{\prime}$.
Alice then requests that Bob measure in bases with angles $\delta_{i} = \phi_{i}^{\prime} + \theta_{i} + r_{i} \pi$, where the $r_{i} \in \{0,1\}$ are randomly chosen by Alice and kept secret.
The $r_{i}$ do not change the measurement aside from swapping its possible outcomes and, after receiving a measurement output bit $b_{i} \in \{0,1\}$ from Bob, Alice can recover the true outcome as $s_{i} = b_{i} \oplus r_{i}$.
Overall, the secret $\theta_{i}$ introduced hide the true measurement angles from Bob, while the $r_{i}$ hide the true measurement outcomes \cite{broadbent2009universal,barz2012demonstration}.

Verifiability (that the server performs the computation correctly) is achieved by further introducing ``trap'' patterns.
These are so-called because, while their position in the graph state cannot be deduced by Bob, the outcome of an honest measurement at the trap is deterministic and known in advance to Alice.
If Bob incorrectly measures during the computation, he may incorrectly measure at a trap, and such a deviation would be detected by Alice.
To construct a trap pattern, we additionally bestow upon Alice the ability to send computational basis qubits to Bob as part of the initial state.
In this context, these qubits are known as ``dummy'' qubits.
This is because after Bob applies entangling controlled-$Z$ gates, any $\ket{0}$ or $\ket{1}$ qubits involved remain unchanged and disentangled from the rest of the graph state.
A trap pattern is formed by surrounding a $\ket{+_{\theta}}$ qubit with dummy qubits such that all of its neighbors in the graph are dummies.
After entanglement gates have been applied by Bob, this \emph{trap qubit} thus remains entirely disentangled from the rest of the graph state.
Moreover, it is either $\ket{+_{\theta}}$ or $\ket{-_{\theta}}$ depending on whether the number of its neighbors that are $\ket{1}$ is even or odd, respectively.
This parity, and thus which of the two states the trap qubit should take, is known to Alice.
Therefore, by Alice requesting that the trap qubit be measured in the $\ket{\pm_{\theta}}$ basis, she can detect any deviation made by Bob at the trap qubit \cite{fitzsimons2017unconditionally,barz2013experimental,kashefi2017optimised}.

\section{Efficient parallel self-testing for DIVBQC}
\label{sec:efficient_testing}

Standard VBQC protocols (such as the original Fitzsimons--Kashefi protocol \cite{fitzsimons2017unconditionally} we consider) require that $n$ qubits be prepared on the server side, each in one of the states $\ket{\pm_{\theta}}$, or the states $\ket{0}$ or $\ket{1}$ (for dummy qubits).
A self-test appropriate for device-independently preparing states for practical VBQC should satisfy the following properties.
\begin{enumerate}
    \item \label[property]{itm:bell_state_product}
    The self-test should certify the presence of $n$ copies of the Bell state $\ket{\Phi^{+}}$ shared between Alice and Bob.
    That is, the state $\ket{\Phi^{+}}^{\otimes n}$ should be self-tested.

    \item \label[property]{itm:pauli_product}
    Denote by $\ket{\sigma_{\chi}^{\lambda}}$ the eigenstate of $\sigma_{\chi}$ with eigenvalue $\lambda \in \{+1, -1\}$ and $\chi \in \{\mathrm{x}, \mathrm{y}, \mathrm{z}, \mathrm{x} + \mathrm{y}, \mathrm{x} - \mathrm{y}\}$ (see \cref{sec:prelims}).
    For $n$-tuples of Pauli bases $\bm \chi = (\chi_{1}, \dots, \chi_{n}) \in \{\mathrm{x}, \mathrm{y}, \mathrm{z}, \mathrm{x} + \mathrm{y}, \mathrm{x} - \mathrm{y}\}^{n}$, the self-test should certify the projective measurement of the certified state $\ket{\Phi^{+}}^{\otimes n}$ on Alice's side
    \begin{equation}
        \Bigg\{ \bigotimes_{j=1}^{n} \proj*{\sigma_{\chi_{j}}^{\lambda_{j}}} \Bigg\}_{\bm \lambda \in \{+1, -1\}^{n}}
    \end{equation}
    with each measurement operator composed of the tensor product of $n$ projectors.
    We refer to these as \emph{special} or \emph{preparation} measurements for reasons that will become clear.

    \item \label[property]{itm:question_size}
    The number of possible input questions, and hence the number of different measurements required, should be small.
    Specifically, questions should be of size at most logarithmic in $n$.

    \item \label[property]{itm:robust}
    The self-test should be robust to the observation of nonideal statistics.

    \item \label[property]{itm:client_measurement}
    All measurements performed on Alice's subsystem (the client side) must be local to one of her qubits in the honest case.

    \item \label[property]{itm:data_processing}
    Classical processing of the gathered experimental outcomes should scale efficiently in $n$.

    \item \label[property]{itm:conjugation}
    Possible complex conjugation of measurement operators should occur only in the preparation of dummy qubits $\ket{0}$ and $\ket{1}$ (i.e., measurement of Pauli $\sigma_{\mathrm{z}}$).

    \item \label[property]{itm:independent_isom} The isometry on Bob's subsystem and reduced junk state guaranteed by the self-test should be independent of the choice of preparation question $\bm \chi$.
\end{enumerate}

Together, \cref{itm:bell_state_product,itm:pauli_product} allow Alice to remotely prepare $n$ qubit states on Bob's subsystem (up to the freedoms allowed by self-testing), each in one of the two states comprising a prechosen basis for the purpose of VBQC (any of the bases corresponding to observables $\sigma_{\mathrm{x}}$, $\sigma_{\mathrm{y}}$, $\sigma_{\mathrm{z}}$, $\sigma_{\mathrm{x} + \mathrm{y}}$, or $\sigma_{\mathrm{x} - \mathrm{y}}$).
The outcome Alice receives allows her to determine the state that has been prepared on Bob's subsystem.
Meanwhile, Bob is aware of neither the state that was prepared on his side, nor the specific bases chosen.

\Cref{itm:question_size} means that it is experimentally feasible to gather outcome statistics upon all possible questions.
If, instead, questions had size linear in $n$, then the time required to achieve good statistical confidence would scale exponentially with the number of qubits to be prepared.
At first glance, \cref{itm:pauli_product,itm:question_size} may appear slightly contradictory, since the number of possible $n$-tuples of bases is exponential in $n$.
We avoid this issue by requiring that only a single \emph{special} one of these exponentially many measurements be among those self-tested.
We still require, however, that this special measurement may be freely defined using any of the possible basis tuples, prior to the initiation of the self-testing protocol.
The reason for this is that the special measurement will later be used by Alice to prepare the $n$ server-side qubits in the suitable bases to perform an arbitrary computation (that is prechosen by Alice), for which the ability to prepare each qubit in an arbitrary basis is desirable.
In this sense, we must in fact define a whole class of self-testing protocols---one for each choice of special measurement---all of which satisfy the other properties.

The robustness included as \cref{itm:robust} is important to allow for experimental noise and imperfections.
The connection between the level of robustness guaranteed by a self-testing statement and that which can be achieved in remote state preparation is examined in more detail in \cref{sec:post-measurement_states}.

\Cref{itm:client_measurement} is included to ensure that it is sufficient for the client to possess only a simple quantum device.
Such a device may only be capable of performing single-qubit measurements and, thus, we must ensure that all measurements included as part of the honest strategy for the self-test (even those that are not used for preparation) can be performed by the client.

\Cref{itm:data_processing} means that the experimental data that is collected may be combined and processed such that the conditions of the self-test are shown to be true in a reasonable time.
This may not necessarily be the case if, for example, the Bell expressions which must be evaluated in order to perform the certification have exponentially many terms in $n$ \cite{supic2021device}.

\Cref{itm:conjugation} preserves the verifiability property of the FK protocol while still allowing states to be prepared in all required bases.
This is discussed further in \cref{sec:conjugation} (see also \cite{gheorghiu2015robustness}).

\Cref{itm:independent_isom} is a slight restriction on isometries found in self-testing protocols, allowing the physical states prepared using an untrusted strategy to be used as resource states for VBQC.
In the FK protocol, the possibly deviating server of Bob is allowed to apply any unitary (or, more generally, quantum channel) to the reference qubits that are received, but has no knowledge of the bases in which these qubits have been prepared \cite{fitzsimons2017unconditionally}.
We would thus like to interpret the physical state prepared on Bob's side due to a self-testing protocol as originating solely from such a process \cite{gheorghiu2015robustness}.
This can be done be ``undoing'' Bob's self-testing isometry, provided that his junk state (ancillary to the reference qubits) and isometry do not depend on any information about the choice of bases (i.e., the preparation question chosen).

\subsection{Post-measurement states}
\label{sec:post-measurement_states}

Robust parallel self-tests for $n$-qubit states and measurements typically guarantee estimates for observables acting on the states of the form
\begin{equation}
    V M^{\bm s} \ket{\psi} \approxe{\delta} {M^{\prime}}^{\bm s} \ket{\psi^{\prime}} \otimes \ket{\xi}
\end{equation}
for all $\bm s \in \{0, 1\}^{n}$ (or in expectation over all $\bm s$ in some cases \cite{natarajan2017quantum,natarajan2018low}).
This naively leads to statements for individual outcomes $\bm a \in \{+, -\}^{n}$ of the form
\begin{equation}
\label{eq:naive_proj_estimate}
    V M_{\bm a} \ket{\psi} \approxe{\delta} M_{\bm a}^{\prime} \ket{\psi^{\prime}} \otimes \ket{\xi} ,
\end{equation}
where the $M_{\bm a}$ and $M_{\bm a}^{\prime}$ are projection operators (cf. \cref{eq:certify_state_obs,eq:certify_projectors} and \cref{def:real_self-testing,def:complex_self-testing}).
For most practical applications (including the present one), however, one is instead interested in characterizing the physical and reference post-measurement states
\begin{equation}
    \frac{M_{\bm a} \ket{\psi}}{\sqrt{\bra{\psi} M_{\bm a} \ket{\psi}}} ,\quad
    \frac{M_{\bm a}^{\prime} \ket{\psi^{\prime}}}{\sqrt{\bra{\psi} M_{\bm a}^{\prime} \ket{\psi^{\prime}}}} .
\end{equation}
Since there are $2^{n}$ possible outcomes $\bm a$ (and assuming they occur approximately uniformly), their respective probabilities $\bra{\psi} M_{\bm a} \ket{\psi}$ must all be approximately $2^{-n}$.
Therefore, self-testing guarantees of the form of \cref{eq:naive_proj_estimate} typically lead to vector norm distances $\delta \sqrt{2^{n}}$ of physical post-measurement states from their reference counterparts, which blow up exponentially in $n$.

In the following, we overcome this impracticality by slightly relaxing the distance guaranteed by a factor polynomial in $\delta$, and also allowing that this guarantee fails to be satisfied with some probability that is polynomially small in $\delta$.
In the noiseless honest case, $\delta$ is written in terms of some error $\varepsilon$ that (as is typical of the statistical tools used to analyze data for self-testing protocols \cite{lin2018device}) can be experimentally saturated with statistical confidence exponentially close to unity in the number of experimental trials performed.

We first exhibit in \cref{lem:general_robust_prob} a general result on the trace distance between pure states that is compatible with distances that are expressed in expectation (as is the case for some self-testing statements).
We then use a special case of this to show in \cref{thm:robust_prob} that, given certain self-testing guarantees of the form that will be derived in our context later (see \cref{thm:state_preparation,thm:protocol_isometry,cor:protocol_isometry}), the post-measurement states of the physical experiment must be close in trace distance to those of the reference experiment most of the time.

\begin{lemma}
\label{lem:general_robust_prob}
    Let $\Sigma \times \Omega$ be some finite sample space.
    Let $\pi$ be some probability mass function on $\Sigma$ for the random variable $S \colon \Sigma \times \Omega \to \Sigma$ defined by $S(\sigma, \omega) = \sigma$.
    For all $\sigma \in \Sigma$ and $\omega \in \Omega$, let $\ket{u_{\sigma}^{\omega}}$ be vectors satisfying
    \begin{equation}
        \sum_{\omega \in \Omega} \lVert \ket{u_{\sigma}^{\omega}} \rVert^{2} = 1 .
    \end{equation}
    For each $\sigma \in \Sigma$, let $p_{\sigma}$ be the function on $\Omega$ defined by
    \begin{equation}
        p_{\sigma}(\omega) = \lVert \ket{u_{\sigma}^{\omega}} \rVert^{2} .
    \end{equation}
    Define the probability mass function $p$ on $\Sigma \times \Omega$ by
    \begin{equation}
        p(\sigma, \omega) = p_{\sigma}(\omega) \pi(\sigma) .
    \end{equation}
    For each $\sigma \in \Sigma$ and $\omega \in \Omega$ satisfying $p(\sigma, \omega) > 0$, let $\ket{v_{\sigma}^{\omega}}$ be nonzero vectors.
    Denote normalized versions of all the vectors (when they are defined) by
    \begin{equation}
        \ket{\hat{u}_{\sigma}^{\omega}} = \frac{\ket{u_{\sigma}^{\omega}}}{\lVert \ket{u_{\sigma}^{\omega}} \rVert} ,\quad
        \ket{\hat{v}_{\sigma}^{\omega}} = \frac{\ket{v_{\sigma}^{\omega}}}{\lVert \ket{v_{\sigma}^{\omega}} \rVert} .
    \end{equation}
    Let $D$ be a random variable on $\Sigma \times \Omega$ defined by the trace distance between these normalized states
    \begin{equation}
        D(\sigma, \omega) = 
        \begin{cases}
            \frac{1}{2} \lVert \proj{\hat{u}_{\sigma}^{\omega}} - \proj{\hat{v}_{\sigma}^{\omega}} \rVert_{1} & \text{if $p(\sigma, \omega) > 0$,} \\
            0 & \text{if $p(\sigma, \omega) = 0$.}
        \end{cases}
    \end{equation}
    Suppose that for some $\delta \geq 0$ we have
    \begin{equation}
    \label{eq:norm_squares_sum_bound}
        \sum_{\sigma \in \Sigma} \pi(\sigma) \sum_{\omega \in \Omega} \lVert \ket{u_{\sigma}^{\omega}} - \ket{v_{\sigma}^{\omega}} \rVert^{2} \leq \delta^{2} .
    \end{equation}
    Then, for any $c > 0$,
    \begin{equation}
        \Pr(D \leq \delta^{c}) \geq 1 - 4 \delta^{2(1-c)}
    \end{equation}
    for all $\sigma \in \Sigma$ for which $\pi(\sigma) > 0$.
\end{lemma}
\begin{proof}
    See \cref{sec:robust_prob}.
\end{proof}

The proof of \cref{lem:general_robust_prob} relies on the following elementary result, which gives a useful bound on the trace distance between operators of the form $\proj{v}$, where vectors $\ket{v}$ may be subnormalized.

\begin{lemma}
\label{lem:trace_dist_bound}
    Let vectors $\ket{u}$ and $\ket{v}$ belonging to the same Hilbert space satisfy $\lVert \ket{u} \rVert \leq 1$ and $\lVert \ket{v} \rVert \leq 1$.
    The trace distance is then bounded as
    \begin{equation}
        \frac{1}{2} \lVert \proj{u} - \proj{v} \rVert_{1} \leq 2 \lVert \ket{u} - \ket{v} \rVert .
    \end{equation}
    If the vectors have unit length then the bound can be tightened to
    \begin{equation}
        \frac{1}{2} \lVert \proj{u} - \proj{v} \rVert_{1} \leq \lVert \ket{u} - \ket{v} \rVert .
    \end{equation}
\end{lemma}
\begin{proof}
    See \cref{sec:estimation_lemmas}.
    For unit vectors the result immediately follows from the inequality
    \begin{equation}
        \lvert \braket{u}{v} \rvert \geq \Re{\braket{u}{v}} = 1 - \frac{1}{2} \lVert \ket{u} - \ket{v} \rVert^{2}
    \end{equation}
    applied to the Fuchs--van de Graaf expression
    \begin{equation}
        \frac{1}{2} \lVert \proj{u} - \proj{v} \rVert_{1} = \sqrt{1 - \lvert \braket{u}{v} \rvert^{2}}
    \end{equation}
    for pure states.
\end{proof}

We now proceed to use a special case of \cref{lem:general_robust_prob} to state a similar result for the context of self-testing with robustness guarantees given for all possible observables that can be formed from measurement operators for parallel binary outcomes.

\begin{theorem}
\label{thm:robust_prob}
    Let $V \colon \hilb{H} \to \hilb{H}^{\prime}$ be an isometry.
    Define the probability mass function $p$ for outcomes $\bm a \in \{0, 1\}^{n}$ of a projective measurement $\{M_{\bm a}\}_{\bm a} \subset \mathcal{L}(\hilb{H})$ on a state $\ket{\psi} \in \hilb{H}$ by $p(\bm a) = \bra{\psi} M_{\bm a} \ket{\psi}$.
    Let $\ket{\psi^{\prime}} \in \hilb{H}^{\prime}$ be a state and let $\{ M_{\bm a}^{\prime} \}_{\bm a} \subset \mathcal{L}(\hilb{H}^{\prime})$ be a projective measurement such that $M_{\bm a}^{\prime} \ket{\psi^{\prime}} \neq 0$ for all $\bm a$.
    For all $\bm s \in \{0, 1\}^{n}$, define observables
    \begin{equation}
    \label{eq:obs_transform}
        M^{\bm s} = \sum_{\bm a} (-1)^{\bm a \cdot \bm s} M_{\bm a} ,\quad
        {M^{\prime}}^{\bm s} = \sum_{\bm a} (-1)^{\bm a \cdot \bm s} M_{\bm a}^{\prime} .
    \end{equation}
    Suppose that for all $\bm s$ we have
    \begin{equation}
    \label{eq:all_obs_robust}
        V M^{\bm s} \ket{\psi} \approxe{\delta} {M^{\prime}}^{\bm s} \ket{\psi^{\prime}} .
    \end{equation}
    Defining a random variable $D$ for the trace distance between post-measurement states
    \begin{equation}
        D(\bm a)
        = \frac{1}{2} \mathopen{}\left\lVert \frac{V M_{\bm a} \proj{\psi} M_{\bm a} V^{\dagger}}{p(\bm a)}
        - \frac{M_{\bm a}^{\prime} \proj{\psi^{\prime}} M_{\bm a}^{\prime}}{\bra{\psi^{\prime}} M_{\bm a}^{\prime} \ket{\psi^{\prime}}} \right\rVert_{1}
    \end{equation}
    we then have (with respect to the probability distribution $p$) that
    \begin{equation}
        \Pr \mathopen{}\left( D \leq \delta^{2/3} \right) \geq 1 - 4 \delta^{2/3} .
    \end{equation}
\end{theorem}
\begin{remark}
    We assume without loss of generality that all outcomes satisfy $p(\bm a) > 0$, since values of $D$ that could never be observed would not contribute towards the resulting probability.
\end{remark}
\begin{proof}
    Let us denote
    \begin{equation}
        \ket{w_{\bm a}} = V M_{\bm a} \ket{\psi} - M_{\bm a}^{\prime} \ket{\psi^{\prime}} .
    \end{equation}
    Combining \cref{eq:obs_transform,eq:all_obs_robust}, we can write
    \begin{equation}
    \begin{split}
        \delta^{2}
        & \geq \frac{1}{2^{n}} \sum_{\bm s} \mathopen{}\left\lVert V M^{\bm s} \ket{\psi} - {M^{\prime}}^{\bm s} \ket{\psi^{\prime}} \right\rVert^{2} \\
        & = \frac{1}{2^{n}} \sum_{\bm s} \mathopen{}\left\lVert \sum_{\bm a} (-1)^{\bm a \cdot \bm s} \ket{w_{\bm a}} \right\rVert^{2} \\
        & = \frac{1}{2^{n}} \sum_{\bm a, \bm b, \bm s} (-1)^{(\bm a \oplus \bm b) \cdot \bm s} \braket{w_{\bm b}}{w_{\bm a}} \\
        & = \sum_{\bm a} \lVert V M_{\bm a} \ket{\psi} - M_{\bm a}^{\prime} \ket{\psi^{\prime}} \rVert^{2} ,
    \end{split}
    \end{equation}
    where we have used the fact that
    \begin{equation}
        \sum_{\bm s} (-1)^{\bm c \cdot \bm s} =
        \begin{cases}
            2^{n} & \text{if $\bm c = \bm 0$,} \\
            0 & \text{otherwise.}
        \end{cases}
    \end{equation}
    Note that $\lVert V M_{\bm a} \ket{\psi} \rVert^{2} = \bra{\psi} M_{\bm a} \ket{\psi}$ and $\lVert M_{\bm a}^{\prime} \ket{\psi^{\prime}} \rVert^{2} = \bra{\psi^{\prime}} M_{\bm a}^{\prime} \ket{\psi^{\prime}}$.
    Finally, apply \cref{lem:general_robust_prob} with sample spaces $\Sigma = \{ \sigma \}$ for some $\sigma$ (so that $\pi(\sigma) = 1$) and $\Omega = \{0, 1\}^{n}$, vectors $\ket{u_{\bm a}} = V M_{\bm a} \ket{\psi}$ and $\ket{v_{\bm a}} = M_{\bm a}^{\prime} \ket{\psi^{\prime}}$, and choosing $c = 2/3$.
\end{proof}

\subsection{Correlated complex conjugation}
\label{sec:conjugation}

As we have already seen in \cref{sec:self-testing}, it is impossible to distinguish a reference strategy from that with its measurement operators replaced by their complex conjugates (performed in some fixed local orthonormal bases for which the reference state may be assumed to have real matrix elements).
Since there is no basis in which all of $\sigma_{\mathrm{x}}$, $\sigma_{\mathrm{y}}$, and $\sigma_{\mathrm{z}}$ simultaneously have real matrix representations, some complex conjugation ambiguity must necessarily be included in our self-testing statement in order that we certify measurements of all $\sigma_{\mathrm{x}}$, $\sigma_{\mathrm{y}}$, $\sigma_{\mathrm{z}}$, $\sigma_{\mathrm{x} + \mathrm{y}}$, and $\sigma_{\mathrm{x} - \mathrm{y}}$ on Alice's side, as required to satisfy \cref{itm:pauli_product}.

Note that there exist a pair of local orthonormal bases for which the matrix representations of the state $\ket{\Phi^{+}}$ and observables $\sigma_{\mathrm{x}}$ and $\sigma_{\mathrm{y}}$ are all real, while $\sigma_{\mathrm{z}}^{*} = - \sigma_{\mathrm{z}}$ (with $*$ denoting complex conjugation performed in the aforementioned basis).
The observable $- \sigma_{\mathrm{z}}$ has eigenstate $\ket{1}$ for the eigenvalue $+1$ and eigenstate $\ket{0}$ for the eigenvalue $-1$.
With this choice of complex conjugation basis, the only ambiguity introduced in the preparation of qubits on Bob's side is that a qubit supposedly prepared by the measurement of $\sigma_{\mathrm{z}}$ in a state $\ket{0}$ or $\ket{1}$ may in fact be prepared in the opposite state $\ket{1}$ or $\ket{0}$, respectively.
States prepared by measurements of any of the other observables $\sigma_{\mathrm{x}}$, $\sigma_{\mathrm{y}}$, $\sigma_{\mathrm{x} + \mathrm{y}}$, or $\sigma_{\mathrm{x} - \mathrm{y}}$ remain unambiguous.

The FK protocol correctly handles input states in which all of the ``dummy'' qubits of the honest input $\ket{0}$ or $\ket{1}$ are unknowingly (with some unknown probability) flipped to $\ket{1}$ and $\ket{0}$, respectively.
Hence, for the remainder of this work, we take complex conjugation to be performed in the local orthonormal bases described in this section, unless otherwise stated.

\begin{proposition}[{\cite[Lemma~10]{gheorghiu2015robustness}}]
    If the initial input state of the FK protocol is close to the ideal input state with all dummy qubits $\ket{0}$ and $\ket{1}$ replaced with $\ket{1}$ and $\ket{0}$, respectively, the protocol will reject it with high probability.
\end{proposition}
\begin{proof}[Proof sketch]
    Given a trap that has an odd number of dummy qubit neighbors, the verifier expects to apply a $Z$ correction based on the parity of the number of $\ket{1}$ neighbors.
    With all $\ket{0}$ and $\ket{1}$ qubits of the input flipped with respect to the ideal input, the number of $\ket{1}$ neighbors of the trap has opposite parity to what the verifier expects.
    In this case, the verifier will always get the opposite result from the trap to that which is expected.
    Therefore, as long as the verifier makes sure that at least one trap has an odd number of dummy neighbors (which is easily achievable), the state is rejected in the protocol.
\end{proof}

We require that possible flipping of the dummy qubits in the input state occurs \emph{globally}: either all such states are flipped or none are flipped.
Note that this is in correspondence with \cref{def:complex_self-testing} of complex self-testing, in which complex conjugation is possibly performed on the whole reference measurement operator without any mention of its structure (in our case the many-qubit tensor product structure).
One may otherwise imagine a weaker statement of self-testing for the special case of $n$-fold product states, in which the reference experiment is certified up to complex conjugation at any combination of the $n$ positions.
It is possible to construct tests that enforce global complex conjugation from some such statements \cite{bowles2018self}.

In previous works lifting the FK protocol to the device-independent scenario, that complex conjugation must be accounted for in a global fashion was not a consideration, since state preparation was performed sequentially (self-testing single EPR pairs as in \cite{hajdusek2015device} or based on the rigidity of the CHSH game \cite{reichardt2013classical2,reichardt2013classical1} as in \cite{gheorghiu2015robustness}).

\section{Triple CHSH inequality}
\label{sec:triple_chsh}

Let us first consider the well-known problem of self-testing a single Bell state and single-qubit Pauli observables in the following scenario.
In each round, Alice is provided with one of three possible input questions $x \in \{1, 2, 3\}$ and answers with $a \in \{+1, -1\}$.
These are denoted by the $\pm 1$-outcome observables acting on Alice's subsystem $A_{1}$, $A_{2}$, and $A_{3}$ respectively.
Meanwhile, Bob is provided with one of six possible input questions $y \in \{1, \dots, 6\}$ and answers with $b \in \{+1, -1\}$.
These are denoted by the $\pm 1$-outcome observables acting on Bob's subsystem $D_{\mathrmpair{z}{x}}$, $E_{\mathrmpair{z}{x}}$, $D_{\mathrmpair{z}{y}}$, $E_{\mathrmpair{z}{y}}$, $D_{\mathrmpair{x}{y}}$, and $E_{\mathrmpair{x}{y}}$, respectively.

Consider the triple CHSH operator \cite{acin2016optimal,bowles2018self} defined as
\begin{equation}
\label{eq:triple_chsh}
\begin{split}
    C
    ={} & A_{3} \otimes (D_{\mathrmpair{z}{x}} + E_{\mathrmpair{z}{x}}) + A_{1} \otimes (D_{\mathrmpair{z}{x}} - E_{\mathrmpair{z}{x}}) \\
    & + A_{3} \otimes (D_{\mathrmpair{z}{y}} + E_{\mathrmpair{z}{y}}) + A_{2} \otimes (D_{\mathrmpair{z}{y}} - E_{\mathrmpair{z}{y}}) \\
    & + A_{1} \otimes (D_{\mathrmpair{x}{y}} + E_{\mathrmpair{x}{y}}) + A_{2} \otimes (D_{\mathrmpair{x}{y}} - E_{\mathrmpair{x}{y}}) .
\end{split}
\end{equation}
This operator is the sum of three CHSH operators, with each of $A_{1}$, $A_{2}$, and $A_{3}$ contained in two of them.
The expectation value satisfies $\bra{\psi} C \ket{\psi} \leq 6\sqrt{2}$ for any state $\ket{\psi}$ shared between Alice and Bob, since each of the three CHSH operators has expectation upper bounded by 2$\sqrt{2}$.
We can saturate this bound by taking the shared state to be
\begin{equation}
\label{eq:triple_chsh_honest_state}
    \ket{\psi} = \ket{\Phi^{+}} \equiv \frac{\ket{00} + \ket{11}}{\sqrt{2}}
\end{equation}
and the observables to be
\begin{equation}
\label{eq:triple_chsh_honest_observables}
\begin{gathered}
    A_{1} = \sigma_{\mathrm{x}} ,\quad
    A_{2} = - \sigma_{\mathrm{y}} ,\quad
    A_{3} = \sigma_{\mathrm{z}} ,\quad \\
    D_{j,k} = \frac{\sigma_{j} + \sigma_{k}}{\sqrt{2}} ,\quad
    E_{j,k} = \frac{\sigma_{j} - \sigma_{k}}{\sqrt{2}} .
\end{gathered}
\end{equation}
In the classical case, we have the triple CHSH inequality $\langle C \rangle \leq 6$.
The minus sign preceding $\sigma_{\mathrm{y}}$ in \cref{eq:triple_chsh_honest_observables} is due to the perfect anticorrelation of $\sigma_{\mathrm{y}}$ between the subsystems of Alice and Bob in the state $\ket{\Phi^{+}}$.
We could just as easily saturate the quantum bound instead taking $A_{2} = \sigma_{\mathrm{y}}$ by changing the sign of each $A_{2}$ in \cref{eq:triple_chsh}.

It has previously been shown (see \cite{acin2016optimal,bowles2018self}) that a maximal violation $\bra{\psi} C \ket{\psi} = 6\sqrt{2}$ self-tests the reference state $\ket{\Phi^{+}}$ and the reference observables $\{\sigma_{\mathrm{x}}, \sigma_{\mathrm{y}}, \sigma_{\mathrm{z}}\}$ or their complex conjugates (in the computational basis) $\{\sigma_{\mathrm{x}}, - \sigma_{\mathrm{y}}, \sigma_{\mathrm{z}}\}$, acting on Alice's subsystem (with the complex measurement $\sigma_{\mathrm{y}}$ self-tested in the sense of \textcite{mckague2011generalized}).
Specifically, we have the following theorem.
\begin{theorem}[\textcite{bowles2018self}]
\label{thm:triple_chsh_test}
    Suppose the state $\ket{\psi} \in \hilb{A} \otimes \hilb{B}$ and the observables $A_{j} \in \mathcal{L}(\hilb{A})$ and $D_{j,k}, E_{j,k} \in \mathcal{L}(\hilb{B})$ satisfy
    \begin{equation}
        \bra{\psi} C \ket{\psi} = 6\sqrt{2} - \varepsilon .
    \end{equation}
    Then there exist linear isometries $V_{\hilb{A}} \colon \hilb{A} \to \hilb{A} \otimes \hilb{A}^{\prime} \otimes \hilb{A}^{\prime\prime}$ and $V_{\hilb{B}} \colon \hilb{B} \to \hilb{B} \otimes \hilb{B}^{\prime} \otimes \hilb{B}^{\prime\prime}$ defining the local isometry $V = V_{\hilb{A}} \otimes V_{\hilb{B}}$ such that
    \begin{subequations}
    \begin{align}
        V \ket{\psi}
        & \approxe{c \sqrt{\varepsilon}} \ket{\Phi^{+}}_{\hilb{A}^{\prime} \hilb{B}^{\prime}} \otimes \ket{\xi} , \\
        V A_{1} \ket{\psi}
        & \approxe{c \sqrt{\varepsilon}} \sigma_{\mathrm{x}}^{\hilb{A}^{\prime}} \ket{\Phi^{+}}_{\hilb{A}^{\prime} \hilb{B}^{\prime}} \otimes \ket{\xi} , \\
        \label{eq:triple_chsh_y}
        V A_{2} \ket{\psi}
        & \approxe{c \sqrt{\varepsilon}} - \sigma_{\mathrm{y}}^{\hilb{A}^{\prime}} \ket{\Phi^{+}}_{\hilb{A}^{\prime} \hilb{B}^{\prime}} \otimes \sigma_{\mathrm{z}}^{\hilb{A}^{\prime\prime}} \ket{\xi} , \\
        V A_{3} \ket{\psi}
        & \approxe{c \sqrt{\varepsilon}} \sigma_{\mathrm{z}}^{\hilb{A}^{\prime}} \ket{\Phi^{+}}_{\hilb{A}^{\prime} \hilb{B}^{\prime}} \otimes \ket{\xi} ,
    \end{align}
    \end{subequations}
    where $c$ is a nonnegative constant and the state $\ket{\xi} \in \hilb{A} \otimes \hilb{A}^{\prime\prime} \otimes \hilb{B} \otimes \hilb{B}^{\prime\prime}$ has the form
    \begin{equation}
        \ket{\xi}
        = \ket{00}_{\hilb{A}^{\prime\prime} \hilb{B}^{\prime\prime}} \otimes \ket{\xi_{0}}_{\hilb{A} \hilb{B}}
        + \ket{11}_{\hilb{A}^{\prime\prime} \hilb{B}^{\prime\prime}} \otimes \ket{\xi_{1}}_{\hilb{A} \hilb{B}}
    \end{equation}
    for some subnormalized $\ket{\xi_0}_{\hilb{A}\hilb{B}}$ and $\ket{\xi_1}_{\hilb{A}\hilb{B}}$ satisfying $\braket{\xi_{0}}{\xi_{0}}_{\hilb{A} \hilb{B}} + \braket{\xi_{1}}{\xi_{1}}_{\hilb{A} \hilb{B}} = 1$.
\end{theorem}

The appearance of the additional $\sigma_{\mathrm{z}}$ observable in \cref{eq:triple_chsh_y} acting on the ancilla space $\hilb{A}^{\prime\prime}$ can be explained as performing a measurement of Alice's junk state ancilla in the computational basis, the outcome of which controls whether $\sigma_{\mathrm{y}}$ or $- \sigma_{\mathrm{y}}$ is applied to Alice's half of $\ket{\Phi^{+}}$.
The probability of applying the complex-conjugate observable is given by $\braket{\xi_{1}}{\xi_{1}}_{\hilb{A} \hilb{B}}$.

We show a variant of this result that is instead consistent with complex conjugation being performed in a basis in which the state $\ket{\Phi^{+}}$ and observables $\sigma_{\mathrm{x}}$ and $\sigma_{\mathrm{y}}$ have real matrices, while $\sigma_{\mathrm{z}}$ satisfies $\sigma_{\mathrm{z}}^{*} = - \sigma_{\mathrm{z}}$.
The isometry used to do this is depicted in \cref{fig:triple_chsh_isom} as a circuit acting on the state $\ket{\psi}_{AB}$.
This circuit is similar to the usual partial \emph{swap} isometry (used for self-testing in \cite{mckague2012robust}) followed by phase kickback unitaries controlled by additional ancilla qubits \cite{mckague2011generalized}.
It is modified to use physical operators corresponding to Pauli $X$ and $Y$ in the first ``swap'' stage, and Pauli $Z$ in the second ``phase kickback'' stage.
The unitary operators $\regular{X}_{\hilb{B}}$, $\regular{Y}_{\hilb{B}}$, and $\regular{Z}_{\hilb{B}}$ contained on Bob's side of the circuit are regularized versions of
\begin{subequations}
\begin{align}
    X_{\hilb{B}} &= \frac{D_{\mathrmpair{x}{y}} + E_{\mathrmpair{x}{y}}}{\sqrt{2}} , \\
    Y_{\hilb{B}} &= \frac{D_{\mathrmpair{x}{y}} - E_{\mathrmpair{x}{y}}}{\sqrt{2}} , \\
    Z_{\hilb{B}} &= \frac{D_{\mathrmpair{z}{x}} + E_{\mathrmpair{z}{x}}}{\sqrt{2}} .
\end{align}
\end{subequations}

\begin{figure*}[htb]
    \centering
    \begin{quantikz}
        \lstick{$\ket{0}_{\hilb{A}^{\prime}}$} & \gate{H} & \ctrl{2} & \gate{H} & \ctrl{2} & \qw & \qw & \qw & \qw & \qw & \qw\rstick[wires=6]{$\ket{\Phi^{+}}_{\hilb{A}^{\prime} \hilb{B}^{\prime}}$} \\
        \lstick{$\ket{0}_{\hilb{A}^{\prime\prime}}$} & \qw & \qw & \qw & \qw & \gate{H} & \ctrl{1} & \gate{H} & \qw\rstick[wires=4]{$\ket{\xi}$} \\
        \makeebit{$\ket{\psi}_{\hilb{A}\hilb{B}}$} & \qw & \gate{i Y_{\hilb{A}} X_{\hilb{A}}} & \qw & \gate{X_{\hilb{A}}} & \qw & \gate{Z_{\hilb{A}}} & \qw & \qw \\
        & \qw & \gate{i \regular{Y}_{\hilb{B}} \regular{X}_{\hilb{B}}} & \qw & \gate{\regular{X}_{\hilb{B}}} & \qw & \gate{\regular{Z}_{\hilb{B}}} & \qw & \qw \\
        \lstick{$\ket{0}_{\hilb{B}^{\prime\prime}}$} & \qw & \qw & \qw & \qw & \gate{H} & \ctrl{-1} & \gate{H} & \qw \\
        \lstick{$\ket{0}_{\hilb{B}^{\prime}}$} & \gate{H} & \ctrl{-2} & \gate{H} & \ctrl{-2} & \qw & \qw & \qw & \qw & \qw & \qw
    \end{quantikz}
    \caption{
        A modified partial swap isometry followed by phase kickback unitaries, acting on the state $\ket{\psi}_{\hilb{A} \hilb{B}}$, which is used to self-test the state $\ket{\Phi^{+}}$ and measurements of $\{\sigma_{\mathrm{x}}, \sigma_{\mathrm{y}}, \sigma_{\mathrm{z}}\}$ or $\{\sigma_{\mathrm{x}}, \sigma_{\mathrm{y}}, - \sigma_{\mathrm{z}}\}$, given a maximal violation of the triple CHSH inequality.
        The unitary operators on Alice's subsystem are simply $X_{\hilb{A}} = A_{1}$, $Y_{\hilb{A}} = - A_{2}$, and $Z_{\hilb{A}} = A_{3}$.
        The unitary operators $\regular{X}_{\hilb{B}}$, $\regular{Y}_{\hilb{B}}$, and $\regular{Z}_{\hilb{B}}$ are regularized versions of $X_{\hilb{B}}$, $Y_{\hilb{B}}$, and $Z_{\hilb{B}}$ respectively, which are each defined in terms of the operators $D_{j,k}$ and $E_{j,k}$.
    }
    \label{fig:triple_chsh_isom}
\end{figure*}

We state the result in a more explicit form than that of \cref{def:complex_self-testing}, as doing so will later prove useful in the proof of our parallel version of the self-test (see \cref{sec:parallel_test_proof}).
To write the result, we introduce the notation $W \colon \hilb{A} \otimes \hilb{B} \to \hilb{A} \otimes \hilb{A}^{\prime} \otimes \hilb{B} \otimes \hilb{B}^{\prime}$ and $K \colon \hilb{A} \otimes \hilb{B} \to \hilb{A} \otimes \hilb{A}^{\prime\prime} \otimes \hilb{B} \otimes \hilb{B}^{\prime\prime}$ for the local ``swap'' and ``phase kickback'' isometries, respectively, that are constructed in \cref{fig:swap_kickback}.

\begin{figure*}[htb]
    \centering
    \begin{quantikz}
        \lstick{$\ket{0}_{\hilb{A}^{\prime}}$} & \gate{H} & \ctrl{2} & \gate{H} & \ctrl{2}\slice{Swapped} & \qw & \qw & \qw & \qw & \qw & \qw\rstick[wires=6]{$\ket{\Phi^{+}}_{\hilb{A}^{\prime} \hilb{B}^{\prime}}$} \\
        \lstick{$\ket{0}_{\hilb{A}^{\prime\prime}}$} & \qw & \qw & \qw & \qw & \gate{H} & \ctrl{1} & \gate{H} & \qw\rstick[wires=4]{$\ket{\xi}$} \\
        \makeebit{$\ket{\psi}_{\hilb{A} \hilb{B}}$} & \qw & \gate{- i S_{2} S_{1}} & \qw & \gate{S_{1}} & \qw & \gate{S_{3}} & \qw & \qw \\
        & \qw & \gate{i T_{2} T_{1}} & \qw & \gate{T_{1}} & \qw & \gate{T_{3}} & \qw & \qw \\
        \lstick{$\ket{0}_{\hilb{B}^{\prime\prime}}$} & \qw & \qw & \qw & \qw & \gate{H} & \ctrl{-1} & \gate{H} & \qw \\
        \lstick{$\ket{0}_{\hilb{B}^{\prime}}$} & \gate{H} & \ctrl{-2} & \gate{H} & \ctrl{-2} & \qw & \qw & \qw & \qw & \qw & \qw
    \end{quantikz}
    \caption{
        The circuit describing the action of the local isometry $V = K W$ on the state $\ket{\psi}_{\hilb{A} \hilb{B}}$.
        The isometries $W_{\hilb{A}} \colon \hilb{A} \to \hilb{A} \otimes \hilb{A}^{\prime}$ and $W_{\hilb{B}} \colon \hilb{B} \to \hilb{B} \otimes \hilb{B}^{\prime}$ and the local isometry $W = W_{\hilb{A}} \otimes W_{\hilb{B}}$ are defined by the first ``swap'' stage of the circuit (preceding the dotted line), in which a maximally entangled state is extracted.
        In the second ``phase kickback'' stage (succeeding the dotted line), denoted by $K = K_{\hilb{A}} \otimes K_{\hilb{B}}$ for isometries $K_{\hilb{A}} \colon \hilb{A} \to \hilb{A} \otimes \hilb{A}^{\prime\prime}$ and $K_{\hilb{B}} \colon \hilb{B} \to \hilb{B} \otimes \hilb{B}^{\prime\prime}$, possible complex conjugation in the presence of a Pauli $\sigma_{\mathrm{z}}$ operator is accounted for.
    }
    \label{fig:swap_kickback}
\end{figure*}

\begin{proposition}
\label{prop:swap_and_kickback_isoms}
    Let $\ket{\psi}_{\hilb{A} \hilb{B}} \in \hilb{A} \otimes \hilb{B}$.
    Suppose for each $q \in \{1, 2, 3\}$ that there exist $\pm 1$-outcome observables $S_{q}$ on $\hilb{A}$ and $T_{q}$ on $\hilb{B}$ satisfying (for some $\eta \geq 0$) the following relations:
    \begin{enumerate}
        \item \label[relation]{rel:corr_single} $(S_{q} - T_{q}) \ket{\psi}_{\hilb{A} \hilb{B}} \approxe{\eta} 0$ for all $q$.
        \item \label[relation]{rel:acomm_single} $\{S_{q}, S_{r}\} \ket{\psi}_{\hilb{A} \hilb{B}} \approxe{\eta} 0$ and $\{T_{q}, T_{r}\} \ket{\psi}_{\hilb{A} \hilb{B}} \approxe{\eta} 0$ for all $q$ and $r$ such that $q \neq r$.
    \end{enumerate}
    Construct the local ``swap'' isometry $W \colon \hilb{A} \otimes \hilb{B} \to \hilb{A} \otimes \hilb{A}^{\prime} \otimes \hilb{B} \otimes \hilb{B}^{\prime}$ and the local ``phase kickback'' isometry $K \colon \hilb{A} \otimes \hilb{B} \to \hilb{A} \otimes \hilb{A}^{\prime\prime} \otimes \hilb{B} \otimes \hilb{B}^{\prime\prime}$ as in \cref{fig:swap_kickback}.
    Then
    \begin{subequations}
    \begin{align}
        W \ket{\psi}_{\hilb{A} \hilb{B}}
        & \approxe{c_{0} \eta} \ket{\Phi^{+}}_{\hilb{A}^{\prime} \hilb{B}^{\prime}} \otimes \ket{\varphi}_{\hilb{A} \hilb{B}} , \\
        W S_{1} \ket{\psi}_{\hilb{A} \hilb{B}}
        & \approxe{c_{1} \eta} \sigma_{\mathrm{x}}^{\hilb{B}^{\prime}} \ket{\Phi^{+}}_{\hilb{A}^{\prime} \hilb{B}^{\prime}}
        \otimes \ket{\varphi}_{\hilb{A} \hilb{B}} , \\
        W S_{2} \ket{\psi}_{\hilb{A} \hilb{B}}
        & \approxe{c_{2} \eta} \sigma_{\mathrm{y}}^{\hilb{B}^{\prime}} \ket{\Phi^{+}}_{\hilb{A}^{\prime} \hilb{B}^{\prime}}
        \otimes \ket{\varphi}_{\hilb{A} \hilb{B}} , \\
        W S_{3} \ket{\psi}_{\hilb{A} \hilb{B}}
        & \approxe{c_{3} \eta} \sigma_{\mathrm{z}}^{\hilb{B}^{\prime}} \ket{\Phi^{+}}_{\hilb{A}^{\prime} \hilb{B}^{\prime}}
        \otimes S_{3} \ket{\varphi}_{\hilb{A} \hilb{B}} ,
    \end{align}
    \end{subequations}
    where $\ket{\varphi}_{\hilb{A} \hilb{B}} \in \hilb{A} \otimes \hilb{B}$ is defined as
    \begin{equation}
        \ket{\varphi}_{\hilb{A} \hilb{B}} = \frac{1}{\sqrt{2}} (I + i T_{2} T_{1}) \ket{\psi}_{\hilb{A} \hilb{B}}
    \end{equation}
    and nonnegative constants $c_{0}$, $c_{1}$, $c_{2}$, and $c_{3}$ are defined as
    \begin{subequations}
    \begin{align}
        c_{0} & = \frac{1}{4} \mathopen{}\left( 18 + 5 \sqrt{2} \right)\mathclose{} , \\
        c_{1} & = \frac{1}{4} \mathopen{}\left( 26 + 5 \sqrt{2} \right)\mathclose{} , \\
        c_{2} & = \frac{1}{4} \mathopen{}\left( 34 + 5 \sqrt{2} \right)\mathclose{} , \\
        c_{3} & = \frac{1}{4} \mathopen{}\left( 66 + 5 \sqrt{2} \right)\mathclose{} .
    \end{align}
    \end{subequations}
    Moreover,
    \begin{subequations}
    \begin{align}
        K \ket{\varphi}_{\hilb{A} \hilb{B}}
        & \approxe{6 \sqrt{2} \eta} \ket{\xi} , \\
        K S_{3} \ket{\varphi}_{\hilb{A} \hilb{B}}
        & \approxe{6 \sqrt{2} \eta} \sigma_{\mathrm{z}}^{\hilb{B}^{\prime\prime}} \ket{\xi} .
    \end{align}
    \end{subequations}
    where the state $\ket{\xi} \in \hilb{A} \otimes \hilb{A}^{\prime\prime} \otimes \hilb{B} \otimes \hilb{B}^{\prime\prime}$ has the form
    \begin{equation}
        \ket{\xi}
        = \ket{0}_{\hilb{A}^{\prime\prime}} \ket{0}_{\hilb{B}^{\prime\prime}} \otimes \ket{\xi_{0}}_{\hilb{A} \hilb{B}}
        + \ket{1}_{\hilb{A}^{\prime\prime}} \ket{1}_{\hilb{B}^{\prime\prime}} \otimes \ket{\xi_{1}}_{\hilb{A} \hilb{B}} ,
    \end{equation}
    and the subnormalized $\ket{\xi_{0}}_{\hilb{A} \hilb{B}}$ and $\ket{\xi_{1}}_{\hilb{A} \hilb{B}}$ are given by
    \begin{subequations}
    \begin{align}
        \ket{\xi_{0}}_{\hilb{A} \hilb{B}}
        & = \frac{I + T_{3}}{2} \ket{\varphi}_{\hilb{A} \hilb{B}} , \\
        \ket{\xi_{1}}_{\hilb{A} \hilb{B}}
        & = \frac{I - T_{3}}{2} \ket{\varphi}_{\hilb{A} \hilb{B}} .
    \end{align}
    \end{subequations}
\end{proposition}
\begin{proof}
    See \cref{sec:single_test_proof} for the ideal case.
    Robustness is discussed in \cref{sec:single_test_robustness}.
\end{proof}

The operators $X_{\hilb{A}}$, $- Y_{\hilb{A}}$, and $Z_{\hilb{A}}$ for Alice, and $\regular{X}_{\hilb{B}}$, $\regular{Y}_{\hilb{B}}$, and $\regular{Z}_{\hilb{B}}$ for Bob satisfy the conditions of \cref{prop:swap_and_kickback_isoms}.
This can be seen from the SOS decomposition (see \cref{sec:sos_decomposition}) of the triple CHSH operator of \cref{eq:triple_chsh}
\begin{equation}
\begin{split}
    6 \sqrt{2} - C ={}
    &   \frac{1}{\sqrt{2}} \mathopen{}\left( A_{3} - \frac{D_{\mathrmpair{z}{x}} + E_{\mathrmpair{z}{x}}}{\sqrt{2}} \right)^{2} \\
    & + \frac{1}{\sqrt{2}} \mathopen{}\left( A_{1} - \frac{D_{\mathrmpair{z}{x}} - E_{\mathrmpair{z}{x}}}{\sqrt{2}} \right)^{2} \\
    & + \frac{1}{\sqrt{2}} \mathopen{}\left( A_{3} - \frac{D_{\mathrmpair{z}{y}} + E_{\mathrmpair{z}{y}}}{\sqrt{2}} \right)^{2} \\
    & + \frac{1}{\sqrt{2}} \mathopen{}\left( A_{2} - \frac{D_{\mathrmpair{z}{y}} - E_{\mathrmpair{z}{y}}}{\sqrt{2}} \right)^{2} \\
    & + \frac{1}{\sqrt{2}} \mathopen{}\left( A_{1} - \frac{D_{\mathrmpair{x}{y}} + E_{\mathrmpair{x}{y}}}{\sqrt{2}} \right)^{2} \\
    & + \frac{1}{\sqrt{2}} \mathopen{}\left( A_{2} - \frac{D_{\mathrmpair{x}{y}} - E_{\mathrmpair{x}{y}}}{\sqrt{2}} \right)^{2} .
\end{split}
\end{equation}
Considering the ideal case for simplicity, a maximal violation $\bra{\psi} C \ket{\psi} = 6 \sqrt{2}$ then implies that
\begin{subequations}
\begin{alignat}{2}
    A_{1} \ket{\psi}
    & = \frac{D_{\mathrmpair{x}{y}} + E_{\mathrmpair{x}{y}}}{\sqrt{2}} \ket{\psi}
    && = \frac{D_{\mathrmpair{z}{x}} - E_{\mathrmpair{z}{x}}}{\sqrt{2}} \ket{\psi} , \\
    A_{2} \ket{\psi}
    & = \frac{D_{\mathrmpair{x}{y}} - E_{\mathrmpair{x}{y}}}{\sqrt{2}} \ket{\psi}
    && = \frac{D_{\mathrmpair{z}{y}} - E_{\mathrmpair{z}{y}}}{\sqrt{2}} \ket{\psi} , \\
    A_{3} \ket{\psi}
    & = \frac{D_{\mathrmpair{z}{x}} + E_{\mathrmpair{z}{x}}}{\sqrt{2}} \ket{\psi}
    && = \frac{D_{\mathrmpair{z}{y}} + E_{\mathrmpair{z}{y}}}{\sqrt{2}} \ket{\psi} .
\end{alignat}
\end{subequations}
Since operators $D+E$ and $D-E$ anticommute for any $\pm 1$-outcome observables $D$ and $E$, the state-dependent anticommutation relations for Alice's observables $X_{\hilb{A}}$, $- Y_{\hilb{A}}$, and $Z_{\hilb{A}}$ are satisfied.
Similarly, since Alice's operators commute with Bob's observables, state-dependent anticommutation relations for $X_{\hilb{B}}$, $Y_{\hilb{B}}$, and $Z_{\hilb{B}}$ are also satisfied.
Applying \cref{lem:regularized_acomm} shows the required state-dependent anticommutation relations for $\regular{X}_{\hilb{B}}$, $\regular{Y}_{\hilb{B}}$, and $\regular{Z}_{\hilb{B}}$ (which are also unitary as required by \cref{prop:swap_and_kickback_isoms}).
Similar arguments apply in the robust case, resulting in the following counterpart to \cref{thm:triple_chsh_test}, which has possible complex conjugation appearing in the certification of $\sigma_{\mathrm{z}}$ rather than $\sigma_{\mathrm{y}}$.

\begin{corollary}
    Suppose the state $\ket{\psi} \in \hilb{A} \otimes \hilb{B}$ and the observables $A_{j} \in \mathcal{L}(\hilb{A})$ and $D_{j,k}, E_{j,k} \in \mathcal{L}(\hilb{B})$ satisfy
    \begin{equation}
        \bra{\psi} C \ket{\psi} = 6\sqrt{2} - \varepsilon .
    \end{equation}
    Then there exist linear isometries $V_{\hilb{A}} \colon \hilb{A} \to \hilb{A} \otimes \hilb{A}^{\prime} \otimes \hilb{A}^{\prime\prime}$ and $V_{\hilb{B}} \colon \hilb{B} \to \hilb{B} \otimes \hilb{B}^{\prime} \otimes \hilb{B}^{\prime\prime}$ defining the local isometry $V = V_{\hilb{A}} \otimes V_{\hilb{B}}$ such that
    \begin{subequations}
    \begin{align}
        V \ket{\psi}
        & \approxe{c \sqrt{\varepsilon}} \ket{\Phi^{+}}_{\hilb{A}^{\prime} \hilb{B}^{\prime}} \otimes \ket{\xi} , \\
        V A_{1} \ket{\psi}
        & \approxe{c \sqrt{\varepsilon}} \sigma_{\mathrm{x}}^{\hilb{A}^{\prime}} \ket{\Phi^{+}}_{\hilb{A}^{\prime} \hilb{B}^{\prime}} \otimes \ket{\xi} , \\
        V A_{2} \ket{\psi}
        & \approxe{c \sqrt{\varepsilon}} - \sigma_{\mathrm{y}}^{\hilb{A}^{\prime}} \ket{\Phi^{+}}_{\hilb{A}^{\prime} \hilb{B}^{\prime}} \otimes \ket{\xi} , \\
        V A_{3} \ket{\psi}
        & \approxe{c \sqrt{\varepsilon}} \sigma_{\mathrm{z}}^{\hilb{A}^{\prime}} \ket{\Phi^{+}}_{\hilb{A}^{\prime} \hilb{B}^{\prime}} \otimes \sigma_{\mathrm{z}}^{\hilb{A}^{\prime\prime}} \ket{\xi} ,
    \end{align}
    \end{subequations}
    where $c$ is a nonnegative constant and the state $\ket{\xi} \in \hilb{A} \otimes \hilb{A}^{\prime\prime} \otimes \hilb{B} \otimes \hilb{B}^{\prime\prime}$ has the form
    \begin{equation}
        \ket{\xi}
        = \ket{00}_{\hilb{A}^{\prime\prime} \hilb{B}^{\prime\prime}} \otimes \ket{\xi_{0}}_{\hilb{A} \hilb{B}}
        + \ket{11}_{\hilb{A}^{\prime\prime} \hilb{B}^{\prime\prime}} \otimes \ket{\xi_{1}}_{\hilb{A} \hilb{B}}
    \end{equation}
    for some subnormalized $\ket{\xi_0}_{\hilb{A}\hilb{B}}$ and $\ket{\xi_1}_{\hilb{A}\hilb{B}}$ satisfying $\braket{\xi_{0}}{\xi_{0}}_{\hilb{A} \hilb{B}} + \braket{\xi_{1}}{\xi_{1}}_{\hilb{A} \hilb{B}} = 1$.
\end{corollary}

\section{Parallel self-testing (with complex measurements)}
\label{sec:parallel_self-testing}

We state in \cref{sec:protocol} a measurement scenario and Bell observations for which any valid measurement strategy results in the existence of observables which satisfy certain natural state-dependent relations, as though they are Pauli operators acting on Bell states.
Typically, one can use such relations to construct a self-testing isometry.
For example, as is often done in the case of real Pauli measurements \cite{mckague2012robust,mckague2016self,wu2016device,mckague2017self,coladangelo2017parallel,chao2018test}.
We require such a result which can, in the complex sense of \cref{def:complex_self-testing}, certify Bell states and all (possibly complex) Pauli operators from similar relations.
Such a construction was also discussed in \cite{bowles2018self}.
In \cref{fig:swap_kickback_xy_copy}, we exhibit a robust isometry of this form, which also captures the effect of complex conjugation in $\sigma_{\mathrm{z}}$ measurements (as desired for VBQC), instead of in $\sigma_{\mathrm{y}}$ as is standard.

\begin{figure*}[htb]
    \centering
    \begin{quantikz}
        \lstick{$\ket{0}_{\hilb{A}_{j}^{\prime}}$} & \gate{H} & \ctrl{2} & \gate{H} & \ctrl{2}\slice{Swapped} & \qw & \qw & \qw & \qw & \qw & \qw\rstick[wires=6]{$\ket{\Phi^{+}}_{\hilb{A}_{j}^{\prime} \hilb{B}_{j}^{\prime}}$} \\
        \lstick{$\ket{0}_{\hilb{A}_{j}^{\prime\prime}}$} & \qw & \qw & \qw & \qw & \gate{H} & \ctrl{1} & \gate{H} & \qw \\
        \makeebit{$\ket{\psi}_{\hilb{A} \hilb{B}}$} & \qw & \gate{- i S_{2}^{(j)} S_{1}^{(j)}} & \qw & \gate{S_{1}^{(j)}} & \qw & \gate{S_{3}^{(j)}} & \qw & \qw \\
        & \qw & \gate{i T_{2}^{(j)} T_{1}^{(j)}} & \qw & \gate{T_{1}^{(j)}} & \qw & \gate{T_{3}^{(j)}} & \qw & \qw \\
        \lstick{$\ket{0}_{\hilb{B}_{j}^{\prime\prime}}$} & \qw & \qw & \qw & \qw & \gate{H} & \ctrl{-1} & \gate{H} & \qw \\
        \lstick{$\ket{0}_{\hilb{B}_{j}^{\prime}}$} & \gate{H} & \ctrl{-2} & \gate{H} & \ctrl{-2} & \qw & \qw & \qw & \qw & \qw & \qw
    \end{quantikz}
    \caption{
        The circuit describing the action of the local isometry $V^{(j)} = K^{(j)} W^{(j)}$ on the state $\ket{\psi}_{\hilb{A} \hilb{B}}$.
        In the first ``swap'' stage (preceding the dotted line), denoted by $W^{(j)} = W_{\hilb{A}}^{(j)} \otimes W_{\hilb{B}}^{(j)}$, a maximally entangled state is extracted.
        In the second ``phase kickback'' stage (succeeding the dotted line), denoted by $K^{(j)} = K_{\hilb{A}}^{(j)} \otimes K_{\hilb{A}}^{(j)}$, possible complex conjugation in the presence of a Pauli $\sigma_{\mathrm{z}}$ operator is accounted for.
        The full isometry $V = V^{(n)} \dots V^{(1)}$ is a parallel version of this circuit.
        It is defined by applying the circuit for each $j$ successively, appending ancillae states in $\hilb{A}_{j}^{\prime}$ and $\hilb{B}_{j}^{\prime}$ for each swap stage, and $\hilb{A}_{j}^{\prime\prime}$ and $\hilb{B}_{j}^{\prime\prime}$ for each phase kickback stage.
        Each isometry $V^{(k)}$ is defined to act trivially on all ancilla spaces with $j < k$.
    }
    \label{fig:swap_kickback_xy_copy}
\end{figure*}

The following result shows (by application of the isometry defined in \cref{fig:swap_kickback_xy_copy}) the existence of a local isometry with self-testing properties, given natural relations that we find for our protocol in \cref{sec:operator_rels}.
\begin{theorem}
\label{thm:single_observable_isometry}
    There exists a function $\delta \colon \mathbb{R}_{\geq 0} \times \mathbb{N}^{*} \to \mathbb{R}_{\geq 0}$ satisfying $\delta(0, n) = 0$ for all $n$ such that the following holds.
    Let $n \in \mathbb{N}^{*}$ and let $\ket{\psi}_{\hilb{A} \hilb{B}}$ be a state.
    Suppose for each $q \in \{1, 2, 3\}$ and $j \in \{1, \dots, n\}$ that there exist $\pm 1$-outcome observables $S_{q}^{(j)}$ on $\hilb{A}$ and $T_{q}^{(j)}$ on $\hilb{B}$ satisfying (for some $\eta \geq 0$) the following relations:
    \begin{enumerate}
        \item \label[relation]{rel:corr} $\left( S_{q}^{(j)} - T_{q}^{(j)} \right) \ket{\psi}_{\hilb{A} \hilb{B}} \approxe{\eta} 0$ for all $q$ and $j$.
        \item \label[relation]{rel:acomm} $\left\{ S_{q}^{(j)}, S_{r}^{(j)} \right\} \ket{\psi}_{\hilb{A} \hilb{B}} \approxe{\eta} 0$ and $\left\{ T_{q}^{(j)}, T_{r}^{(j)} \right\} \ket{\psi}_{\hilb{A} \hilb{B}} \approxe{\eta} 0$ for all $q,r$ and $j$ such that $q \neq r$.
        \item \label[relation]{rel:comm} $\left[ S_{q}^{(j)}, S_{r}^{(k)} \right] \ket{\psi}_{\hilb{A} \hilb{B}} \approxe{\eta} 0$ and $\left[ T_{q}^{(j)}, T_{r}^{(k)} \right] \ket{\psi}_{\hilb{A} \hilb{B}} \approxe{\eta} 0$ for all $q,r$ and $j,k$ such that $j \neq k$.
        \item \label[relation]{rel:conj} $\left( I + S_{1}^{(j)} S_{1}^{(j+1)} S_{2}^{(j)} S_{2}^{(j+1)} S_{3}^{(j)} S_{3}^{(j+1)} \right) \ket{\psi}_{\hilb{A} \hilb{B}} \approxe{\eta} 0$ for all $j < n$.
    \end{enumerate}
    Then there exist subnormalized $\ket{\xi_{0}}_{\hilb{A} \hilb{B}}$ and $\ket{\xi_{1}}_{\hilb{A} \hilb{B}}$ satisfying $\braket{\xi_{0}}{\xi_{0}}_{\hilb{A} \hilb{B}} + \braket{\xi_{1}}{\xi_{1}}_{\hilb{A} \hilb{B}} = 1$ that do not depend on any of the $S_{q}^{(j)}$, an isometry $V_{\hilb{A}} \colon \hilb{A} \to \hilb{A} \otimes \hilb{A}^{\prime} \otimes \hilb{A}^{\prime\prime}$, and an isometry $V_{\hilb{B}} \colon \hilb{B} \to \hilb{B} \otimes \hilb{B}^{\prime} \otimes \hilb{B}^{\prime\prime}$ that does not depend on any of the $S_{q}^{(j)}$ such that for all $q \in \{1, 2, 3\}$ and $k \in \{1, \dots, n\}$ we have
    \begin{subequations}
    \begin{align}
        V \ket{\psi}_{\hilb{A} \hilb{B}}
        & \approxe{\delta(\eta, n)} \bigotimes_{j=1}^{n} \ket{\Phi^{+}}_{\hilb{A}_{j}^{\prime} \hilb{B}_{j}^{\prime}} \otimes \ket{\xi} , \\
        V S_{q}^{(k)} \ket{\psi}_{\hilb{A} \hilb{B}}
        & \approxe{\delta(\eta, n)} \sigma_{q}^{\hilb{B}_{k}^{\prime}} \bigotimes_{j=1}^{n} \ket{\Phi^{+}}_{\hilb{A}_{j}^{\prime} \hilb{B}_{j}^{\prime}} \otimes \sigma_{3 [q = 3]}^{\hilb{B}^{\prime\prime}} \ket{\xi} ,
    \end{align}
    \end{subequations}
    where $V = V_{\hilb{A}} \otimes V_{\hilb{B}}$ and the junk state $\ket{\xi} \in \hilb{A} \otimes \hilb{A}^{\prime\prime} \otimes \hilb{B} \otimes \hilb{B}^{\prime\prime}$ is defined as
    \begin{equation}
    \label{eq:junk_form}
        \ket{\xi}
        = \ket{0}_{\hilb{A}^{\prime\prime}} \ket{0}_{\hilb{B}^{\prime\prime}} \otimes \ket{\xi_{0}}_{\hilb{A} \hilb{B}}
        + \ket{1}_{\hilb{A}^{\prime\prime}} \ket{1}_{\hilb{B}^{\prime\prime}} \otimes \ket{\xi_{1}}_{\hilb{A} \hilb{B}} .
    \end{equation}
\end{theorem}
\begin{remark}
    In the statement, the untrusted observables are treated so that the $T_{2}^{(j)}$ represent Pauli $\sigma_{\mathrm{y}}$ operators acting on Bob's side, while the $S_{2}^{(j)}$ represent $- \sigma_{\mathrm{y}}$ operators acting on Alice's side.

    That $V_{\hilb{B}}$, $\ket{\xi_{0}}$, and $\ket{\xi_{1}}$ do not depend on $S_{q}^{(j)}$ means that if we replaced this observable with some $\tilde{S}_{q}^{(j)}$ such that all assumptions were still satisfied, then the same $V_{\hilb{B}}$, $\ket{\xi_{0}}$, and $\ket{\xi_{1}}$ would still be sufficient to meet the conditions of the result.

    We prove that the robustness $\delta(\eta, n)$ satisfies $\delta(0, n) = 0$ for all $n$, which is the ideal, noiseless case.
    For the case where $\eta > 0$, we expect that standard existing techniques \cite{mckague2016self,mckague2017self,coladangelo2017parallel,chao2018test,bowles2018self} can be applied to achieve polynomial robustness bounds in the sense $\delta(\eta, n) = \bigO(\eta n^{2})$ as $\eta \to 0$ or $n \to \infty$.
\end{remark}
\begin{proof}[Proof sketch]
    See \cref{sec:parallel_test_proof} for full details.
    As noted, here we consider only the ideal case with $\eta = 0$.
    We show that a sufficient local isometry is $V = V^{(n)} \dots V^{(1)}$, where each $V^{(j)} = K^{(j)} W^{(j)}$, and $K^{(j)} = K_{\hilb{A}}^{(j)} \otimes K_{\hilb{B}}^{(j)}$ and $W^{(j)} = W_{\hilb{A}}^{(j)} \otimes W_{\hilb{B}}^{(j)}$ are as defined in \cref{fig:swap_kickback_xy_copy}.
    We could also have considered the isometry constructed by first applying all swap isometries, followed by applying all phase kickback isometries.
    However, applying $K^{(j)}$ immediately after the corresponding $W^{(j)}$ reduces the number of terms that must be manipulated at a time.
    This, along with the usage of \cref{rel:conj} (via \cref{lem:global_conj}) at each step, allows us to keep the number of terms constant after the application of each $V^{(j)}$.

    The isometry $V_{\hilb{B}} = V_{\hilb{B}}^{(n)} \dots V_{\hilb{B}}^{(1)}$, where each $V_{\hilb{B}}^{(j)} = K_{\hilb{B}}^{(j)} W_{\hilb{B}}^{(j)}$, is immediately seen from the construction given in \cref{fig:swap_kickback_xy_copy} not to depend on any of the observables $S_{q}^{(j)}$.
    Let us define, for all $1 \leq k \leq n$, the vectors
    \begin{equation}
    \label{eq:iter_junk_sketch}
        \ket{\xi_{\pm}^{k}}
        = \frac{1}{\left( 2 \sqrt{2} \right)^{k}} \prod_{j=1}^{k} \mathopen{}\left( I \pm T_{3}^{(j)} \right)\mathclose{} \mathopen{}\left( I + i T_{2}^{(j)} T_{1}^{(j)} \right) \ket{\psi} .
    \end{equation}
    With this notation, \cref{prop:swap_and_kickback_isoms} gives that
    \begin{subequations}
    \label{eq:single_isom}
    \begin{align}
    \begin{split}
        \label{eq:single_isom_1_state}
        V^{(1)} \ket{\psi}
        & = \ket{\Phi^{+}}_{\hilb{A}_{1}^{\prime} \hilb{B}_{1}^{\prime}}
        \otimes \Bigl( \ket{0}_{\hilb{A}_{1}^{\prime\prime}} \ket{0}_{\hilb{B}_{1}^{\prime\prime}}
        \otimes \ket{\xi_{+}^{1}} \\
        & \phantom{{}={}} + \ket{1}_{\hilb{A}_{1}^{\prime\prime}} \ket{1}_{\hilb{B}_{1}^{\prime\prime}}
        \otimes \ket{\xi_{-}^{1}} \Bigr)\mathclose{} ,
    \end{split} \\
    \begin{split}
        \label{eq:single_isom_1_obs}
        V^{(1)} S_{q}^{(1)} \ket{\psi}
        & = \sigma_{q}^{\hilb{B}_{1}^{\prime}} \ket{\Phi^{+}}_{\hilb{A}_{1}^{\prime} \hilb{B}_{1}^{\prime}}
        \otimes \Bigl( \ket{0}_{\hilb{A}_{1}^{\prime\prime}} \ket{0}_{\hilb{B}_{1}^{\prime\prime}} \otimes \ket{\xi_{+}^{1}} \\
        & \phantom{{}={}} + (-1)^{[q=3]} \ket{1}_{\hilb{A}_{1}^{\prime\prime}} \ket{1}_{\hilb{B}_{1}^{\prime\prime}} \otimes \ket{\xi_{-}^{1}} \Bigr)\mathclose{} .
    \end{split}
    \end{align}
    \end{subequations}
    For all $1 < k \leq n$, we have the following properties.
    First,
    \begin{subequations}
    \label{eq:isom_step}
    \begin{align}
        V^{(k)} \ket{\xi_{+}^{k-1}}
        & = \ket{\Phi^{+}}_{\hilb{A}_{k}^{\prime} \hilb{B}_{k}^{\prime}}
        \otimes \ket{0}_{\hilb{A}_{k}^{\prime\prime}} \ket{0}_{\hilb{B}_{k}^{\prime\prime}}
        \otimes \ket{\xi_{+}^{k}} , \\
        V^{(k)} \ket{\xi_{-}^{k-1}}
        & = \ket{\Phi^{+}}_{\hilb{A}_{k}^{\prime} \hilb{B}_{k}^{\prime}}
        \otimes \ket{1}_{\hilb{A}_{k}^{\prime\prime}} \ket{1}_{\hilb{B}_{k}^{\prime\prime}}
        \otimes \ket{\xi_{-}^{k}} .
    \end{align}
    \end{subequations}
    That is, $V^{(k)}$ extracts a Bell state from $\ket{\xi_{\pm}^{k-1}}$, raises its superscript index from $k-1$ to $k$, and appends the appropriate ancilla states depending on the sign in subscript.
    Second,
    \begin{equation}
    \label{eq:single_isom_1_state_switch}
        V^{(1)} S_{q}^{(k)} \ket{\psi}
        = S_{q}^{(k)} V^{(1)} \ket{\psi} .
    \end{equation}
    Third, whenever $j \neq k$,
    \begin{equation}
    \label{eq:single_isom_junk_switch}
        V^{(k)} S_{q}^{(j)} \ket{\xi_{\pm}^{k-1}}
        = S_{q}^{(j)} V^{(k)} \ket{\xi_{\pm}^{k-1}} .
    \end{equation}
    Finally,
    \begin{subequations}
    \label{eq:isom_step_obs}
    \begin{align}
        V^{(k)} S_{q}^{(k)} \ket{\xi_{+}^{k-1}}
        & = \sigma_{q}^{\hilb{B}_{k}^{\prime}} \ket{\Phi^{+}}_{\hilb{A}_{k}^{\prime} \hilb{B}_{k}^{\prime}}
        \otimes \ket{0}_{\hilb{A}_{k}^{\prime\prime}} \ket{0}_{\hilb{B}_{k}^{\prime\prime}}
        \otimes \ket{\xi_{+}^{k}} , \\
    \begin{split}
        V^{(k)} S_{q}^{(k)} \ket{\xi_{-}^{k-1}}
        & = (-1)^{[q=3]} \sigma_{q}^{\hilb{B}_{k}^{\prime}} \ket{\Phi^{+}}_{\hilb{A}_{k}^{\prime} \hilb{B}_{k}^{\prime}} \\
        & \phantom{{}={}} \otimes \ket{1}_{\hilb{A}_{k}^{\prime\prime}} \ket{1}_{\hilb{B}_{k}^{\prime\prime}}
        \otimes \ket{\xi_{-}^{k}} .
    \end{split}
    \end{align}
    \end{subequations}

    After the full application of the isometry $V = V^{(n)} \dots V^{(1)}$, and defining
    \begin{subequations}
    \begin{align}
        \ket{0}_{\hilb{A^{\prime\prime}}} & = \ket{0 \dots 0}_{\hilb{A}^{\prime\prime}} , \\
        \ket{1}_{\hilb{A^{\prime\prime}}} & = \ket{1 \dots 1}_{\hilb{A}^{\prime\prime}} , \\
        \ket{0}_{\hilb{B^{\prime\prime}}} & = \ket{0 \dots 0}_{\hilb{B}^{\prime\prime}} , \\
        \ket{1}_{\hilb{B^{\prime\prime}}} & = \ket{1 \dots 1}_{\hilb{B}^{\prime\prime}} ,
    \end{align}
    \end{subequations}
    \cref{eq:single_isom_1_state,eq:isom_step} give
    \begin{equation}
    \begin{split}
    \label{eq:isom_state_result}
        V \ket{\psi}
        = \bigotimes_{j=1}^{n} \ket{\Phi^{+}}_{\hilb{A}_{j}^{\prime} \hilb{B}_{j}^{\prime}}
        \otimes \bigl( &
        \ket{0}_{\hilb{A}^{\prime\prime}} \ket{0}_{\hilb{B}^{\prime\prime}} \otimes \ket{\xi_{+}^{n}} \\
        & + \ket{1}_{\hilb{A}^{\prime\prime}} \ket{1}_{\hilb{B}^{\prime\prime}} \otimes \ket{\xi_{-}^{n}}
        \bigr)\mathclose{} .
    \end{split}
    \end{equation}
    Similarly, using \cref{eq:single_isom_1_obs,eq:isom_step}, we have
    \begin{equation}
    \label{eq:isom_obs_keq1}
    \begin{split}
        V S_{q}^{(1)} \ket{\psi}
        & = \sigma_{q}^{\hilb{B}_{1}^{\prime}} \ket{\Phi^{+}}_{\hilb{A}_{1}^{\prime} \hilb{B}_{1}^{\prime}} \otimes V^{(n)} \dots V^{(2)} \mathopen{}\left( \ket{0}_{\hilb{A}_{1}^{\prime\prime}} \ket{0}_{\hilb{B}_{1}^{\prime\prime}} \otimes \ket{\xi_{+}^{1}} + (-1)^{[q=3]} \ket{1}_{\hilb{A}_{1}^{\prime\prime}} \ket{1}_{\hilb{B}_{1}^{\prime\prime}} \otimes \ket{\xi_{-}^{1}} \right) \\
        & = \sigma_{q}^{\hilb{B}_{1}^{\prime}} \bigotimes_{j=1}^{n} \ket{\Phi^{+}}_{\hilb{A}_{j}^{\prime} \hilb{B}_{j}^{\prime}} \otimes \left( \ket{0}_{\hilb{A}^{\prime\prime}} \ket{0}_{\hilb{B}^{\prime\prime}} \otimes \ket{\xi_{+}^{n}} + (-1)^{[q=3]} \ket{1}_{\hilb{A}^{\prime\prime}} \ket{1}_{\hilb{B}^{\prime\prime}} \otimes \ket{\xi_{-}^{n}} \right) \\
        & = \sigma_{q}^{\hilb{B}_{1}^{\prime}} \bigotimes_{j=1}^{n} \ket{\Phi^{+}}_{\hilb{A}_{j}^{\prime} \hilb{B}_{j}^{\prime}} \otimes \sigma_{3[q=3]}^{\hilb{B}^{\prime\prime}} \mathopen{}\left( \ket{0}_{\hilb{A}^{\prime\prime}} \ket{0}_{\hilb{B}^{\prime\prime}} \otimes \ket{\xi_{+}^{n}} + \ket{1}_{\hilb{A}^{\prime\prime}} \ket{1}_{\hilb{B}^{\prime\prime}} \otimes \ket{\xi_{-}^{n}} \right)\mathclose{} .
    \end{split}
    \end{equation}
    Furthermore, for $1 < k \leq n$, we can use \cref{eq:single_isom_1_state_switch,eq:single_isom_1_state,eq:single_isom_junk_switch,eq:isom_step,eq:isom_step_obs} to write
    \begin{equation}
    \label{eq:isom_obs_kgt1}
    \begin{split}
        V S_{q}^{(k)} \ket{\psi}
        & = \ket{\Phi^{+}}_{\hilb{A}_{1}^{\prime} \hilb{B}_{1}^{\prime}} \otimes V^{(n)} \dots V^{(2)} S_{q}^{(k)} \mathopen{}\left( \ket{0}_{\hilb{A}_{1}^{\prime\prime}} \ket{0}_{\hilb{B}_{1}^{\prime\prime}} \otimes \ket{\xi_{+}^{1}} + \ket{1}_{\hilb{A}_{1}^{\prime\prime}} \ket{1}_{\hilb{B}_{1}^{\prime\prime}} \otimes \ket{\xi_{-}^{1}} \right) \\
        & = \bigotimes_{j=1}^{k-1} \ket{\Phi^{+}}_{\hilb{A}_{j}^{\prime} \hilb{B}_{j}^{\prime}} \otimes V^{(n)} \dots V^{(k)} S_{q}^{(k)} \mathopen{}\left( \ket{0 \dots 0} \otimes \ket{\xi_{+}^{k-1}} + \ket{1 \dots 1} \otimes \ket{\xi_{-}^{k-1}} \right) \\
        & = \sigma_{q}^{\hilb{B}_{k}^{\prime}} \bigotimes_{j=1}^{n} \ket{\Phi^{+}}_{\hilb{A}_{j}^{\prime} \hilb{B}_{j}^{\prime}} \otimes \left( \ket{0}_{\hilb{A}^{\prime\prime}} \ket{0}_{\hilb{B}^{\prime\prime}} \otimes \ket{\xi_{+}^{n}} + (-1)^{[q=3]} \ket{1}_{\hilb{A}^{\prime\prime}} \ket{1}_{\hilb{B}^{\prime\prime}} \otimes \ket{\xi_{-}^{n}} \right) \\
        & = \sigma_{q}^{\hilb{B}_{k}^{\prime}} \bigotimes_{j=1}^{n} \ket{\Phi^{+}}_{\hilb{A}_{j}^{\prime} \hilb{B}_{j}^{\prime}} \otimes \sigma_{3[q=3]}^{\hilb{B}^{\prime\prime}} \mathopen{}\left( \ket{0}_{\hilb{A}^{\prime\prime}} \ket{0}_{\hilb{B}^{\prime\prime}} \otimes \ket{\xi_{+}^{n}} + \ket{1}_{\hilb{A}^{\prime\prime}} \ket{1}_{\hilb{B}^{\prime\prime}} \otimes \ket{\xi_{-}^{n}} \right)\mathclose{} .
    \end{split}
    \end{equation}

    Together, \cref{eq:isom_state_result,eq:isom_obs_keq1,eq:isom_obs_kgt1} have the desired form by taking $\ket{\xi_{0}} = \ket{\xi_{+}^{n}}$ and $\ket{\xi_{1}} = \ket{\xi_{-}^{n}}$.
    These $\ket{\xi_{0}}$ and $\ket{\xi_{1}}$ have the desired properties: the observables $S_{q}^{(j)}$ are not present in their definition given in \cref{eq:iter_junk_sketch}, and they satisfy $\braket{\xi_{+}}{\xi_{+}} + \braket{\xi_{-}}{\xi_{-}} = 1$ due to \cref{eq:isom_state_result} together with the fact that the isometry $V$ preserves inner products.
\end{proof}

\Cref{thm:single_observable_isometry} allows us to certify the action of one Pauli operator at a time.
To prepare all of Bob's qubits, however, we require that Pauli measurements of all $n$ of Alice's qubits be certified simultaneously.
This is not immediate from \cref{thm:single_observable_isometry} since, after applying one of the physical operators to $\ket{\psi}$, its conclusion says nothing about the action of a second physical operator on the new state, even if the two operators commute.
Using the symmetry properties of pairs of Pauli operators with respect to $\ket{\Phi^{+}}$, this limitation can be overcome.
\begin{lemma}
\label{lem:single_to_multiple_observable}
    Let $\ket{\psi} \in \hilb{A} \otimes \hilb{B}$ and $\ket{\phi} \in \tilde{\hilb{A}} \otimes \tilde{\hilb{B}}$.
    For some $m > 1$, let $A_{j} \colon \hilb{A} \to \hilb{A}$ satisfy $\lVert A_{j} \rVert \leq 1$ and $\tilde{A}_{j} \colon \tilde{\hilb{A}} \to \tilde{\hilb{A}}$ for all $j \in \{1, \dots, m\}$.
    Let linear isometries $V_{\hilb{A}} \colon \hilb{A} \to \tilde{\hilb{A}}$ and $V_{\hilb{B}} \colon \hilb{B} \to \tilde{\hilb{B}}$ defining the local isometry $V \colon \hilb{A} \otimes \hilb{B} \to \tilde{\hilb{A}} \otimes \tilde{\hilb{B}}$ by $V = V_{\hilb{A}} \otimes V_{\hilb{B}}$ be such that, for all $j$,
    \begin{subequations}
    \begin{align}
        \label{eq:isom_state}
        V \ket{\psi} & \approxe{\delta} \ket{\phi} , \\
        \label{eq:isom_obs}
        V A_{j} \ket{\psi} & \approxe{\delta} \tilde{A}_{j} \ket{\phi} .
    \end{align}
    \end{subequations}
    Suppose that, for all $j$, there exist $\tilde{B}_{j} \colon \tilde{\hilb{B}} \to \tilde{\hilb{B}}$ satisfying $\lVert \tilde{B}_{j} \rVert \leq 1$ such that
    \begin{equation}
    \label{eq:ref_stab}
        \tilde{A}_{j} \ket{\phi} = \tilde{B}_{j} \ket{\phi} .
    \end{equation}
    Then, the combined action of all operators satisfies
    \begin{equation}
        V (A_{1} \dots A_{m}) \ket{\psi}
        \approxe{(2m + 1) \delta} \tilde{A}_{1} \dots \tilde{A}_{m} \ket{\phi} .
    \end{equation}
\end{lemma}
\begin{proof}
    See \cref{sec:many_operator_action}.
\end{proof}
\begin{remark}
    In practice, the $A_{j}$ may represent some physical $\pm 1$-outcome observables or orthogonal projections acting on a physical state $\ket{\psi}$.
    In both cases, $\lVert A_{j} \rVert \leq 1$ is automatically satisfied.
    Similarly, the $\tilde{A}_{j}$ are understood to represent reference measurements, and we may take the reference state $\ket{\phi} = \ket{\Phi^{+}}^{\otimes n} \otimes \ket{\xi}$.
\end{remark}

\section{The protocol}
\label{sec:protocol}

We consider a scenario in which Alice is provided with one of $m^{n}$ questions $\bm x \in \mathcal{X} \subset \{1, \dots, m\}^{n}$ and answers with $\bm a \in \{+1, -1\}^{n}$.
Bob, on the other hand, is provided with a question $y$ and answers with $\bm b$, whose form depends on the input $y$.
\Cref{prot:summary} exhibits our process for the preparation of $n$ qubits on Bob's side, which together are to act as the initial state in a VBQC protocol.
We soon make all definitions required for this, but first let us introduce some intuition.

\begin{protocol}[htb]
    \caption{
        A protocol that performs the blind preparation of $n$ qubits in appropriate bases on the server-side subsystem, followed by VBQC.
    }
    \label{prot:summary}
    \rule{\linewidth}{0.08em}

    The number of qubits to be prepared is $n \in \mathbb{N}^{*}$.
    A verifier first performs the initial setup of the protocol as follows:
    \begin{enumerate}
        \item The verifier chooses a random set of ``special'' questions $\mathcal{S} \subset \{1, \dots, 5\}^{n}$ with polynomial cardinality $\lvert \mathcal{S} \rvert = \poly(n)$, each element of which represents an $n$-tuple of bases in which $n$ qubits may be prepared.

        \item The verifier expands $\mathcal{S}$ to the full set of input questions for Alice $\mathcal{X} \subset \{1, \dots 5\}^{n}$ (which has cardinality $\lvert \mathcal{X} \rvert = \poly(n)$ by construction) as in \cref{eq:alice_question_set}.
    \end{enumerate}
    The verifier then performs the following subprotocols:
    \begin{enumerate}
        \item \emph{Self-test}:
        In each self-testing round, the verifier chooses questions $\bm x \in \mathcal{X}$ and $y \in \mathcal{Y} = \{1, \dots, 6\} \cup \{\lozenge, \blacklozenge\}$.
        The verifier sends $\bm x$ to Alice and $y$ to Bob, and receives an answer $\bm a \in \{+, -\}^{n}$ from Alice.
        \begin{enumerate}
            \item If $y \in \{1, \dots, 6\}$, Bob answers with $\bm b \in \{+, -\}^{n}$.
            For all $j \in \{1, \dots n\}$, this contributes to the correlations
            \begin{equation}
                \left\langle A_{\bm x}^{(j)} B_{y}^{(j)} \right\rangle\mathclose{} .
            \end{equation}

            \item \emph{Conjugation}: The question sent to Bob was $y \in \{\lozenge, \blacklozenge\}$.
            \begin{enumerate}
                \item If $y = \lozenge$, Bob answers with $\bm b \in \{1, 2, 3, 4\}^{\left\lfloor \frac{n}{2} \right\rfloor}$.

                \item If $y = \blacklozenge$, Bob answers with $\bm b \in \{1, 2, 3, 4\}^{\left\lceil \frac{n}{2} \right\rceil - 1}$.
            \end{enumerate}
            For all $1 \leq j < n$, letting $k = \left\lceil \frac{j}{2} \right\rceil$, these contribute to the correlations
            \begin{equation}
                \left\langle A_{\bm x}^{(j)} A_{\bm x}^{(j+1)} \Gamma_{b_{k}}^{(j)} \right\rangle .
            \end{equation}
        \end{enumerate}
        By combining the correlations appropriately, the verifier can estimate all Bell expressions of \cref{eq:request_triple_chsh,eq:request_perfect_correlations,eq:request_conjugation_correlations} and check that the experiment satisfies the assumptions of the self-testing statement.

        \item \emph{VBQC}:
        On some round after the desired confidence threshold has been reached, the verifier asks Alice a special question $\bm \chi \in \mathcal{S}$ and then performs an interactive FK-type protocol with the server.
    \end{enumerate}

    \rule{\linewidth}{0.08em}
\end{protocol}

The idea is that Alice and Bob should play $n$ triple CHSH games in order achieve $n$ corresponding maximal Bell violations.
For the $j$th game, Alice receives an input basis $x_{j} \in \{\mathrm{x}, \mathrm{y}, \mathrm{z}\}$ and outputs $a_{j} \in \{+, -\}$.
Bob receives the same input $y \in \{1, \dots, 6\}$ for all $n$ games, and outputs $b_{j} \in \{+, -\}$ for the $j$th game.
To ensure that the players cannot cheat by sharing fewer than $n$ Bell states while simultaneously being able to satisfy every test we ask of them [we go on in \cref{thm:protocol_isometry} of \cref{sec:operator_rels} to find that satisfying the Bell expressions of \cref{eq:request_triple_chsh,eq:request_perfect_correlations,eq:request_conjugation_correlations} from \cref{sec:requested_observations} are sufficient tests for the VBQC application we are considering], we must examine the experimental outcomes in such a way that we can detect that they originated from $n$ independent games.
Intuitively this is possible since (round and communication complexities aside) the overall scenario is identical with one in which, instead of reporting the outcomes for all $n$ games simultaneously, Bob's input also contains a specified index of a single one of the games to play and in each round he reports only a single-bit outcome for this game.
Alice is unable to be sure which of the $n$ games Bob is instructed to play, while Alice is not necessarily instructed to measure in the same basis in all $n$ games, and so it is possible to check independence between all games.

It will be necessary in order to perform VBQC that the verifier can also choose to prepare qubits in the eigenbases corresponding to observables $\sigma_{\mathrm{x} + \mathrm{y}}$ and $\sigma_{\mathrm{x} - \mathrm{y}}$.
For this purpose, Alice should accept input questions $x_{j} \in \{\mathrm{x} + \mathrm{y}, \mathrm{x} - \mathrm{y}\}$ to each game, in addition to those already stated.
That is, we take $m = 5$ in our protocol, with the five values $1, \dots, 5$ forming Alice's inputs used interchangeably to denote input bases $\mathrm{x}$, $\mathrm{y}$, $\mathrm{z}$, $\mathrm{x} + \mathrm{y}$, and $\mathrm{x} - \mathrm{y}$.
At positions where $\mathrm{x} + \mathrm{y}$ or $\mathrm{x} - \mathrm{y}$ are chosen in Alice's input, we enforce that relevant perfect correlations between Alice and Bob are observed.
This forces the untrusted operations for both inputs $\mathrm{x} \pm \mathrm{y}$ to act consistently with the correct combination of untrusted operations for inputs $\mathrm{x}$ and $\mathrm{y}$.

We would like to restrict the possible questions for Alice to a subset $\mathcal{X} \subset \{1, \dots, m\}^{n}$ whose cardinality in the worst case scales polynomially in $n$, which is required in order to satisfy \cref{itm:question_size} discussed in \cref{sec:efficient_testing}.
This must be done in such a way that the possibility of cheating does not arise from possible correlations between different positions of Alice's inputs.
It turns out that it is sufficient to keep only inputs with the following two constraints in mind:
\begin{enumerate}
    \item Given any pair of positions, there are inputs with all pairs of values at those positions.
    \item At any given position, there are inputs taking all possible values.
\end{enumerate}
The precise formulation of Alice's question set $\mathcal{X}$ is the subject of \cref{sec:alice_questions}.

While we attempt to ensure the independence of the outcomes of $n$ triple CHSH games, a certain dependence between the measurements used in the different games is desirable.
Namely, in order to satisfy \cref{def:complex_self-testing} of self-testing with complex measurements, it is required that any possible complex conjugation of measurement operators only ever applies simultaneously to measurements at every position.
This \emph{global} complex conjugation has been achieved previously \cite{bowles2018self,coladangelo2019verifier} (and as we do also) by allowing Bob to accept an additional two inputs $y \in \{\lozenge, \blacklozenge\}$, each representing Bell measurements being performed on many pairs of qubits.

Relations that we derive between untrusted observables from the observation of requested Bell values and correlations allow us to make our desired self-testing-based statement for remote state preparation.
After exhibiting it here, we proceed to expand on the details of our discussion thus far.
\begin{theorem}
\label{thm:state_preparation}
    There exists a function $\tau \colon \mathbb{R}_{\geq 0} \times \mathbb{N}^{*} \to \mathbb{R}_{\geq 0}$ satisfying $\tau(0, n) = 0$ for all $n$ such that the following holds.
    Let $n \in \mathbb{N}^{*}$ be the number of qubits to prepare, let $\ket{\psi}$ be a state in $\hilb{A} \otimes \hilb{B}$.
    Choose a set of special questions $\mathcal{S} \subset \{\mathrm{x}, \mathrm{y}, \mathrm{z}, \mathrm{x} + \mathrm{y}, \mathrm{x} - \mathrm{y}\}^{n}$ and let $\Pi_{\bm a \mid \bm \chi}^{\hilb{A}}$ denote the physical projectors on $\hilb{A}$ corresponding to Alice answering with $\bm a \in \{+, -\}^{n}$ upon being asked $\bm \chi \in \mathcal{S}$.
    Define the tensor products of qubits
    \begin{subequations}
    \begin{align}
        \ket{e_{\bm a \mid \bm \chi}} & = \bigotimes_{j=1}^{n} \ket{\sigma_{\chi_{j}}^{a_{j}}} , \\
        \ket{e_{\bm a \mid \bm \chi}^{*}} & = \bigotimes_{j=1}^{n} \ket*{\sigma_{\chi_{j}}^{a_{j} (-1)^{[\chi_{j} = \mathrm{z}]}}} ,
    \end{align}
    \end{subequations}
    and denote the reduced physical states on Bob's subsystem after Alice's possible measurements by
    \begin{equation}
        \rho_{\hilb{B}}^{\bm a \mid \bm \chi}
        = \tr_{\hilb{A}} \mathopen{}\left(
        \frac{\Pi_{\bm a \mid \bm \chi}^{\hilb{A}} \proj{\psi} \Pi_{\bm a \mid \bm \chi}^{\hilb{A}}}{\bra{\psi} \Pi_{\bm a \mid \bm \chi}^{\hilb{A}} \ket{\psi}}
        \right)\mathclose{} .
    \end{equation}
    Suppose that all requested Bell inequalities and correlations are $\varepsilon$-approximately saturated (for some $\varepsilon \geq 0$) as in \cref{eq:request_triple_chsh,eq:request_perfect_correlations,eq:request_conjugation_correlations} of \cref{sec:requested_observations} and let $V_{\hilb{B}} \colon \hilb{B} \to \hilb{B}^{\prime} \otimes \hilb{B}^{\prime\prime} \otimes \hilb{B}$ be the isometry thus constructed from Bob's measurement operators as in \cref{thm:protocol_isometry}.
    Then there exist subnormalized density operators $\beta_{0}$ and $\beta_{1}$ on $\hilb{B}$ satisfying $\tr(\beta_{0}) + \tr(\beta_{1}) = 1$ such that, with probability at least $1 - 4 \tau(\varepsilon, n)$ over all possible answers $\bm a$ given any special question $\bm \chi$, we have
    \begin{equation}
    \begin{split}
        \frac{1}{2} \Bigl\lVert V_{\hilb{B}} \rho_{\hilb{B}}^{\bm a \mid \bm \chi} V_{\hilb{B}}^{\dagger}
        - \Bigl(
        & \proj{e_{\bm a \mid \bm \chi}} \otimes \proj{0} \otimes \beta_{0} \\
        & + \proj{e_{\bm a \mid \bm \chi}^{*}} \otimes \proj{1} \otimes \beta_{1}
        \Bigr)\mathclose{} \Bigr\rVert_{1}
        \leq \tau(\varepsilon, n) .
    \end{split}
    \end{equation}
\end{theorem}
\begin{remark}
    The robustness of \cref{thm:single_observable_isometry} expected using the standard techniques mentioned is inherited here to yield $\tau(\varepsilon, n) = \bigO(\varepsilon^{1/3} n^{4/3})$ as $\varepsilon \to 0$ or $n \to \infty$.
\end{remark}
\begin{proof}
    See \cref{sec:state_preparation_proof}.
    Makes use of \cref{cor:protocol_isometry,thm:robust_prob}.
\end{proof}

That all the special questions $\mathcal{S}$ correspond to valid input states to the FK protocol is easily achievable, since the only requirement is that enough dummy qubits (those in the $\sigma_{\mathrm{z}}$ basis) are prepared from which to create enough traps.
For example, using the ``dotted triple graph'' resource state construction \cite{kashefi2017optimised}, the number of additional physical qubits needed to be prepared for verification is linear in the size of the desired computation (the number of qubits needed without any verification).
Thus, to achieve a given level of verifiability, one only needs that a constant fraction of the qubits prepared are dummies.
This can clearly be achieved by instead preparing a large enough constant multiple of the number of qubits and then discarding some of those that are not prepared as dummies.
The probability that a special question still corresponds to too few dummy qubits is exponentially small.

\subsection{Alice's question set}
\label{sec:alice_questions}

We construct the specific subset of questions $\mathcal{X}$ based on a set of possible special input questions $\mathcal{S} \subset \{1, \dots, m\}^{n}$, each element of which is chosen to correspond to preparation bases desired by the verifier.
Throughout, we take addition of questions to be performed (componentwise for strings) modulo $m$.

We first define the set
\begin{equation}
\label{eq:alice_question_set_base}
    \mathcal{D} = \bigl\{ k \bm{e}_{i}^{n} + l \bm{e}_{j}^{n} \bigm| k,l \in \{0, \dots, m - 1\} \text{ and } 1 \leq i < j \leq n \bigr\}\mathclose{} ,
\end{equation}
where $\bm{e}_{i}^{n} = (\delta_{ij})_{j=1}^{n}$ denotes the $i$th standard basis vector with $n$ entries.
The possible input set $\mathcal{X}$ for Alice is selected by expanding some choice of $\mathcal{S}$.
This is performed as follows.
For each $\bm \chi \in \mathcal{S}$, we let
\begin{equation}
    \label{eq:alice_question_set_single}
    \mathcal{X}_{\bm \chi} = \{ \bm \chi + \bm d \mid \bm d \in \mathcal{D} \} .
\end{equation}
Then, we combine all such questions to form
\begin{equation}
\label{eq:alice_question_set}
    \mathcal{X} = \bigcup_{\bm \chi \in \mathcal{S}} \mathcal{X}_{\bm \chi} .
\end{equation}
For example, if $m = 3$, $n = 3$, and $\mathcal{S} = \{\bm \chi\}$, where $\bm \chi = (1,1,1)$, then $\mathcal{X}_{\bm \chi}$ (and $\mathcal{X}$ in this case) has elements
\begin{align*}
    (1,1,1) , \\
    (2,1,1) , && (3,1,1) , \\
    (1,2,1) , && (1,3,1) , \\
    (1,1,2) , && (1,1,3) , \\
    (2,2,1) , && (3,2,1) , && (2,3,1) , && (3,3,1) , \\
    (2,1,2) , && (3,1,2) , && (2,1,3) , && (3,1,3) , \\
    (1,2,2) , && (1,3,2) , && (1,2,3) , && (1,3,3) .
\end{align*}
Each expanded special set has cardinality $\lvert \mathcal{X}_{\bm \chi} \rvert = 1 + (m-1) n + \frac{1}{2} (m-1)^{2} n (n-1)$, and so the total number of questions for Alice is bounded by
\begin{equation}
\label{eq:alice_question_cardinality_bound}
    \lvert \mathcal{X} \rvert \leq \lvert \mathcal{S} \rvert \mathopen{}\left[1 + (m-1) n + \frac{1}{2} (m-1)^{2} n (n-1) \right]
\end{equation}
(quadratic in the number of qubits if the number of special questions is taken to be constant).
We can thus choose $\mathcal{S}$ to be such that $\lvert \mathcal{S} \rvert = \poly(n)$, so that indeed $\lvert \mathcal{X} \rvert = \poly(n)$.
The sets $\mathcal{X}_{\bm \chi}$ are constructed such that starting with any particular special question, for any pair of positions, every possible pair of values from $\{1, \dots, m\}$ appears.
Additionally, at any position, every possible value appears.
We formalize this here.
\begin{lemma}
\label{lem:alice_question_properties_single}
    Fix $m \geq 1$ and let $\bm \chi \in \mathcal{S} \subset \{1, \dots, m\}^{n}$ for some $n \geq 1$.
    Define $\mathcal{X}_{\bm \chi} \subset \{1, \dots, m\}^{n}$ as in \cref{eq:alice_question_set_single} and $\mathcal{X}$ as in \cref{eq:alice_question_set}.
    The following properties then hold:
    \begin{enumerate}
        \item If $n > 1$ then, for any $1 \leq i < j \leq n$ and $q, r \in \{1, \dots, m\}$, there exists $\bm x \in \mathcal{X}_{\bm \chi} \subset \mathcal{X}$ such that $x_{i} = q$, $x_{j} = r$, and, moreover, $x_{k} = \chi_{k}$ for all $k \in \{1, \dots, n\} \setminus \{i, j\}$.

        \item For any $1 \leq i \leq n$ and $q \in \{1, \dots, m\}$, there exists $\bm x \in \mathcal{X}_{\bm \chi} \subset \mathcal{X}$ such that $x_{i} = q$ and, moreover, $x_{k} = \chi_{k}$ for all $k \neq i$.
    \end{enumerate}
    Furthermore, $\lvert \mathcal{X}_{\bm \chi} \rvert = \bigO(n^{2})$.
\end{lemma}
\begin{proof}
    For the first property, take $\bm x = \bm \chi + (q - \chi_{i}) \bm{e}_{i}^{n} + (r - \chi_{j}) \bm{e}_{j}^{n}$.
    The second property is implied by the first for $n > 1$ using the choice $v = \chi_{j}$, and for $n = 1$ we can take $\bm x = \bm \chi + (q - \chi_{1}) \bm{e}_{1}^{n}$.
    Let $\mathcal{D}$ be defined as in \cref{eq:alice_question_set_base}.
    The cardinality $\lvert \mathcal{X}_{\bm \chi} \rvert = \lvert \mathcal{D} \rvert$ since each element of $\mathcal{D}$ is simply shifted by addition with the fixed $\bm \chi$ to form $\mathcal{X}_{\bm \chi}$.
    Finally, $\lvert \mathcal{D} \rvert = 1 + (m-1) n + \frac{1}{2} (m-1)^{2} n (n-1)$.
\end{proof}
\begin{remark}
    It is clear from \cref{eq:alice_question_set} that $\lvert \mathcal{X} \rvert \leq \lvert \mathcal{S} \rvert \lvert \mathcal{X}_{\bm \chi} \rvert = \lvert \mathcal{S} \rvert \lvert \mathcal{D} \rvert$.
\end{remark}

The set $\mathcal{X}$ of \cref{eq:alice_question_set} is generated from the many questions in $\mathcal{S}$ (rather than just from a single choice of preparation question) so that Alice has a low chance of guessing which $\bm \chi$ is used for the computation from her knowledge of the structure of $\mathcal{X}$ (the elements of which she can deduce from the questions asked of her during self-test rounds of \cref{prot:summary}).

We now present another construction of the same set $\mathcal{X}$, which will be useful in \cref{sec:requested_observations} to define the Bell expressions we must consider.
Recall the notation $\bm{x}_{j} = (x_{1}, \dots, x_{j-1}, x_{j+1}, \dots, x_{n})$ for a string $\bm x = (x_{1}, \dots, x_{n})$ excluding the $j$th position.
For each $\bm \chi \in \mathcal{S}$ and $1 \leq j \leq n$, define
\begin{equation}
\label{eq:shortened_questions_single}
    \mathcal{R}_{\bm \chi}^{(j)} = \bigl\{ \bm{\chi}_{j} + k \bm{e}_{i}^{n-1} \bigm| 0 \leq k \leq m-1 \text{ and } 1 \leq i \leq n-1 \bigr\}
\end{equation}
and then use this to define
\begin{equation}
\label{eq:shortened_questions}
    \mathcal{R}^{(j)} = \bigcup_{\bm \chi \in \mathcal{S}} \mathcal{R}_{\bm \chi}^{(j)} .
\end{equation}
These satisfy $\left\lvert \mathcal{R}_{\bm \chi}^{(j)} \right\rvert = 1 + (m - 1) n$ for all $\bm \chi \in \mathcal{S}$ and thus
\begin{equation}
\label{eq:alice_remaining_questions_cardinality_bound}
    \left\lvert \mathcal{R}^{(j)} \right\rvert \leq \lvert \mathcal{S} \rvert [1 + (m - 1) n] .
\end{equation}
Define further
\begin{equation}
\label{eq:pos_j_alice_inputs}
\begin{split}
    \mathcal{X}^{(j)} = \bigl\{
    & (r_{1}, \dots, r_{j-1}, q, r_{j}, \dots, r_{n-1}) \\
    & \bigm| 0 \leq q \leq m-1 \text{ and } \bm r \in \mathcal{R}^{(j)} \bigr\} .
\end{split}
\end{equation}
We now have the notation to reconstruct $\mathcal{X}$ as follows.
\begin{lemma}
\label{lem:alt_alice_question_set}
    The set $\mathcal{X}$ of \cref{eq:alice_question_set} can alternatively be written as
    \begin{equation}
        \mathcal{X} = \bigcup_{j=1}^{n} \mathcal{X}^{(j)} .
    \end{equation}
\end{lemma}
\begin{proof}
    We first show that $\mathcal{X}^{(j)} \subset \mathcal{X}$ for all $1 \leq j \leq n$.
    Let $0 \leq q \leq m-1$ and $\bm r \in \mathcal{R}^{(j)}$.
    Then there exists $\bm \chi \in \mathcal{S}$ such that $\bm r \in \mathcal{R}_{\bm \chi}^{(j)}$.
    Thus, there exists $0 \leq k \leq m-1$ and $1 \leq i \leq n$ such that $\bm r = \bm{\chi}_{j} + k \bm{e}_{i}^{n-1}$.
    We can therefore write for some $l \neq j$ that
    \begin{equation}
    \begin{split}
        & (r_{1}, \dots, r_{j-1}, q, r_{j}, \dots, r_{n-1}) \\
        & \quad = \bm{\chi} + (q - \chi_{j}) \bm{e}_{j}^{n} + k \bm{e}_{l}^{n}
        \in \mathcal{X}_{\bm \chi} \subset \mathcal{X} .
    \end{split}
    \end{equation}

    For the reverse inclusion, take any $\bm x \in \mathcal{X}$.
    Then there exists $\bm \chi \in \mathcal{S}$ such that $\bm x \in \mathcal{X}_{\bm \chi}$.
    For some $k,l \in \{0, \dots, m-1\}$ and $i,j \in \{1, \dots, n\}$ satisfying $i < j$, we thus have $\bm x = \bm \chi + k \bm{e}_{i}^{n} + l \bm{e}_{j}^{n}$.
    Choosing $q = \chi_{j} + l$ and $\bm r = \bm{\chi}_{j} + k \bm{e}_{i}^{n-1}$ gives $\bm x = (r_{1}, \dots, r_{j-1}, q, r_{j}, \dots, r_{n-1})$.
    Finally, since $\bm r \in \mathcal{R}_{\bm \chi}^{(j)} \subset \mathcal{R}^{(j)}$, we have $\bm x \in \mathcal{X}^{(j)}$.
\end{proof}
A useful property of certain input questions for switching the local subsystem on which pairs of observables act in \cref{sec:individual_relations} is given by the following.
\begin{lemma}
\label{lem:double_shortened_question}
    Let $\bm \chi \in \mathcal{S}$, let $1 \leq i < j \leq n$, and let $q,r \in \{1, \dots, m\}$.
    Suppose that $\bm x \in \{1, \dots, m\}^{n}$ is defined such that $x_{i} = q$, $x_{j} = r$, and $x_{k} = \chi_{k}$ for all $k \in \{1, \dots, n\} \setminus \{i, j\}$.
    Then $\bm{x}_{i} \in \mathcal{R}_{\bm \chi}^{(i)} \subset \mathcal{R}^{(i)}$ and $\bm{x}_{j} \in \mathcal{R}_{\bm \chi}^{(j)} \subset \mathcal{R}^{(j)}$.
\end{lemma}
\begin{proof}
    It is clear that $\bm{x}_{i} = \bm{\chi}_{i} + (r - \chi_{j}) \bm{e}_{j-1}^{n-1}$ and $\bm{x}_{j} = \bm{\chi}_{j} + (q - \chi_{i}) \bm{e}_{i}^{n-1}$.
\end{proof}

\subsection{Bell value observations}
\label{sec:requested_observations}

Let us model Alice's behavior upon the question $\bm x \in \mathcal{X}$ as projective measurements with projection operators $\Pi_{\bm a \mid \bm x}^{\hilb{A}}$, where $\bm a \in \{+, -\}^{n}$.
Similarly, let us model Bob's behavior upon the question $y \in \{1, \dots, 6\} \cup \{\lozenge, \blacklozenge\}$ using projections $\Pi_{\bm b \mid y}^{\hilb{B}}$.
For $y \in \{1, \dots, 6\}$ Bob answers with $\bm b \in \{+, -\}^{n}$, while for $y = \lozenge$ he answers with $\bm b \in \{1, 2, 3, 4\}^{\left\lfloor \frac{n}{2} \right\rfloor}$ and for $y = \blacklozenge$ with $\bm b \in \{1, 2, 3, 4\}^{\left\lceil \frac{n}{2} \right\rceil - 1}$.

For all questions $\bm x \in \mathcal{X}$ for Alice and $y \in \{1, \dots, 6\}$ for Bob, we define projections corresponding to Alice observing $a$ and Bob observing $b$ at the $j$th positions of their respective answers
\begin{subequations}
\begin{align}
    \label{eq:local_proj_alice}
    \Pi_{a \mid \bm x}^{\hilb{A}, j} & = \sum_{\bm a \mid a_{j} = a} \Pi_{\bm a \mid \bm x}^{\hilb{A}} , \\
    \label{eq:local_proj_bob}
    \Pi_{b \mid y}^{\hilb{B}, j} & = \sum_{\bm b \mid b_{j} = b} \Pi_{\bm b \mid y}^{\hilb{B}} .
\end{align}
\end{subequations}
In the case of $y \in \{\lozenge, \blacklozenge\}$, we similarly define projections corresponding to Bob observing $b \in \{1, 2, 3, 4\}$ at the $k$th position of $\bm b$ as
\begin{subequations}
\label{eq:local_proj_bob_untrusted_bell}
\begin{align}
    \Gamma_{b}^{(2k-1)} & = \sum_{\bm b \mid b_{k} = b} \Pi_{\bm b \mid \lozenge}^{\hilb{B}} , \\
    \Gamma_{b}^{(2k)} & = \sum_{\bm b \mid b_{k} = b} \Pi_{\bm b \mid \blacklozenge}^{\hilb{B}} .
\end{align}
\end{subequations}
Notice that, due to the form of $\bm b$ in the cases $y = \lozenge$ and $y = \blacklozenge$, the projections $\Gamma_{b}^{(j)}$ are defined for $1 \leq j < n$.
In the honest case, such projectors will correspond to an outcome $b$ for the Bell measurement of the $j$th and $(j+1)$th qubit pair of Bob's subsystem.
Performing the measurements $\bigl\{ \Gamma_{1}^{(j)}, \Gamma_{2}^{(j)}, \Gamma_{3}^{(j)}, \Gamma_{4}^{(j)} \bigr\}$ for all odd $j$ is equivalent to the original measurement $\bigl\{ \Pi_{\bm b \mid \lozenge}^{\hilb{B}} \bigr\}_{\bm b}$, and for all even $j$ is equivalent to the original measurement $\bigl\{ \Pi_{\bm b \mid \blacklozenge}^{\hilb{B}} \bigr\}_{\bm b}$.

We may now define corresponding $\pm 1$-outcome observables for each input $\bm x \in \mathcal{X}$ for Alice and $y \in \{1, \dots, 6\}$ for Bob, and each position $j \in \{1, \dots, n\}$, as
\begin{subequations}
\label{eq:model_observables}
\begin{align}
    \label{eq:model_observables_alice}
    A_{\bm x}^{(j)} &= \Pi_{+ \mid \bm x}^{\hilb{A},j} - \Pi_{- \mid \bm x}^{\hilb{A},j} , \\
    B_{y}^{(j)} &= \Pi_{+ \mid y}^{\hilb{B},j} - \Pi_{- \mid y}^{\hilb{B},j} .
\end{align}
\end{subequations}
Alice's observables commute with Bob's observables since they are defined on separate subsystems.
Moreover, for any questions $\bm x$ and $y$, and any positions $j$ and $k$, we have commutation relations
\begin{subequations}
\begin{align}
    \label{eq:obs_commute_alice}
    \left[ A_{\bm x}^{(j)}, A_{\bm x}^{(k)} \right] & = 0 , \\
    \left[ B_{y}^{(j)}, B_{y}^{(k)} \right] & = 0 .
\end{align}
\end{subequations}
Measuring $A_{\bm x}^{(j)}$ and $B_{y}^{(j)}$ for all $j$ is equivalent to performing the original measurements $\bigl\{ \Pi_{\bm a \mid \bm x}^{\hilb{A}} \bigr\}_{\bm a}$ and $\bigl\{ \Pi_{\bm b \mid y}^{\hilb{B}} \bigr\}_{\bm b}$, respectively.

As mentioned previously, we have that $m = 5$ for constructing $\mathcal{X}$ in our protocol.
It will be convenient to assign some alternative labels to the observables defined in \cref{eq:model_observables}.
We sometimes use notation defined by
\begin{equation}
    A_{x_{j}, \bm{x}_{j}}^{(j)} = A_{\bm x}^{(j)}
\end{equation}
for all $1 \leq j \leq n$ and $\bm x \in \mathcal{X}$.
In particular, at any position $j$, we can write labels for the observables of Alice associated with all input values $1 \leq q \leq 5$ at position $j$ and a fixed question at all other positions
\begin{equation}
\label{eq:alice_isom_obs}
    A_{q, \bm{x}_{j}}^{(j)} ,
\end{equation}
provided that this fixed question $\bm x \in \mathcal{X}^{(j)}$.
All such observables are well defined due to \cref{eq:pos_j_alice_inputs,lem:alt_alice_question_set}.

We define $D_{\mathrmpair{z}{x}}^{(j)}$, $E_{\mathrmpair{z}{x}}^{(j)}$, $D_{\mathrmpair{z}{y}}^{(j)}$, $E_{\mathrmpair{z}{y}}^{(j)}$, $D_{\mathrmpair{x}{y}}^{(j)}$, and $E_{\mathrmpair{x}{y}}^{(j)}$ for all $1 \leq j \leq n$ on Bob's side such that
\begin{equation}
    B_{y}^{(j)} =
    \begin{cases}
        D_{\mathrmpair{z}{x}}^{(j)} & \text{if $y = 1$,} \\
        E_{\mathrmpair{z}{x}}^{(j)} & \text{if $y = 2$,} \\
        D_{\mathrmpair{z}{y}}^{(j)} & \text{if $y = 3$,} \\
        E_{\mathrmpair{z}{y}}^{(j)} & \text{if $y = 4$,} \\
        D_{\mathrmpair{x}{y}}^{(j)} & \text{if $y = 5$,} \\
        E_{\mathrmpair{x}{y}}^{(j)} & \text{if $y = 6$.}
    \end{cases}
\end{equation}
We also define on Bob's side the (not necessarily unitary) combinations of these observables
\begin{subequations}
\label{eq:bob_untrusted_pauli}
\begin{align}
    \label{eq:bob_untrusted_pauli_x_xy}
    X_{\mathrmpair{x}{y}}^{(j)} & = \frac{D_{\mathrmpair{x}{y}}^{(j)} + E_{\mathrmpair{x}{y}}^{(j)}}{\sqrt{2}} , \\
    \label{eq:bob_untrusted_pauli_y_xy}
    Y_{\mathrmpair{x}{y}}^{(j)} & = \frac{D_{\mathrmpair{x}{y}}^{(j)} - E_{\mathrmpair{x}{y}}^{(j)}}{\sqrt{2}} , \\
    Z_{\mathrmpair{z}{x}}^{(j)} & = \frac{D_{\mathrmpair{z}{x}}^{(j)} + E_{\mathrmpair{z}{x}}^{(j)}}{\sqrt{2}} , \\
    \label{eq:bob_untrusted_pauli_x_zx}
    X_{\mathrmpair{z}{x}}^{(j)} & = \frac{D_{\mathrmpair{z}{x}}^{(j)} - E_{\mathrmpair{z}{x}}^{(j)}}{\sqrt{2}} , \\
    Y_{\mathrmpair{z}{y}}^{(j)} & = \frac{D_{\mathrmpair{z}{y}}^{(j)} - E_{\mathrmpair{z}{y}}^{(j)}}{\sqrt{2}} , \\
    Z_{\mathrmpair{z}{y}}^{(j)} & = \frac{D_{\mathrmpair{z}{y}}^{(j)} + E_{\mathrmpair{z}{y}}^{(j)}}{\sqrt{2}} .
\end{align}
\end{subequations}
Finally, we define on Bob's side the observables
\begin{subequations}
\label{eq:bob_isom_obs}
\begin{align}
    \label{eq:bob_isom_obs_x}
    Q_{1}^{(j)} & = X_{\mathrmpair{x}{y}}^{(j)} , \\
    \label{eq:bob_isom_obs_y}
    Q_{2}^{(j)} & = Y_{\mathrmpair{x}{y}}^{(j)} , \\
    \label{eq:bob_isom_obs_z}
    Q_{3}^{(j)} & = Z_{\mathrmpair{z}{x}}^{(j)} , \\
    Q_{4}^{(j)} & = D_{\mathrmpair{x}{y}}^{(j)} , \\
    Q_{5}^{(j)} & = E_{\mathrmpair{x}{y}}^{(j)} .
\end{align}
\end{subequations}

It is from (regularized versions of) the observables defined in \cref{eq:alice_isom_obs,eq:bob_isom_obs} which, after proving the necessary relations between them, we construct our isometry.
To prove the relations, we request that certain Bell expressions be observed to take maximal values (within some small $\varepsilon$ to account for experimental imperfections).

\subparagraph{Commutation structure.}

In order that the observables satisfy the main state-dependent commutativity and anticommutativity properties required, we request that, at all positions $1 \leq j \leq n$, maximal violations
\begin{equation}
\label{eq:request_triple_chsh}
    \bra{\psi} C_{\bm{x}_{j}}^{(j)} \ket{\psi} \geq 6\sqrt{2} - \frac{\varepsilon}{\sqrt{2}}
\end{equation}
are observed of the triple CHSH operators
\begin{equation}
\label{eq:request_triple_chsh_operators}
\begin{split}
    C_{\bm{x}_{j}}^{(j)}
    ={} & A_{3, \bm{x}_{j}}^{(j)} \otimes \left( D_{\mathrmpair{z}{x}}^{(j)} + E_{\mathrmpair{z}{x}}^{(j)} \right) \\
    & + A_{1, \bm{x}_{j}}^{(j)} \otimes \left( D_{\mathrmpair{z}{x}}^{(j)} - E_{\mathrmpair{z}{x}}^{(j)} \right) \\
    & + A_{3, \bm{x}_{j}}^{(j)} \otimes \left( D_{\mathrmpair{z}{y}}^{(j)} + E_{\mathrmpair{z}{y}}^{(j)} \right) \\
    & + A_{2, \bm{x}_{j}}^{(j)} \otimes \left( D_{\mathrmpair{z}{y}}^{(j)} - E_{\mathrmpair{z}{y}}^{(j)} \right) \\
    & + A_{1, \bm{x}_{j}}^{(j)} \otimes \left( D_{\mathrmpair{x}{y}}^{(j)} + E_{\mathrmpair{x}{y}}^{(j)} \right) \\
    & + A_{2, \bm{x}_{j}}^{(j)} \otimes \left( D_{\mathrmpair{x}{y}}^{(j)} - E_{\mathrmpair{x}{y}}^{(j)} \right)\mathclose{}
\end{split}
\end{equation}
for all inputs for Alice at other positions $\bm{x}_{j} \in \mathcal{R}^{(j)}$.

The expectation values of all terms appearing in \cref{eq:request_triple_chsh_operators} for $C_{\bm{x}_{j}}^{(j)}$ are derivable from the observed statistics.
For any $1 \leq j \leq n$, since $\bm{x}_{j} \in \mathcal{R}^{(j)}$, the observables for Alice all correspond to questions in $\mathcal{X}^{(j)} \subset \mathcal{X}$ by its definition in \cref{eq:pos_j_alice_inputs}, and so are well defined.
Moreover, for each $j$ there are $12 \mathopen{}\left\lvert \mathcal{R}^{(j)} \right\rvert$ correlations.
Due to the choice $\lvert \mathcal{S} \rvert = \poly(n)$ of special questions, inserting \cref{eq:alice_remaining_questions_cardinality_bound} then gives a total of at most $12 n [1 + 4n] \lvert \mathcal{S} \rvert = \poly(n)$ correlations needed to verify these requests.

\subparagraph{Additional bases.}

So that qubits may be prepared in the additional bases necessary for VBQC, we incorporate the additional observables of Alice corresponding to inputs where $x_{j} = \mathrm{x} + \mathrm{y}$ and $x_{j} = \mathrm{x} - \mathrm{y}$.
We request that, for all special questions $\bm \chi \in \mathcal{S}$, perfect correlations
\begin{subequations}
\label{eq:request_perfect_correlations}
\begin{align}
    \bra{\psi} A_{4, \bm{\chi}_{j}}^{(j)} \otimes D_{\mathrmpair{x}{y}}^{(j)} \ket{\psi} & \geq 1 - \varepsilon , \\
    \bra{\psi} A_{5, \bm{\chi}_{j}}^{(j)} \otimes E_{\mathrmpair{x}{y}}^{(j)} \ket{\psi} & \geq 1 - \varepsilon
\end{align}
\end{subequations}
are observed for all $1 \leq j \leq n$.
This serves to ensure that (in the case of any special question) the untrusted operations corresponding to inputs where $x_{j} = \mathrm{x} \pm \mathrm{y}$ are consistent with the correct combinations of those for the separate inputs $x_{j} = \mathrm{x}$ and $x_{j} = \mathrm{y}$.

All of Alice's observables appearing in \cref{eq:request_perfect_correlations} are well defined (see \cref{lem:alice_question_properties_single}).
There are at most $2n \lvert \mathcal{S} \rvert = \poly(n)$ correlations needed to verify these requests.

\subparagraph{Complex conjugation.}

To ensure that possible complex conjugation of measurement operators may only occur across all positions simultaneously, we enforce certain correlations that include (commuting) pairs of Alice's observables.
This has also been achieved previously by similar methods \cite{bowles2018self,coladangelo2019verifier}.
Fix any choice of $\bm \chi \in \mathcal{S}$.
We request for all $1 \leq j < n$ and $q \in \{1, 2, 3\}$ that
\begin{equation}
\label{eq:request_conjugation_correlations}
\begin{split}
    \bra{\psi} A_{\bm w}^{(j)} A_{\bm w}^{(j+1)}
    \otimes \mathopen{}\Bigl[ &
    (-1)^{[q=2]} \Gamma_{1}^{(j)}
    + (-1)^{[q=1]} \Gamma_{2}^{(j)} \\
    & + (-1)^{[q=3]} \Gamma_{3}^{(j)}
    - \Gamma_{4}^{(j)} \Bigr] \ket{\psi}
    \geq 1 - \frac{\varepsilon}{2} ,
\end{split}
\end{equation}
where $\bm w = (\chi_{1}, \dots, \chi_{j-1}, q, q, \chi_{j+2}, \dots, \chi_{n})$.

All of Alice's observables appearing in \cref{eq:request_conjugation_correlations} are well defined since $\bm w \in \mathcal{X}_{\bm \chi} \subset \mathcal{X}$, as can be seen from \cref{eq:alice_question_set_single,eq:alice_question_set}.
The observables $A_{\bm w}^{(j)}$ and $A_{\bm w}^{(j+1)}$ are chosen so that they both correspond to the same question $\bm w$, and thus commute with one another.
This ensures that the correlations being considered are derivable from observed statistics, since $A_{\bm w}^{(j)} A_{\bm w}^{(j+1)}$ is then a valid observable corresponding to restricting measurement to the product of entries $a_{j} a_{j+1} \in \{+, -\}$ in the outcome $\bm a$.
There are a total of $\bigO(n)$ probabilities involved in verifying these requests.

\subparagraph{Classical processing.}

Overall, the total amount of classical processing required to check all of the Bell values and correlations requested in this section grows at worst as $\bigO(n^{2}) \lvert \mathcal{S} \rvert = \poly(n)$.
It is thus efficient to perform these checks; the desirable \cref{itm:data_processing} discussed in \cref{sec:efficient_testing} is satisfied.

\subsection{Completeness (honest strategy)}
\label{sec:honest_strat}

The ideal values of Bell expressions and correlations given at the end of \cref{sec:requested_observations} can be satisfied using an honest strategy.
We take the shared state to be
\begin{equation}
    \ket{\psi} = \bigotimes_{j=1}^{n} \ket{\Phi^{+}}_{\hilb{A}_{j} \hilb{B}_{j}} ,
\end{equation}
where $\hilb{A}_{j}$ and $\hilb{B}_{j}$ denote the $j$th reference qubit registers of Alice and Bob, respectively.
For all $1 \leq j \leq n$ and all $\bm r$ for which they are defined, we take Alice's observables to be
\begin{subequations}
\begin{align}
    A_{1, \bm r}^{(j)} & = \sigma_{1}^{\hilb{A}_{j}} , \\
    A_{2, \bm r}^{(j)} & = - \sigma_{2}^{\hilb{A}_{j}} , \\
    A_{3, \bm r}^{(j)} & = \sigma_{3}^{\hilb{A}_{j}} , \\
    A_{4, \bm r}^{(j)} & = \sigma_{5}^{\hilb{A}_{j}} , \\
    A_{5, \bm r}^{(j)} & = \sigma_{4}^{\hilb{A}_{j}} ,
\end{align}
\end{subequations}
where $\sigma_{4}$ and $\sigma_{5}$ are defined as stated in \cref{sec:prelims}.
We take Bob's observables for the questions $y \in \{1, \dots, 6\}$ to be
\begin{subequations}
\begin{align}
    D_{k,l}^{(j)} & = \frac{\sigma_{k}^{\hilb{B}_{j}} + \sigma_{l}^{\hilb{B}_{j}}}{\sqrt{2}} , \\
    E_{k,l}^{(j)} & = \frac{\sigma_{k}^{\hilb{B}_{j}} - \sigma_{l}^{\hilb{B}_{j}}}{\sqrt{2}} ,
\end{align}
\end{subequations}
where $k,l \in \{1, 2, 3\}$.
In the case $k = 1$ and $l = 2$, these are $\sigma_{4}$ and $\sigma_{5}$, respectively.
As for the inputs $y \in \{\lozenge, \blacklozenge\}$, the correlations of \cref{eq:request_conjugation_correlations} can be achieved by performing a Bell measurement of pairs of qubits on Bob's subsystem.
The projections $\Gamma_{b}^{(j)}$ denote the projective measurement operator with outcome $b \in \{1, 2, 3, 4\}$ for such a measurement performed on Bob's $j$th and $(j+1)$th qubits.
Specifically, at any position $1 \leq j < n$, we take
\begin{subequations}
\begin{align}
    \Gamma_{1}^{(j)} & = \proj{\Phi^{+}}_{\hilb{B}_{j} \hilb{B}_{j+1}} , \\
    \Gamma_{2}^{(j)} & = \proj{\Phi^{-}}_{\hilb{B}_{j} \hilb{B}_{j+1}} , \\
    \Gamma_{3}^{(j)} & = \proj{\Psi^{+}}_{\hilb{B}_{j} \hilb{B}_{j+1}} , \\
    \Gamma_{4}^{(j)} & = \proj{\Psi^{-}}_{\hilb{B}_{j} \hilb{B}_{j+1}} ,
\end{align}
\end{subequations}
and perform the projective measurement specified by $\bigl\{ \Gamma_{1}^{(j)}, \Gamma_{2}^{(j)}, \Gamma_{3}^{(j)}, \Gamma_{4}^{(j)} \bigr\}$.

\section{Operator relations in the self-test subprotocol}
\label{sec:operator_rels}

In this section, we demonstrate that the values for Bell expressions and correlations requested as part of \cref{prot:summary} (detailed at the end of \cref{sec:requested_observations}) imply the existence of unitary observables satisfying \cref{thm:single_observable_isometry}.
That is, the observed experimental probabilities self-test (in the sense of \cref{def:complex_self-testing}) $n$ Bell states, and also Pauli measurements for Alice.
We go on to state this self-testing result in \cref{thm:protocol_isometry,cor:protocol_isometry}.
The following theorem summarizes the relations derived between observables, which are shown in individual detail subsequently in \cref{sec:individual_relations}.
\begin{theorem}
\label{thm:assumptions_satisfied}
    Suppose that the values of Bell expressions and correlations requested in \cref{sec:requested_observations} are attained (within the tolerance specified in terms of $\varepsilon$).
    Then the $\pm 1$-outcome observables $\regular{Q}_{q}^{(j)}$ (regularized versions of the observables $Q_{q}^{(j)}$) defined by \cref{eq:bob_isom_obs_x,eq:bob_isom_obs_y,eq:bob_isom_obs_z} and, for any choice of special question $\bm \chi \in \mathcal{S} \subset \{1, \dots, 5\}^{n}$, the $\pm 1$-outcome observables $A_{q, \bm{\chi}_{j}}^{(j)}$ defined for $q \in \{1, 2, 3\}$ and $1 \leq j \leq n$ by \cref{eq:alice_isom_obs}, satisfy the assumptions of \cref{thm:single_observable_isometry} with $\eta \leq 21 \sqrt{\varepsilon}$.
\end{theorem}
\begin{proof}
    For any given $\bm \chi \in \mathcal{S}$, take $S_{q}^{(j)} = A_{q, \bm{\chi}_{j}}^{(j)}$ and $T_{q}^{(j)} = \regular{Q}_{q}^{(j)}$ for \cref{thm:single_observable_isometry}.
    \Cref{prop:symmetry} shows \cref{rel:corr}.
    \Cref{rel:comm} on the commutativity of observables is shown by \cref{prop:com_bob,prop:com_alice}.
    \Cref{rel:acomm} on anticommutativity is shown by \cref{prop:acomm_alice,prop:acomm_bob}.
    Finally, \cref{rel:conj} (which ensures global complex conjugation) is shown by \cref{prop:conj_rel}.
    All relations are shown to within a maximum of $21 \sqrt{\varepsilon}$ from the ideal, and thus we may choose $\eta \leq 21 \sqrt{\varepsilon}$.
\end{proof}

Each of the triple CHSH operators of \cref{eq:request_triple_chsh_operators} has an SOS decomposition
\begin{equation}
\begin{split}
    6 \sqrt{2} - C_{\bm{x}_{j}}^{(j)} ={}
    &   \frac{1}{\sqrt{2}} \mathopen{}\left( A_{3, \bm{x}_{j}}^{(j)} - \frac{D_{\mathrmpair{z}{x}}^{(j)} + E_{\mathrmpair{z}{x}}^{(j)}}{\sqrt{2}} \right)^{2} \\
    & + \frac{1}{\sqrt{2}} \mathopen{}\left( A_{1, \bm{x}_{j}}^{(j)} - \frac{D_{\mathrmpair{z}{x}}^{(j)} - E_{\mathrmpair{z}{x}}^{(j)}}{\sqrt{2}} \right)^{2} \\
    & + \frac{1}{\sqrt{2}} \mathopen{}\left( A_{3, \bm{x}_{j}}^{(j)} - \frac{D_{\mathrmpair{z}{y}}^{(j)} + E_{\mathrmpair{z}{y}}^{(j)}}{\sqrt{2}} \right)^{2} \\
    & + \frac{1}{\sqrt{2}} \mathopen{}\left( A_{2, \bm{x}_{j}}^{(j)} - \frac{D_{\mathrmpair{z}{y}}^{(j)} - E_{\mathrmpair{z}{y}}^{(j)}}{\sqrt{2}} \right)^{2} \\
    & + \frac{1}{\sqrt{2}} \mathopen{}\left( A_{1, \bm{x}_{j}}^{(j)} - \frac{D_{\mathrmpair{x}{y}}^{(j)} + E_{\mathrmpair{x}{y}}^{(j)}}{\sqrt{2}} \right)^{2} \\
    & + \frac{1}{\sqrt{2}} \mathopen{}\left( A_{2, \bm{x}_{j}}^{(j)} - \frac{D_{\mathrmpair{x}{y}}^{(j)} - E_{\mathrmpair{x}{y}}^{(j)}}{\sqrt{2}} \right)^{2}
\end{split}
\end{equation}
Thus, for all $j$, \cref{eq:request_triple_chsh} implies (see \cref{sec:sos_decomposition}) for all $\bm{x}_{j} \in \mathcal{R}^{(j)}$ specified that
\begin{subequations}
\label{eq:corr_rel}
\begin{align}
    \label{eq:corr_rel_x_xy}
    \left( A_{1, \bm{x}_{j}}^{(j)} - X_{\mathrmpair{x}{y}}^{(j)} \right) \ket{\psi} & \approxe{\sqrt{\varepsilon}} 0 , \\
    \label{eq:corr_rel_y_xy}
    \left( A_{2, \bm{x}_{j}}^{(j)} - Y_{\mathrmpair{x}{y}}^{(j)} \right) \ket{\psi} & \approxe{\sqrt{\varepsilon}} 0 , \\
    \label{eq:corr_rel_z_zx}
    \left( A_{3, \bm{x}_{j}}^{(j)} - Z_{\mathrmpair{z}{x}}^{(j)} \right) \ket{\psi} & \approxe{\sqrt{\varepsilon}} 0 , \\
    \label{eq:corr_rel_x_zx}
    \left( A_{1, \bm{x}_{j}}^{(j)} - X_{\mathrmpair{z}{x}}^{(j)} \right) \ket{\psi} & \approxe{\sqrt{\varepsilon}} 0 , \\
    \label{eq:corr_rel_y_zy}
    \left( A_{2, \bm{x}_{j}}^{(j)} - Y_{\mathrmpair{z}{y}}^{(j)} \right) \ket{\psi} & \approxe{\sqrt{\varepsilon}} 0 , \\
    \label{eq:corr_rel_z_zy}
    \left( A_{3, \bm{x}_{j}}^{(j)} - Z_{\mathrmpair{z}{y}}^{(j)} \right) \ket{\psi} & \approxe{\sqrt{\varepsilon}} 0 .
\end{align}
\end{subequations}
Due to the definitions made in \cref{eq:bob_isom_obs}, Bob's observables in \cref{eq:corr_rel_x_xy,eq:corr_rel_y_xy,eq:corr_rel_z_zx} may also be replaced with $Q_{1}^{(j)}$, $Q_{2}^{(j)}$, and $Q_{3}^{(j)}$, respectively.
Since Alice's observables in \cref{eq:corr_rel} are unitary, regularization can be applied to Bob's observables (see \cref{sec:regularization}).
This results in expressions involving only unitary operators
\begin{subequations}
\label{eq:corr_rel_regular}
\begin{align}
    \label{eq:corr_rel_x_xy_regular}
    \left( A_{1, \bm{x}_{j}}^{(j)} - \regular{X}_{\mathrmpair{x}{y}}^{(j)} \right) \ket{\psi} & \approxe{2 \sqrt{\varepsilon}} 0 , \\
    \label{eq:corr_rel_y_xy_regular}
    \left( A_{2, \bm{x}_{j}}^{(j)} - \regular{Y}_{\mathrmpair{x}{y}}^{(j)} \right) \ket{\psi} & \approxe{2 \sqrt{\varepsilon}} 0 , \\
    \label{eq:corr_rel_z_zx_regular}
    \left( A_{3, \bm{x}_{j}}^{(j)} - \regular{Z}_{\mathrmpair{z}{x}}^{(j)} \right) \ket{\psi} & \approxe{2 \sqrt{\varepsilon}} 0 , \\
    \label{eq:corr_rel_x_zx_regular}
    \left( A_{1, \bm{x}_{j}}^{(j)} - \regular{X}_{\mathrmpair{z}{x}}^{(j)} \right) \ket{\psi} & \approxe{2 \sqrt{\varepsilon}} 0 , \\
    \label{eq:corr_rel_y_zy_regular}
    \left( A_{2, \bm{x}_{j}}^{(j)} - \regular{Y}_{\mathrmpair{z}{y}}^{(j)} \right) \ket{\psi} & \approxe{2 \sqrt{\varepsilon}} 0 , \\
    \label{eq:corr_rel_z_zy_regular}
    \left( A_{3, \bm{x}_{j}}^{(j)} - \regular{Z}_{\mathrmpair{z}{y}}^{(j)} \right) \ket{\psi} & \approxe{2 \sqrt{\varepsilon}} 0 .
\end{align}
\end{subequations}
Similarly to before, Bob's observables in \cref{eq:corr_rel_x_xy_regular,eq:corr_rel_y_xy_regular,eq:corr_rel_z_zx_regular} may also be replaced with $\regular{Q}_{1}^{(j)}$, $\regular{Q}_{2}^{(j)}$, and $\regular{Q}_{3}^{(j)}$, respectively.

Additionally to the relations required for \cref{thm:single_observable_isometry}, we have from \cref{eq:bob_untrusted_pauli_x_xy,eq:bob_untrusted_pauli_y_xy} that, for all positions $1 \leq j \leq n$,
\begin{equation}
    D_{\mathrmpair{x}{y}}^{(j)} = \frac{X_{\mathrmpair{x}{y}}^{(j)} + Y_{\mathrmpair{x}{y}}^{(j)}}{\sqrt{2}} ,\quad
    E_{\mathrmpair{x}{y}}^{(j)} = \frac{X_{\mathrmpair{x}{y}}^{(j)} - Y_{\mathrmpair{x}{y}}^{(j)}}{\sqrt{2}} .
\end{equation}
The required observation of the correlations in \cref{eq:request_perfect_correlations} then implies (via \cref{lem:state_estimate}) that
\begin{subequations}
\begin{align}
    A_{4, \bm{\chi}_{j}}^{(j)} \ket{\psi}
    & \approxe{\sqrt{2 \varepsilon}} \frac{X_{\mathrmpair{x}{y}}^{(j)} + Y_{\mathrmpair{x}{y}}^{(j)}}{\sqrt{2}} \ket{\psi} , \\
    A_{5, \bm{\chi}_{j}}^{(j)} \ket{\psi}
    & \approxe{\sqrt{2 \varepsilon}} \frac{X_{\mathrmpair{x}{y}}^{(j)} - Y_{\mathrmpair{x}{y}}^{(j)}}{\sqrt{2}} \ket{\psi} .
\end{align}
\end{subequations}
From the relations of \cref{eq:corr_rel_x_xy,eq:corr_rel_y_xy} inferred from the triple CHSH inequalities (in the special case $\bm{x}_{j} = \bm{\chi}_{j}$), we then have for all $\bm \chi \in \mathcal{S}$ that
\begin{subequations}
\label{eq:extra_obs_linear}
\begin{align}
    A_{4, \bm{\chi}_{j}}^{(j)} \ket{\psi}
    & \approxe{2 \sqrt{2 \varepsilon}} \frac{A_{1, \bm{\chi}_{j}}^{(j)} + A_{2, \bm{\chi}_{j}}^{(j)}}{\sqrt{2}} \ket{\psi} , \\
    A_{5, \bm{\chi}_{j}}^{(j)} \ket{\psi}
    & \approxe{2 \sqrt{2 \varepsilon}} \frac{A_{1, \bm{\chi}_{j}}^{(j)} - A_{2, \bm{\chi}_{j}}^{(j)}}{\sqrt{2}} \ket{\psi} .
\end{align}
\end{subequations}
We may now write a version of \cref{thm:single_observable_isometry} for the experimental observations required by our protocol which also incorporates certification of $\sigma_{4}$ and $\sigma_{5}$.
This, along with \cref{lem:single_to_multiple_observable}, will allow us to prove \cref{thm:state_preparation}.
\begin{theorem}
\label{thm:protocol_isometry}
    There exists $\delta(\varepsilon, n) \geq 0$ satisfying $\delta(0, n) = 0$ such that the following holds.
    Suppose that values of Bell expressions and correlations requested in \cref{sec:requested_observations} are attained (within the tolerance specified in terms of $\varepsilon$).
    Then there exist subnormalized $\ket{\xi_{0}}_{\hilb{A} \hilb{B}}$ and $\ket{\xi_{1}}_{\hilb{A} \hilb{B}}$ satisfying $\braket{\xi_{0}}{\xi_{0}}_{\hilb{A} \hilb{B}} + \braket{\xi_{1}}{\xi_{1}}_{\hilb{A} \hilb{B}} = 1$, isometries $V_{\hilb{A}}^{\bm \chi} \colon \hilb{A} \to \hilb{A} \otimes \hilb{A}^{\prime} \otimes \hilb{A}^{\prime\prime}$ for all $\bm \chi \in \mathcal{S}$, and an isometry $V_{\hilb{B}} \colon \hilb{B} \to \hilb{B} \otimes \hilb{B}^{\prime} \otimes \hilb{B}^{\prime\prime}$ such that for all $q \in \{1, \dots, 5\}$, $k \in \{1, \dots, n\}$, and $\bm \chi \in \mathcal{S}$ we have
    \begin{subequations}
    \begin{align}
        V^{\bm \chi} \ket{\psi}_{\hilb{A} \hilb{B}}
        & \approxe{\delta(\varepsilon, n)} \bigotimes_{j=1}^{n} \ket{\Phi^{+}}_{\hilb{A}_{j}^{\prime} \hilb{B}_{j}^{\prime}}
        \otimes \ket{\xi} , \\
        V^{\bm \chi} A_{q, \bm{\chi}_{j}}^{(k)} \ket{\psi}_{\hilb{A} \hilb{B}}
        & \approxe{\delta(\varepsilon, n)} \sigma_{q}^{\hilb{B}_{k}^{\prime}} \bigotimes_{j=1}^{n} \ket{\Phi^{+}}_{\hilb{A}_{j}^{\prime} \hilb{B}_{j}^{\prime}}
        \otimes \sigma_{3[q = 3]}^{\hilb{B}^{\prime\prime}} \ket{\xi} ,
    \end{align}
    \end{subequations}
    where $V^{\bm \chi} = V_{\hilb{A}}^{\bm \chi} \otimes V_{\hilb{B}}$ and the junk state $\ket{\xi} \in \hilb{A} \otimes \hilb{A}^{\prime\prime} \otimes \hilb{B} \otimes \hilb{B}^{\prime\prime}$ is defined as
    \begin{equation}
        \ket{\xi}
        = \ket{0}_{\hilb{A}^{\prime\prime}} \ket{0}_{\hilb{B}^{\prime\prime}} \otimes \ket{\xi_{0}}_{\hilb{A} \hilb{B}}
        + \ket{1}_{\hilb{A}^{\prime\prime}} \ket{1}_{\hilb{B}^{\prime\prime}} \otimes \ket{\xi_{1}}_{\hilb{A} \hilb{B}} .
    \end{equation}
\end{theorem}
\begin{proof}
    For $q \in \{1, 2, 3\}$, \cref{thm:single_observable_isometry} can be applied for each $\bm \chi \in \mathcal{S}$ due to \cref{thm:assumptions_satisfied}, giving an appropriate function $\delta$, junk state $\ket{\xi}$, and isometries $V_{\hilb{A}}^{\bm \chi}$ and $V_{\hilb{B}}$.
    To get the remaining cases $q \in \{4, 5\}$, note the linearity of $V^{\bm \chi}$ and use \cref{eq:extra_obs_linear}.

    The isometry on Bob's subsystem $V_{\hilb{B}}$ given by each application of \cref{thm:single_observable_isometry} to \cref{thm:assumptions_satisfied} (once for each $\bm \chi \in \mathcal{S}$) is that of \cref{fig:swap_kickback_xy_copy}.
    For each $\bm \chi \in \mathcal{S}$ it is constructed from the same set of Bob's observables, as stated in \cref{thm:assumptions_satisfied}.
    Thus, given any choice of measurement strategy for Bob as in \cref{eq:local_proj_bob,eq:local_proj_bob_untrusted_bell}, \cref{thm:single_observable_isometry} guarantees that the isometry $V_{\hilb{B}}$ remains unchanged under all different choices of special question $\bm \chi \in \mathcal{S}$.
    That the same $\ket{\xi}$ is sufficient for all $\bm \chi$ can be seen from the fact that the junk states in \cref{thm:single_observable_isometry} do not depend on any of Alice's observables.
\end{proof}
\begin{remark}
    That $V_{\hilb{B}}$ does not depend on the choice of special question $\bm \chi$ leads to the required \cref{itm:independent_isom} of self-tests applicable to our DIVBQC protocol.
\end{remark}
The local isometries $V^{\bm \chi}$ in \cref{thm:protocol_isometry} are constructed such that the actions of all observables $A_{q, \bm{\chi}_{k}}^{(k)}$ on $\ket{\psi}_{\hilb{A} \hilb{B}}$ under $V^{\bm \chi}$ are shown (in a slightly different form) to be
\begin{equation}
\label{eq:protocol_isometry_single_obs_split}
\begin{split}
    V^{\bm \chi} A_{q, \bm{\chi}_{k}}^{(k)} \ket{\psi}_{\hilb{A} \hilb{B}}
    \approxe{\delta(\varepsilon, n)}
    & \left( \sigma_{q}^{\hilb{B}_{k}^{\prime}} \otimes \proj{0}_{\hilb{B}^{\prime\prime}} + {\sigma_{q}^{\hilb{B}_{k}^{\prime}}}^{*} \otimes \proj{1}_{\hilb{B}^{\prime\prime}} \right) \\
    & \quad \bigotimes_{j=1}^{n} \ket{\Phi^{+}}_{\hilb{A}_{j}^{\prime} \hilb{B}_{j}^{\prime}}
    \otimes \ket{\xi} .
\end{split}
\end{equation}
In particular (by looking at the case $q = \chi_{k}$ for each $k$), it is constructed so that the actions of all $A_{\bm \chi}^{(k)} = A_{\chi_{k}, \bm{\chi}_{k}}^{(k)}$ (corresponding to special questions $\bm \chi \in \mathcal{S}$) are shown.
The following statement is then immediate.
\begin{corollary}
\label{cor:protocol_isometry}
    For some $\gamma(\varepsilon, n) \geq 0$ satisfying $\gamma(0, n) = 0$, the isometries $V^{\bm \chi} = V_{\hilb{A}}^{\bm \chi} \otimes V_{\hilb{B}}$ resulting from \cref{thm:protocol_isometry} act for all $\bm s \in \{0, 1\}^{n}$ as
    \begin{equation}
    \label{eq:protocol_isometry_obs_split}
        V^{\bm \chi} A_{\bm \chi}^{\bm s} \ket{\psi}_{\hilb{A} \hilb{B}}
        \approxe{\gamma(\varepsilon, n)} \left[
        \bigotimes_{j=1}^{n} \mathopen{}\left( \sigma_{\chi_{j}}^{\hilb{B}_{j}^{\prime}} \right)^{s_{j}} \otimes \proj{0}_{\hilb{B}^{\prime\prime}}
        + \bigotimes_{j=1}^{n} \mathopen{}\left( {\sigma_{\chi_{j}}^{\hilb{B}_{j}^{\prime}}}^{*} \right)^{s_{j}} \otimes \proj{1}_{\hilb{B}^{\prime\prime}}
        \right]\mathclose{}
        \bigotimes_{j=1}^{n} \ket{\Phi^{+}}_{\hilb{A}_{j}^{\prime} \hilb{B}_{j}^{\prime}}
        \otimes \ket{\xi} .
    \end{equation}
\end{corollary}
\begin{proof}
    For each $\bm \chi \in \mathcal{S}$, apply \cref{lem:single_to_multiple_observable} to \cref{eq:protocol_isometry_single_obs_split} with $\ket{\phi} = \bigotimes_{j=1}^{n} \ket{\Phi^{+}}_{\hilb{A}_{j}^{\prime} \hilb{B}_{j}^{\prime}} \otimes \ket{\xi}$.
\end{proof}

As discussed after the statement of \cref{thm:single_observable_isometry}, for which \cref{thm:assumptions_satisfied} guarantees $\eta \leq 21 \sqrt{\varepsilon}$, we expect standard techniques to yield $\delta(\varepsilon, n) = \bigO(\sqrt{\varepsilon} n^{2})$ and $\gamma(\varepsilon, n) = \bigO(\sqrt{\varepsilon} n^{2})$ in the previous \cref{thm:protocol_isometry,cor:protocol_isometry}.

\subsection{Individual relations}
\label{sec:individual_relations}

We now show each relation that forms \cref{thm:assumptions_satisfied} in individual detail.
\Cref{prop:symmetry,prop:com_alice,prop:com_bob,prop:acomm_alice,prop:acomm_bob,prop:conj_rel} that we exhibit here were together used to form the proof of \cref{thm:assumptions_satisfied} at the beginning of this section.
We take it as given throughout the remainder of this section that the values of all Bell expressions and correlations requested in \cref{sec:requested_observations} are attained (within the tolerance specified in terms of $\varepsilon$).

\begin{proposition}[Symmetry, Alice and Bob]
\label{prop:symmetry}
    For any $\bm \chi \in \mathcal{S}$,
    \begin{equation}
        \left( A_{q, \bm{\chi}_{j}}^{(j)} - \regular{Q}_{q}^{(j)} \right) \ket{\psi} \approxe{2 \sqrt{\varepsilon}} 0
    \end{equation}
    for any position $j$ and for all $q \in \{1, 2, 3\}$.
\end{proposition}
\begin{proof}
    The relations are special cases of those of \cref{eq:corr_rel_regular} with $\bm{x}_{j} = \bm{\chi}_{j}$.
\end{proof}

\begin{proposition}[Commutativity, Bob]
\label{prop:com_bob}
    For any distinct positions $j$ and $k$ such that $j \neq k$, and for all $q,r \in \{1, 2, 3\}$,
    \begin{equation}
    \label{eq:comm_bob}
        \left[ \regular{Q}_{q}^{(j)}, \regular{Q}_{r}^{(k)} \right] \ket{\psi}
        \approxe{8 \sqrt{\varepsilon}} 0 .
    \end{equation}
\end{proposition}
\begin{proof}
    By construction (see \cref{lem:alice_question_properties_single}), there exists $\bm \chi \in \mathcal{S}$ and an input $\bm x \in \mathcal{X}$ such that $x_{j} = q$, $x_{k} = r$, and $x_{i} = \chi_{i}$ for all $i \in \{1, \dots, n\} \setminus \{j, k\}$.
    For this $\bm x$, consider the observables $A_{\bm x}^{(j)}$ and $A_{\bm x}^{(k)}$ which commute by construction as in \cref{eq:obs_commute_alice}.
    Since $x_{j} = q$ and $x_{k} = r$, we have $A_{\bm x}^{(j)} = A_{q, \bm{x}_{j}}^{(j)}$ and $A_{\bm x}^{(k)} = A_{r, \bm{x}_{k}}^{(k)}$, respectively.
    \Cref{lem:double_shortened_question} ensures that $\bm{x}_{j} \in \mathcal{R}^{(j)}$ and $\bm{x}_{k} \in \mathcal{R}^{(k)}$, and thus from \cref{eq:corr_rel_regular} we have $A_{q, \bm{x}_{j}}^{(j)} \ket{\psi} \approxe{2 \sqrt{\varepsilon}} \regular{Q}_{q}^{(j)} \ket{\psi}$ and $A_{r, \bm{x}_{k}}^{(k)} \ket{\psi} \approxe{2 \sqrt{\varepsilon}} \regular{Q}_{r}^{(k)} \ket{\psi}$.
    Therefore,
    \begin{equation}
        \left[ \regular{Q}_{q}^{(j)}, \regular{Q}_{r}^{(k)} \right] \ket{\psi}
        \approxe{8 \sqrt{\varepsilon}} \left[ A_{r, \bm{x}_{k}}^{(k)}, A_{q, \bm{x}_{j}}^{(j)} \right] \ket{\psi}
        = \left[ A_{\bm x}^{(k)}, A_{\bm x}^{(j)} \right] \ket{\psi}
        = 0
    \end{equation}
    as required.
\end{proof}

\begin{proposition}[Commutativity, Alice]
\label{prop:com_alice}
    For any $\bm \chi \in \mathcal{S}$,
    \begin{equation}
        \left[ A_{q, \bm{\chi}_{j}}^{(j)}, A_{r, \bm{\chi}_{k}}^{(k)} \right] \ket{\psi} \approxe{16 \sqrt{\varepsilon}} 0
    \end{equation}
    for any distinct positions $j$ and $k$ such that $j \neq k$, and all $q,r \in \{1, 2, 3\}$.
\end{proposition}
\begin{proof}
    The observables $A_{q, \bm{\chi}_{j}}^{(j)}$ and $A_{r, \bm{\chi}_{k}}^{(k)}$ are valid due to \cref{lem:alice_question_properties_single}.
    One out of the six commutation relations has both its observables coincide with $\bm \chi$, and so also holds exactly and state-independently.
    Let us suppose $q$ and $r$ are chosen such that this is not the case.
    By \cref{eq:shortened_questions_single,eq:shortened_questions}, $\bm{\chi}_{j} \in \mathcal{R}^{(j)}$ and $\bm{\chi}_{k} \in \mathcal{R}^{(k)}$, and so we may use \cref{eq:corr_rel_regular} to see that
    \begin{equation}
        A_{q, \bm{\chi}_{j}}^{(j)} \ket{\psi} \approxe{2 \sqrt{\varepsilon}} \regular{Q}_{q}^{(j)} \ket{\psi} ,\quad
        A_{r, \bm{\chi}_{k}}^{(k)} \ket{\psi} \approxe{2 \sqrt{\varepsilon}} \regular{Q}_{r}^{(k)} \ket{\psi} .
    \end{equation}
    Therefore,
    \begin{equation}
        \left[ A_{q, \bm{\chi}_{j}}^{(j)}, A_{r, \bm{\chi}_{k}}^{(k)} \right] \ket{\psi}
        \approxe{8 \sqrt{\varepsilon}} \left[ \regular{Q}_{r}^{(k)}, \regular{Q}_{q}^{(j)} \right] \ket{\psi}
        \approxe{8 \sqrt{\varepsilon}} 0 ,
    \end{equation}
    where the final equality is simply \cref{eq:comm_bob}.
\end{proof}

\begin{proposition}[Anticommutativity, Alice]
\label{prop:acomm_alice}
    For any $\bm \chi \in \mathcal{S}$,
    \begin{equation}
        \left\{ A_{q, \bm{\chi}_{j}}^{(j)}, A_{r, \bm{\chi}_{j}}^{(j)} \right\} \ket{\psi}
        \approxe{2 \mathopen{}\left( 1 + \sqrt{2} \right)\mathclose{} \sqrt{\varepsilon}} 0
    \end{equation}
    for any position $j$ and for all distinct $q,r \in \{1, 2, 3\}$ such that $q \neq r$.
\end{proposition}
\begin{proof}
    The observables $A_{q, \bm{\chi}_{j}}^{(j)}$ and $A_{r, \bm{\chi}_{j}}^{(j)}$ are valid due to \cref{lem:alice_question_properties_single}.
    Note that if $D$ and $E$ are $\pm 1$-outcome observables, then $\{D + E, D - E\} = 0$.
    It follows from the definitions made in \cref{eq:bob_untrusted_pauli} that
    \begin{subequations}
    \label{eq:anticomm_bob_all}
    \begin{align}
        \label{eq:anticomm_bob_xy}
        \left\{ X_{\mathrmpair{x}{y}}^{(j)}, Y_{\mathrmpair{x}{y}}^{(j)} \right\} & = 0 , \\
        \label{eq:anticomm_bob_xz}
        \left\{ X_{\mathrmpair{z}{x}}^{(j)}, Z_{\mathrmpair{z}{x}}^{(j)} \right\} & = 0 , \\
        \label{eq:anticomm_bob_yz}
        \left\{ Y_{\mathrmpair{z}{y}}^{(j)}, Z_{\mathrmpair{z}{y}}^{(j)} \right\} & = 0 .
    \end{align}
    \end{subequations}
    Since $\bm{\chi}_{j} \in \mathcal{R}^{(j)}$ by \cref{eq:shortened_questions_single,eq:shortened_questions}, we can use \cref{eq:corr_rel} to exchange the observables $A_{q, \bm{\chi}_{j}}^{(j)}$ and $A_{r, \bm{\chi}_{j}}^{(j)}$ with the appropriate two of Bob's observables from one of \cref{eq:anticomm_bob_xy,eq:anticomm_bob_xz,eq:anticomm_bob_yz}, matching the values of $q$ and $r$.
    Let us denote these two observables here by $R_{q}$ and $R_{r}$.
    Since $\{R_{q}, R_{r}\} = 0$, we have
    \begin{equation}
    \begin{split}
        \left\lVert \left\{ A_{q, \bm{\chi}_{j}}^{(j)}, A_{r, \bm{\chi}_{j}}^{(j)} \right\} \ket{\psi} \right\rVert
        & = \left\lVert \left\{ A_{q, \bm{\chi}_{j}}^{(j)}, A_{r, \bm{\chi}_{j}}^{(j)} \right\} \ket{\psi} - \{ R_{q}, R_{r} \} \ket{\psi} \right\rVert \\
        & \leq (2 + \lVert R_{q} \rVert + \lVert R_{r} \rVert) \sqrt{\varepsilon} ,
    \end{split}
    \end{equation}
    where the inequality follows from \cref{eq:corr_rel}, the triangle inequality, unitarity of $A_{q}^{(j)}$ and $A_{r}^{(j)}$, and the definition of the operator norm.
    Applying the triangle inequality to \cref{eq:bob_untrusted_pauli} to get
    \begin{equation}
        \lVert R_{q} \rVert \leq \sqrt{2} ,\quad \lVert R_{r} \rVert \leq \sqrt{2}
    \end{equation}
    yields the desired expression.
\end{proof}

\begin{proposition}[Anticommutativity, Bob]
\label{prop:acomm_bob}
    For any position $j$,
    \begin{subequations}
    \begin{align}
        \left\{ \regular{Q}_{1}^{(j)}, \regular{Q}_{2}^{(j)} \right\} \ket{\psi} & \approxe{2 \mathopen{}\left( 3 + \sqrt{2} \right)\mathclose{} \sqrt{\varepsilon}} 0 , \\
        \left\{ \regular{Q}_{1}^{(j)}, \regular{Q}_{3}^{(j)} \right\} \ket{\psi} & \approxe{2 \mathopen{}\left( 4 + \sqrt{2} \right)\mathclose{} \sqrt{\varepsilon}} 0 , \\
        \left\{ \regular{Q}_{2}^{(j)}, \regular{Q}_{3}^{(j)} \right\} \ket{\psi} & \approxe{2 \mathopen{}\left( 5 + \sqrt{2} \right)\mathclose{} \sqrt{\varepsilon}} 0 .
    \end{align}
    \end{subequations}
\end{proposition}
\begin{proof}
    Since $Q_{1}^{(j)} = X_{\mathrmpair{x}{y}}^{(j)}$ and $Q_{2}^{(j)} = Y_{\mathrmpair{x}{y}}^{(j)}$, we have as before that $Q_{1}^{(j)}$ and $Q_{2}^{(j)}$ anticommute by construction.
    Furthermore, applying the triangle inequality to the operator norms of \cref{eq:bob_untrusted_pauli_x_xy,eq:bob_untrusted_pauli_y_xy} shows that
    \begin{equation}
        \left\lVert Q_{1}^{(j)} \right\rVert \leq \sqrt{2} ,\quad \left\lVert Q_{2}^{(j)} \right\rVert \leq \sqrt{2} .
    \end{equation}
    Thus, \cref{lem:regularized_acomm} applied using \cref{eq:corr_rel_x_xy,eq:corr_rel_y_xy} immediately gives the first desired relation for the regularized operators $\regular{Q}_{1}^{(j)}$ and $\regular{Q}_{2}^{(j)}$.
    
    We also know, as before, that $X_{\mathrmpair{z}{x}}^{(j)}$ anticommutes with $Z_{\mathrmpair{z}{x}}^{(j)}$.
    Our strategy is thus to write the state-dependent anticommutator of the regularized operators $\regular{Q}_{1}^{(j)} = \regular{X}_{\mathrmpair{x}{y}}^{(j)}$ and $\regular{Q}_{3}^{(j)} = \regular{Z}_{\mathrmpair{z}{x}}^{(j)}$ in terms of $X_{\mathrmpair{z}{x}}^{(j)}$ and $Z_{\mathrmpair{z}{x}}^{(j)}$ (note the differing subscripts between $\regular{X}_{\mathrmpair{x}{y}}^{(j)}$ and $X_{\mathrmpair{z}{x}}^{(j)}$).
    Taking any $\bm \chi \in \mathcal{S}$ (which satisfies $\bm{\chi}_{j} \in \mathcal{R}^{(j)}$ by definition), we have
    \begin{equation}
    \label{eq:anticomm_bob_reg_xz}
        \regular{X}_{\mathrmpair{x}{y}}^{(j)} \regular{Z}_{\mathrmpair{z}{x}}^{(j)} \ket{\psi}
        \approxe{4 \sqrt{\varepsilon}} A_{3, \bm{\chi}_{j}}^{(j)} A_{1, \bm{\chi}_{j}}^{(j)} \ket{\psi}
        \approxe{\left( 1 + \sqrt{2} \right)\mathclose{} \sqrt{\varepsilon}} X_{\mathrmpair{z}{x}}^{(j)} Z_{\mathrmpair{z}{x}}^{(j)} \ket{\psi} ,
    \end{equation}
    where the first estimate uses \cref{eq:corr_rel_x_xy_regular,eq:corr_rel_z_zx_regular}, and the second estimate uses \cref{eq:corr_rel_x_zx,eq:corr_rel_z_zx} along with the triangle inequality applied to the operator norm of \cref{eq:bob_untrusted_pauli_x_zx}.
    We can also write
    \begin{equation}
    \label{eq:anticomm_bob_reg_zx}
        \regular{Z}_{\mathrmpair{z}{x}}^{(j)} \regular{X}_{\mathrmpair{x}{y}}^{(j)} \ket{\psi}
        \approxe{3 \sqrt{\varepsilon}} Z_{\mathrmpair{z}{x}}^{(j)} A_{1, \bm{\chi}_{j}}^{(j)} \ket{\psi}
        \approxe{\sqrt{2 \varepsilon}} Z_{\mathrmpair{z}{x}}^{(j)} X_{\mathrmpair{z}{x}}^{(j)} \ket{\psi} ,
    \end{equation}
    where this time for the first estimate (since we do not need to change the subscripts on the first operator) we use the property of regularization that gives
    \begin{equation}
        \regular{Z}_{\mathrmpair{z}{x}}^{(j)} \ket{\psi} \approxe{\sqrt{\varepsilon}} Z_{\mathrmpair{z}{x}}^{(j)} \ket{\psi}
    \end{equation}
    in light of \cref{eq:corr_rel_z_zx} (see \cref{sec:regularization}).
    Therefore, we can now combine \cref{eq:anticomm_bob_reg_xz,eq:anticomm_bob_reg_zx} conclude the second desired relation
    \begin{equation}
    \begin{split}
        \left\{ \regular{Q}_{1}^{(j)}, \regular{Q}_{3}^{(j)} \right\} \ket{\psi}
        & = \left\{ \regular{X}_{\mathrmpair{x}{y}}^{(j)}, \regular{Z}_{\mathrmpair{z}{x}}^{(j)} \right\} \ket{\psi} \\
        \overset{2 \mathopen{}\left( 4 + \sqrt{2} \right)\mathclose{} \sqrt{\varepsilon}}&{\approx} \left\{ X_{\mathrmpair{z}{x}}^{(j)}, Z_{\mathrmpair{z}{x}}^{(j)} \right\} \ket{\psi}
        = 0 .
    \end{split}
    \end{equation}

    To derive the third and final relation, we know that $Y_{\mathrmpair{z}{y}}^{(j)}$ anticommutes with $Z_{\mathrmpair{z}{y}}^{(j)}$.
    Now we can perform the same process as in deriving \cref{eq:anticomm_bob_reg_xz} to write both
    \begin{subequations}
    \begin{align}
        \regular{Y}_{\mathrmpair{x}{y}}^{(j)} \regular{Z}_{\mathrmpair{z}{x}}^{(j)} \ket{\psi}
        & \approxe{4 \sqrt{\varepsilon}} A_{3, \bm{\chi}_{j}}^{(j)} A_{2, \bm{\chi}_{j}}^{(j)} \ket{\psi}
        \approxe{\left( 1 + \sqrt{2} \right)\mathclose{} \sqrt{\varepsilon}} Y_{\mathrmpair{z}{y}}^{(j)} Z_{\mathrmpair{z}{y}}^{(j)} \ket{\psi} , \\
        \regular{Z}_{\mathrmpair{z}{x}}^{(j)} \regular{Y}_{\mathrmpair{x}{y}}^{(j)} \ket{\psi}
        & \approxe{4 \sqrt{\varepsilon}} A_{2, \bm{\chi}_{j}}^{(j)} A_{3, \bm{\chi}_{j}}^{(j)} \ket{\psi}
        \approxe{\left( 1 + \sqrt{2} \right)\mathclose{} \sqrt{\varepsilon}} Z_{\mathrmpair{z}{y}}^{(j)} Y_{\mathrmpair{z}{y}}^{(j)} \ket{\psi} .
    \end{align}
    \end{subequations}
    Finally, we conclude that
    \begin{equation}
    \begin{split}
        \left\{ \regular{Q}_{2}^{(j)}, \regular{Q}_{3}^{(j)} \right\} \ket{\psi}
        & = \left\{ \regular{Y}_{\mathrmpair{x}{y}}^{(j)}, \regular{Z}_{\mathrmpair{z}{x}}^{(j)} \right\} \ket{\psi} \\
        \overset{2 \mathopen{}\left( 5 + \sqrt{2} \right)\mathclose{} \sqrt{\varepsilon}}&{\approx} \left\{ Y_{\mathrmpair{z}{y}}^{(j)}, Z_{\mathrmpair{z}{y}}^{(j)} \right\} \ket{\psi}
        = 0
    \end{split}
    \end{equation}
    and have now shown all desired anticommutation relations for Bob.
\end{proof}

\begin{proposition}[Complex conjugation relation]
\label{prop:conj_rel}
    For any $\bm \chi \in \mathcal{S}$ and $1 \leq j < n$ we have
    \begin{equation}
        \ket{\psi} + A_{1, \bm{\chi}_{j}}^{(j)} A_{1, \bm{\chi}_{j+1}}^{(j+1)} A_{2, \bm{\chi}_{j}}^{(j)} A_{2, \bm{\chi}_{j+1}}^{(j+1)} A_{3, \bm{\chi}_{j}}^{(j)} A_{3, \bm{\chi}_{j+1}}^{(j+1)} \ket{\psi}
        \approxe{21 \sqrt{\varepsilon}} 0 .
    \end{equation}
\end{proposition}
\begin{proof}
    Suppose that $\bm \chi^{\prime} \in \mathcal{S}$ is the special question for which \cref{eq:request_conjugation_correlations} is satisfied.
    For all $q \in \{1, 2, 3\}$, we can see immediately from the definitions in \cref{eq:shortened_questions_single,eq:shortened_questions} that $\bm{w}_{j} \in \mathcal{R}_{\bm \chi}^{(j)} \subset \mathcal{R}^{(j)}$ and $\bm{w}_{j+1} \in \mathcal{R}_{\bm \chi}^{(j+1)} \subset \mathcal{R}^{(j+1)}$, where $\bm w = (\chi_{1}^{\prime}, \dots, \chi_{j-1}^{\prime}, q, q, \chi_{j+2}^{\prime}, \dots, \chi_{n}^{\prime})$.
    Similarly, we can see for all $\bm \chi \in \mathcal{S}$ that $\bm{\chi}_{i} \in \mathcal{R}_{\bm \chi}^{(i)} \subset \mathcal{R}^{(i)}$ for all $i$.

    We can thus write for all $\bm \chi$ and $q$ that
    \begin{equation}
    \label{eq:conj_to_isom_obs}
        A_{q, \bm{\chi}_{j}}^{(j)} A_{q, \bm{\chi}_{j+1}}^{(j+1)} \ket{\psi}
        \approxe{3 \sqrt{\varepsilon}} \regular{Q}_{q}^{(j+1)} Q_{q}^{(j)} \ket{\psi}
        \approxe{3 \sqrt{\varepsilon}} A_{\bm w}^{(j)} A_{\bm w}^{(j+1)} \ket{\psi} ,
    \end{equation}
    where both estimates use \cref{eq:corr_rel,eq:corr_rel_regular}, and again $\bm w = (\chi_{1}^{\prime}, \dots, \chi_{j-1}^{\prime}, q, q, \chi_{j+2}^{\prime}, \dots, \chi_{n}^{\prime})$.

    Notice now that the operators
    \begin{subequations}
    \begin{align}
          & \Gamma_{1}^{(j)} - \Gamma_{2}^{(j)} + \Gamma_{3}^{(j)} - \Gamma_{4}^{(j)} , \\
        - & \Gamma_{1}^{(j)} + \Gamma_{2}^{(j)} + \Gamma_{3}^{(j)} - \Gamma_{4}^{(j)} , \\
          & \Gamma_{1}^{(j)} + \Gamma_{2}^{(j)} - \Gamma_{3}^{(j)} - \Gamma_{4}^{(j)} ,
    \end{align}
    \end{subequations}
    are unitary, since the $\Gamma_{b}^{(j)}$ as defined in \cref{eq:local_proj_bob_untrusted_bell} form a projective measurement for each $j$.
    Alice's operators appearing in \cref{eq:request_conjugation_correlations} are also unitary by definition in \cref{eq:model_observables_alice}.
    We can thus apply \cref{lem:state_estimate} to \cref{eq:request_conjugation_correlations}.
    Together with \cref{eq:conj_to_isom_obs} holding for all $q$, this implies
    \begin{subequations}
    \begin{align}
        A_{1, \bm{\chi}_{j}}^{(j)} A_{1, \bm{\chi}_{j+1}}^{(j+1)} \ket{\psi}
        & \approxe{7 \sqrt{\varepsilon}} \left( \Gamma_{1}^{(j)} - \Gamma_{2}^{(j)} + \Gamma_{3}^{(j)} - \Gamma_{4}^{(j)} \right) \ket{\psi} , \\
        A_{2, \bm{\chi}_{j}}^{(j)} A_{2, \bm{\chi}_{j+1}}^{(j+1)} \ket{\psi}
        & \approxe{7 \sqrt{\varepsilon}} \left( - \Gamma_{1}^{(j)} + \Gamma_{2}^{(j)} + \Gamma_{3}^{(j)} - \Gamma_{4}^{(j)} \right) \ket{\psi} , \\
        A_{3, \bm{\chi}_{j}}^{(j)} A_{3, \bm{\chi}_{j+1}}^{(j+1)} \ket{\psi}
        & \approxe{7 \sqrt{\varepsilon}} \left( \Gamma_{1}^{(j)} + \Gamma_{2}^{(j)} - \Gamma_{3}^{(j)} - \Gamma_{4}^{(j)} \right) \ket{\psi} .
    \end{align}
    \end{subequations}
    Therefore,
    \begin{equation}
    \begin{split}
        A_{1, \bm{\chi}_{j}}^{(j)} A_{1, \bm{\chi}_{j+1}}^{(j+1)} A_{2, \bm{\chi}_{j}}^{(j)} A_{2, \bm{\chi}_{j+1}}^{(j+1)} A_{3, \bm{\chi}_{j}}^{(j)} A_{3, \bm{\chi}_{j+1}}^{(j+1)} \ket{\psi}
        & \approxe{7 \sqrt{\varepsilon}} \left( \Gamma_{1}^{(j)} + \Gamma_{2}^{(j)} - \Gamma_{3}^{(j)} - \Gamma_{4}^{(j)} \right)\mathclose{} A_{1, \bm{\chi}_{j}}^{(j)} A_{1, \bm{\chi}_{j+1}}^{(j+1)} A_{2, \bm{\chi}_{j}}^{(j)} A_{2, \bm{\chi}_{j+1}}^{(j+1)} \ket{\psi} \\
        & \approxe{7 \sqrt{\varepsilon}} \left( - \Gamma_{1}^{(j)} + \Gamma_{2}^{(j)} - \Gamma_{3}^{(j)} + \Gamma_{4}^{(j)} \right)\mathclose{} A_{1, \bm{\chi}_{j}}^{(j)} A_{1, \bm{\chi}_{j+1}}^{(j+1)} \ket{\psi} \\
        & \approxe{7 \sqrt{\varepsilon}} - \left( \Gamma_{1}^{(j)} + \Gamma_{2}^{(j)} + \Gamma_{3}^{(j)} + \Gamma_{4}^{(j)} \right) \ket{\psi} .
    \end{split}
    \end{equation}
    Since $\sum_{b} \Gamma_{b}^{(j)} = I$, the result follows.
\end{proof}

\section{Discussion}
\label{sec:discussion}

We have shown that our self-testing protocol exhibits all of the requirements to be combined with Fitzsimons--Kashefi-type verifiable blind quantum computation delegation schemes to achieve fully device-independent security.
Furthermore, our protocol achieves several properties that are desirable for future practical implementations.
Of particular note is that, despite being able to certify the remote preparation of a wide variety of states in parallel, only a rudimentary quantum measurement device is needed by the client party.
The input randomness required for generating questions is also small, scaling logarithmically in the number of qubits for the client and with constant-sized questions being sent to the remote server.
Our combined protocol would also enjoy many of the benefits brought by further developments in VBQC protocols.
It can already be optimized, for example, by starting with different resource state structures \cite{kashefi2017optimised,xu2020improved,kashefi2021securing}, and can be made fault-tolerant as in \cite{gheorghiu2015robustness}.
Many of the properties shown of our self-test are also desirable in other applications (especially those involving device-independent state preparation), and as such our work is not restricted to use in delegated quantum computation.
In such cases, it may not be necessary to use as many possible input questions as we have done, and simpler special cases of our tests could be used as is necessary.

\subparagraph{Resource consumption.}

Let us comment on the resources used by our protocol by first focusing on the self-testing and remote state preparation components.
Using standard statistical techniques, achieving a fixed statistical confidence (of say 99\%) for a given error tolerance $\varepsilon$ in our protocol is possible in $\bigO(1 / \varepsilon^{2})$ experimental trials of each of $\bigO(n^{2}) \lvert \mathcal{S} \rvert$ questions.
Using the conservative self-testing robustness estimate we expect to be achievable, some fixed constant distance between physical and reference states in our self-testing and remote state preparation statements (\cref{thm:protocol_isometry,cor:protocol_isometry,thm:state_preparation}) would require an error tolerance $\varepsilon = \Omega(1 / n^{4})$.
Thus, a given fixed robustness can be achieved (with 99\% confidence) in $\bigO(n^{8})$ experimental trials per question.
For the sake of argument, let us take our number of special questions to be $\lvert \mathcal{S} \rvert = \bigO(n)$, resulting in $\bigO(n^{3})$ questions overall and a total of $\bigO(n^{11})$ trials.
Since each of our questions consume $\bigO(\log{n})$ bits of randomness, the total self-testing cost is $\bigO(n^{11} \log{n})$ bits in this case.

A circuit with $g$ gates may be delegated using $N = \bigO(g)$ qubits \cite{childs2005unified,broadbent2009universal,fitzsimons2017unconditionally}.
Starting with the dotted triple graph version of the brickwork resource state used for composability in the robust FK protocol of \cite{gheorghiu2015robustness}, it is possible to achieve exponential security with a number of repetitions that is constant in the size of the computation, all the while conserving the composability and fault tolerance properties of the protocol \cite[Appendix~F]{kashefi2017optimised}.
In this case, the overhead due to verification for a fixed level of security is $\bigO(N)$ additional qubits prepared and $\bigO(N)$ bits of total communication (thanks to the constant number of repetitions).
The total number of qubits that must be prepared is then $n = \bigO(g)$ and the \emph{total} computation cost also scales as $\bigO(g)$.
Errors in soundness and completeness of the robust FK protocol due to a nonideal input state depend only of the trace distance of this state from the ideal \cite{gheorghiu2015robustness}.
Therefore, to obtain a correct answer with some fixed high probability is estimated to cost $\bigO(g^{11} \log{g})$ bits of self-testing resources and $\bigO(g)$ total computation resources, resulting in an overall resource cost estimate for our composite device-independent VBQC protocol that is $\bigO(g^{11} \log{g})$.

Clearly, despite outperforming a number of previous works in this metric, this is not ideal (the state-of-the-art is $\Theta(g \log{g})$ \cite{coladangelo2019verifier}).
It should be noted, however, that while our extra nonlinear cost enters entirely from self-testing, resource estimation is performed under the assumption of noiseless and honest provers, with errors originating only from statistical analysis.
In a more realistic setting with nonadversarial provers contending with depolarizing experimental noise local to each of their $n$ EPR pairs, the comparison is less transparent.
The error tolerance $\varepsilon$ we can achieve would only be worsened by a constant factor (depending on the level of noise), as it refers to outcomes for fixed-sized chunks of registers (one or two EPR pairs each).
Meanwhile, self-testing protocols with robustness depending on an $\epsilon$ that instead represents global failure rate in its tests would see $\epsilon$ increasing to some nonzero constant exponentially quickly in $n$ due to such noise.
In this case, even with robustness guarantees scaling polynomially in $\epsilon$ alone (as $\epsilon \to 0$), it would be extraordinarily difficult to achieve the fidelities required in such an experiment as the number of qubits grows.
Further discussion of this can be found in \cite{arnon2018noise}.

\subparagraph{Future works.}

While we believe that the robustness bounds used for our estimation are conservative and achievable using standard existing techniques, we have opted to wait for improved analytic bounds to be developed (perhaps using techniques such as in \cite{kaniewski2016analytic,kaniewski2017self,gowers2017inverse,natarajan2017quantum,natarajan2018low,sekatski2018certifying,li2019analytic,supic2021device}) that are applicable to local error tolerances.
Any analytic improvements on results of the form of our \cref{thm:single_observable_isometry} or \cref{lem:single_to_multiple_observable} would be of direct consequence to our resource costs.
For practical applicability, numerical optimization approaches such as those using semidefinite programming have yielded much better robustness values (and apparent scaling) than those analytically derived \cite{wu2016device,navascues2007bounding,navascues2008convergent,navascues2015almost,yang2014robust,bancal2015physical,supic2021device,wang2016all,wu2014robust,agresti2021experimental}.
Advances of the computational efficiency of such techniques are, thus, also of great interest.

In case a more technologically capable client device is acceptable, it may be possible to adapt the rigidity results used in \cite{coladangelo2019verifier} to prepare input states to FK-type VBQC protocols.
This would likely lead to a protocol whose total resources scale as $\bigO(g)$ in the noiseless case, an improvement by a logarithmic factor over the state-of-the-art.
Whether the more recent self-testing protocol of \cite{natarajan2018low} with smaller communication could also be used for the required state preparation is an open question.
Device-independent ``one-shot'' tests in the spirit of \cite{arnon2018noise,arnon2019device} could also be studied in the context of state preparation.

\section*{Acknowledgments}

I extend my thanks to Alexandru Gheorghiu and Petros Wallden for many helpful discussions and comments.
I also thank Thomas Vidick for his advice about a technical question.
This work was supported by STFC Grant No. ST/W006537/1 and EPSRC studentship funding under Grant No. EP/R513209/1.

\appendix

\section{Estimation lemmas}
\label{sec:estimation_lemmas}

\begin{proof}[Proof of \cref{lem:trace_dist_bound}]
    First note that for any trace-class operator $T$ we have $\lVert T \rVert_{1} \leq \sqrt{\rank(T)} \lVert T \rVert_{2}$, where $\lVert T \rVert_{2} = \sqrt{\tr(T^{\dagger} T)}$ denotes the Hilbert--Schmidt norm.
    Since $\rank(\proj{u} - \proj{v}) \leq 2$, we thus have
    \begin{equation}
    \label{eq:special_trace_hs_bound}
        \lVert \proj{u} - \proj{v} \rVert_{1} \leq \sqrt{2} \lVert \proj{u} - \proj{v} \rVert_{2} .
    \end{equation}
    We then evaluate
    \begin{equation}
    \label{eq:special_hs}
    \begin{split}
        2 \lVert \proj{u} - \proj{v} \rVert_{2}^{2}
        & = 2 \tr \mathopen{}\left[ (\proj{u} - \proj{v})^2 \right]\mathclose{} \\
        & = 2 \lVert \ket{u} \rVert^{4} + 2 \lVert \ket{v} \rVert^{4} - 4 \lvert \braket{u}{v} \rvert^{2} .
    \end{split}
    \end{equation}
    Using the definition of the norm induced by the inner product,
    \begin{equation}
    \label{eq:overlap_bound}
        2 \lvert \braket{u}{v} \rvert
        \geq 2 \Re{\braket{u}{v}}
        = \lVert \ket{u} \rVert^{2} + \lVert \ket{v} \rVert^{2} - \lVert \ket{u} - \ket{v} \rVert^{2} .
    \end{equation}
    For simplicity, let us adopt the notation
    \begin{subequations}
    \begin{align}
        \delta & = \lVert \ket{u} - \ket{v} \rVert , \\
        M & = \max \{ \lVert \ket{u} \rVert, \lVert \ket{v} \rVert \} .
    \end{align}
    \end{subequations}
    Inserting \cref{eq:overlap_bound} into \cref{eq:special_hs} and using that $\lVert \ket{u} \rVert^{2} + \lVert \ket{v} \rVert^{2} \leq 2M^{2}$, we have
    \begin{equation}
        2 \lVert \proj{u} - \proj{v} \rVert_{2}^{2}
        \leq \left( \lVert \ket{u} \rVert^{2} - \lVert \ket{v} \rVert^{2} \right)^{2}
        + 4 M^{2} \delta^{2}
        - \delta^{4} .
    \end{equation}
    It can be seen (similarly to the reverse triangle inequality) that
    \begin{equation}
    \begin{split}
        \left\lvert \lVert \ket{u} \rVert^{2} - \lVert \ket{v} \rVert^{2} \right\rvert
        & \leq \lVert \ket{u} - \ket{v} \rVert^{2} \\
        & \phantom{{}\leq{}} + 2 \max(\lVert \ket{u} \rVert, \lVert \ket{v} \rVert) \lVert \ket{u} - \ket{v} \rVert .
    \end{split}
    \end{equation}
    Inserting this into the previous equation gives
    \begin{equation}
        2 \lVert \proj{u} - \proj{v} \rVert_{2}^{2}
        \leq 8M^{2} \delta^{2} + 4M \delta^{3} .
    \end{equation}
    Since $\lVert \ket{u} \rVert \leq 1$ and $\lVert \ket{v} \rVert \leq 1$, the triangle inequality implies $\delta \leq 2$, and so $\delta^{3} \leq 2 \delta^{2}$.
    We thus have
    \begin{equation}
        2 \lVert \proj{u} - \proj{v} \rVert_{2}^{2}
        \leq 8 \mathopen{}\left( M^{2} + M \right)\mathclose{} \delta^{2} .
    \end{equation}
    Combining this with \cref{eq:special_trace_hs_bound} gives
    \begin{equation}
        \lVert \proj{u} - \proj{v} \rVert_{1}
        \leq \delta \sqrt{8 \mathopen{}\left( M^{2} + M \right)\mathclose{}} .
    \end{equation}
    Finally, since by assumption $0 \leq M \leq 1$, we get
    \begin{equation}
        \lVert \proj{u} - \proj{v} \rVert_{1}
        \leq 4 \delta
    \end{equation}
    as required.
\end{proof}

\section{Post-measurement robustness probability}
\label{sec:robust_prob}

\begin{proof}[Proof of \cref{lem:general_robust_prob}]
    For the case where $\delta = 0$, it is clear from \cref{eq:norm_squares_sum_bound} that for all $\sigma \in \Sigma$ satisfying $\pi(\sigma) > 0$ we have $\ket{u_{\sigma}^{\omega}} = \ket{v_{\sigma}^{\omega}}$ for all $\omega \in \Omega$, and thus we have that $\Pr(D \leq 0) = 1$.
    We may henceforth assume that $\delta > 0$.
    First, let us introduce a new random variable $N$ on $\Sigma \times \Omega$ defined as
    \begin{equation}
        N(\sigma, \omega) =
        \begin{cases}
            \lVert \ket{\hat{u}_{\sigma}^{\omega}} - \ket{\hat{v}_{\sigma}^{\omega}} \rVert & \text{if $p(\sigma, \omega) > 0$,} \\
            0 & \text{if $p(\sigma, \omega) = 0$.}
        \end{cases}
    \end{equation}
    Let $a > 0$.
    By \cref{lem:trace_dist_bound} we have $D(\sigma, \omega) \leq N(\sigma, \omega)$ for all $\sigma \in \Sigma$ and $\omega \in \Omega$, and thus
    \begin{equation}
    \label{eq:trace_vector_norm_probs}
        \Pr(D \geq a) \leq \Pr(N \geq a) .
    \end{equation}
    We now bound the expected value of $N^{2}$.
    We evaluate
    \begin{equation}
    \label{eq:expected_norm_squares}
    \begin{split}
        \E \mathopen{}\left( N^{2} \right)\mathclose{}
        & = \sum_{\sigma \in \Sigma} \pi(\sigma) \sum_{\omega \in \Omega} p_{\sigma}(\omega) N(\sigma, \omega)^{2} \\
        & = \sum_{\sigma \in \Sigma} \pi(\sigma) \sum_{\omega \in \Omega} \lVert \ket{u_{\sigma}^{\omega}} \rVert^{2} N(\sigma, \omega)^{2} \\
        & = \sum_{\sigma \in \Sigma} \pi(\sigma) \sum_{\omega \in \Omega} \mathopen{}\bigl\lVert \ket{u_{\sigma}^{\omega}} - \lVert \ket{u_{\sigma}^{\omega}} \rVert \ket{\hat{v}_{\sigma}^{\omega}} \bigr\rVert^{2}
    \end{split}
    \end{equation}
    Now note that
    \begin{equation}
    \label{eq:expected_triangle_term}
    \begin{split}
        \bigl\lVert \lVert \ket{u_{\sigma}^{\omega}} \rVert \ket{\hat{v}_{\sigma}^{\omega}} - \ket{v_{\sigma}^{\omega}} \bigr\rVert
        & = \bigl\lvert \lVert \ket{u_{\sigma}^{\omega}} \rVert - \lVert \ket{v_{\sigma}^{\omega}} \rVert \bigr\rvert \\
        & \leq \lVert \ket{u_{\sigma}^{\omega}} - \ket{v_{\sigma}^{\omega}} \rVert ,
    \end{split}
    \end{equation}
    where we have used the reverse triangle inequality.
    Starting with the triangle inequality, we can then write for the terms of \cref{eq:expected_norm_squares} that
    \begin{equation}
    \begin{split}
        \bigl\lVert \ket{u_{\sigma}^{\omega}} - \lVert \ket{u_{\sigma}^{\omega}} \rVert \ket{\hat{v}_{\sigma}^{\omega}} \bigr\rVert
        & \leq \lVert \ket{u_{\sigma}^{\omega}} - \ket{v_{\sigma}^{\omega}} \rVert \\
        & \phantom{{}\leq{}} + \bigl\lVert \ket{v_{\sigma}^{\omega}} - \lVert \ket{u_{\sigma}^{\omega}} \rVert \ket{\hat{v}_{\sigma}^{\omega}} \bigr\rVert \\
        & \leq 2 \lVert \ket{u_{\sigma}^{\omega}} - \ket{v_{\sigma}^{\omega}} \rVert ,
    \end{split}
    \end{equation}
    where the final inequality uses \cref{eq:expected_triangle_term}.
    We thus have
    \begin{equation}
    \label{eq:expected_norm_squares_bound}
        \E \mathopen{}\left( N^{2} \right)\mathclose{}
        \leq 4 \sum_{\sigma \in \Sigma} \pi(\sigma) \sum_{\omega \in \Omega} \lVert \ket{u_{\sigma}^{\omega}} - \ket{v_{\sigma}^{\omega}} \rVert^{2}
        \leq 4 \delta^{2} ,
    \end{equation}
    where the final inequality comes from the assumption of \cref{eq:norm_squares_sum_bound}.
    Markov's inequality states that
    \begin{equation}
    \begin{split}
        \Pr(N \geq a)
        & = \Pr \mathopen{}\left( N^{2} \geq a^{2} \right) \\
        & \leq \frac{1}{a^{2}} \E \mathopen{}\left( N^{2} \right)\mathclose{} .
    \end{split}
    \end{equation}
    Combining this with \cref{eq:trace_vector_norm_probs,eq:expected_norm_squares_bound} gives
    \begin{equation}
        \Pr(D \geq a)
        \leq \frac{4 \delta^{2}}{a^{2}} .
    \end{equation}
    Taking complements yields
    \begin{equation}
        \Pr(D \leq a)
        \geq \Pr(D < a)
        \geq 1 - \frac{4 \delta^{2}}{a^{2}} .
    \end{equation}
    Finally, choosing the parameter $a = \delta^{c}$ gives the desired result.
\end{proof}

\section{Proof of single-copy self-test}
\label{sec:single_test_proof}

Here, we exhibit a proof of \cref{prop:swap_and_kickback_isoms} in the ideal case that $\eta = 0$.
The robust case of $\eta > 0$ is discussed in \cref{sec:single_test_robustness}.

\begin{proof}[Proof of \cref{prop:swap_and_kickback_isoms} (ideal case)]
    We first consider the isometry applied to the state $\ket{\psi}_{\hilb{A} \hilb{B}}$.
    After the ``swap'' stage of the circuit, given by $W$, we have the state
    \begin{equation}
    \label{eq:swapped_k}
    \begin{split}
        W \ket{\psi}
        = \frac{1}{4} [
        &   \ket{00}_{\hilb{A}^{\prime} \hilb{B}^{\prime}} \otimes (I - i S_{2} S_{1}) ( I + i T_{2} T_{1}) \ket{\psi} \\
        & + \ket{01}_{\hilb{A}^{\prime} \hilb{B}^{\prime}} \otimes (I - i S_{2} S_{1}) T_{1} (I - i T_{2} T_{1}) \ket{\psi} \\
        & + \ket{10}_{\hilb{A}^{\prime} \hilb{B}^{\prime}} \otimes S_{1} (I + i S_{2} S_{1}) (I + i T_{2} T_{1}) \ket{\psi} \\
        & + \ket{11}_{\hilb{A}^{\prime} \hilb{B}^{\prime}} \otimes S_{1} (I + i S_{2} S_{1}) T_{1} (I - i T_{2} T_{1}) \ket{\psi}
        ] .
    \end{split}
    \end{equation}
    We now simplify this expression.
    Using the relations of the statement we can write
    \begin{equation}
    \label{eq:swap_vanish_terms}
        (I \pm i S_{2} S_{1}) (I \pm i T_{2} T_{1}) \ket{\psi} = 0 .
    \end{equation}
    Thus, the terms corresponding to ancilla states $\ket{01}_{\hilb{A}^{\prime} \hilb{B}^{\prime}}$ and $\ket{10}_{\hilb{A}^{\prime} \hilb{B}^{\prime}}$ vanish, and we are left with
    \begin{equation}
    \label{eq:swap_vanished}
    \begin{split}
        W \ket{\psi}
        = \frac{1}{4} [
        &   \ket{00}_{\hilb{A}^{\prime} \hilb{B}^{\prime}} \otimes (I - i S_{2} S_{1}) (I + i T_{2} T_{1}) \ket{\psi} \\
        & + \ket{11}_{\hilb{A}^{\prime} \hilb{B}^{\prime}} \otimes S_{1} (I + i S_{2} S_{1}) T_{1} (I - i T_{2} T_{1}) \ket{\psi}
        ] .
    \end{split}
    \end{equation}
    Using \cref{rel:acomm_single}, we have that
    \begin{subequations}
    \label{eq:swap_clean}
    \begin{align}
        S_{1} (I + i S_{2} S_{1}) \ket{\psi}
        & = (I - i S_{2} S_{1}) S_{1} \ket{\psi} , \\
        T_{1} (I - i T_{2} T_{1}) \ket{\psi}
        & = (I + i T_{2} T_{1}) T_{1} \ket{\psi} .
    \end{align}
    \end{subequations}
    Thus, using these in addition to \cref{rel:corr_single} and the fact that our observables are involutory results in
    \begin{equation}
        W \ket{\psi}
        = \ket{\Phi^{+}}_{\hilb{A}^{\prime} \hilb{B}^{\prime}} \otimes \frac{1}{2 \sqrt{2}} (I - i S_{2} S_{1}) (I + i T_{2} T_{1}) \ket{\psi} .
    \end{equation}
    Using the relations, and again that our observables are involutory, the state simplifies to
    \begin{equation}
        W \ket{\psi} = \ket{\Phi^{+}}_{\hilb{A}^{\prime} \hilb{B}^{\prime}} \otimes \ket{\varphi} .
    \end{equation}
    With this, we have now extracted the desired maximally entangled state from our initial unknown state.
    In a similar fashion, it can be shown that
    \begin{subequations}
    \begin{align}
        W S_{1} \ket{\psi}
        & = \sigma_{\mathrm{x}}^{\hilb{B}^{\prime}} \ket{\Phi^{+}}_{\hilb{A}^{\prime} \hilb{B}^{\prime}} \otimes \ket{\varphi} , \\
        W S_{2} \ket{\psi}
        & = \sigma_{\mathrm{y}}^{\hilb{B}^{\prime}} \ket{\Phi^{+}}_{\hilb{A}^{\prime} \hilb{B}^{\prime}} \otimes \ket{\varphi} .
    \end{align}
    \end{subequations}
    Furthermore, although with a little more work (as this is the case where complex conjugation will later become relevant), it can be shown that
    \begin{equation}
        W S_{2} \ket{\psi}
        = \sigma_{\mathrm{z}}^{\hilb{B}^{\prime}} \ket{\Phi^{+}}_{\hilb{A}^{\prime} \hilb{B}^{\prime}} \otimes S_{3} \ket{\varphi} .
    \end{equation}

    We now apply the ``phase kickback'' stage of the isometry, given by $K$, to the above simplified states.
    We suppress the extracted state of the primed ancillae in our notation, as it is entirely unaffected by $K$.
    When the state is $\ket{\varphi}$, this gives
    \begin{equation}
    \begin{split}
        K \ket{\varphi}
        = \frac{1}{4 \sqrt{2}} \mathopen{}\bigl[
        & \ket{0}_{\hilb{A}^{\prime\prime}} \ket{0}_{\hilb{B}^{\prime\prime}}
        \otimes (I + S_{3}) (I + T_{3}) (I + i T_{2} T_{1}) \ket{\psi}
        + \ket{0}_{\hilb{A}^{\prime\prime}} \ket{1}_{\hilb{B}^{\prime\prime}}
        \otimes (I + S_{3}) (I - T_{3}) (I + i T_{2} T_{1}) \ket{\psi} \\
        & + \ket{1}_{\hilb{A}^{\prime\prime}} \ket{0}_{\hilb{B}^{\prime\prime}}
        \otimes (I - S_{3}) (I + T_{3}) (I + i T_{2} T_{1}) \ket{\psi}
        + \ket{1}_{\hilb{A}^{\prime\prime}} \ket{1}_{\hilb{B}^{\prime\prime}}
        \otimes (I - S_{3}) (I - T_{3}) (I + i T_{2} T_{1}) \ket{\psi}
        \bigr]\mathclose{} .
    \end{split}
    \end{equation}
    We now simplify this expression.
    Since $S_{3} (I + i T_{2} T_{1}) \ket{\psi} = T_{3} (I + i T_{2} T_{1}) \ket{\psi}$ by \cref{rel:corr,rel:acomm}, and using that $(I \pm T_{3}) (I \pm T_{3}) = 2 (I \pm T_{3})$ and $(I \pm T_{3}) (I \mp T_{3}) = 0$, we get
    \begin{equation}
    \label{eq:kickback_single_simplified}
        K \ket{\varphi}
        = \ket{0}_{\hilb{A}^{\prime\prime}} \ket{0}_{\hilb{B}^{\prime\prime}}
        \otimes \frac{I + T_{3}}{2} \ket{\varphi}
        + \ket{1}_{\hilb{A}^{\prime\prime}} \ket{1}_{\hilb{B}^{\prime\prime}}
        \otimes \frac{I - T_{3}}{2} \ket{\varphi} .
    \end{equation}
    Otherwise, when the state is $S_{3} \ket{\varphi}$, we have
    \begin{equation}
    \begin{split}
        K S_{3} \ket{\varphi}
        = \frac{1}{4 \sqrt{2}} \mathopen{}\bigl[
        & \ket{0}_{\hilb{A}^{\prime\prime}} \ket{0}_{\hilb{B}^{\prime\prime}}
        \otimes (I + S_{3}) S_{3} (I + T_{3}) (I + i T_{2} T_{1}) \ket{\psi}
        + \ket{0}_{\hilb{A}^{\prime\prime}} \ket{1}_{\hilb{B}^{\prime\prime}}
        \otimes (I + S_{3}) S_{3} (I - T_{3}) (I + i T_{2} T_{1}) \ket{\psi} \\
        & + \ket{1}_{\hilb{A}^{\prime\prime}} \ket{0}_{\hilb{B}^{\prime\prime}}
        \otimes (I - S_{3}) S_{3} (I + T_{3}) (I + i T_{2} T_{1}) \ket{\psi}
        + \ket{1}_{\hilb{A}^{\prime\prime}} \ket{1}_{\hilb{B}^{\prime\prime}}
        \otimes (I - S_{3}) S_{3} (I - T_{3}) (I + i T_{2} T_{1}) \ket{\psi}
        \bigr]\mathclose{} .
    \end{split}
    \end{equation}
    Since $S_{3}$ is involutory, we have $(I + S_{3}) S_{3} = (I + S_{3})$ and $(I - S_{3}) S_{3} = - (I - S_{3})$.
    By the same argument as before, we arrive at
    \begin{equation}
    \label{eq:kickback_single_z_simplified}
    \begin{split}
        K S_{3} \ket{\varphi}
        ={} & \ket{0}_{\hilb{A}^{\prime\prime}} \ket{0}_{\hilb{B}^{\prime\prime}}
        \otimes \frac{I + T_{3}}{2} \ket{\varphi} \\
        & - \ket{1}_{\hilb{A}^{\prime\prime}} \ket{1}_{\hilb{B}^{\prime\prime}}
        \otimes \frac{I - T_{3}}{2} \ket{\varphi} .
    \end{split}
    \end{equation}
    Upon defining $\ket{\xi_{0}}$ and $\ket{\xi_{1}}$ as in the statement, we have now finished applying the isometries.
\end{proof}

\section{Robustness of single-copy self-test}
\label{sec:single_test_robustness}

Here, we expand upon the proof of the ideal case in which $\eta = 0$ discussed in \cref{sec:single_test_proof}.
We fill in details concerning the robustness of that proof, extending it to handle also cases where $\eta > 0$, and thus achieving the full claim of \cref{prop:swap_and_kickback_isoms}.

\begin{proof}[Proof of \cref{prop:swap_and_kickback_isoms} (robustness)]
    Expanding the left-hand side of \cref{eq:swap_vanish_terms} and taking its norm gives
    \begin{equation}
    \label{eq:swap_vanish_norm}
    \begin{split}
        & \lVert (I \pm i S_{2} S_{1}) (I \pm i T_{2} T_{1}) \ket{\psi} \rVert \\
        & \quad = \lVert (I - S_{2} S_{1} T_{2} T_{1}) \ket{\psi} \pm i (S_{2} S_{1} + T_{2} T_{1}) \ket{\psi} \rVert \\
        & \quad \leq \lVert \ket{\psi} - S_{2} S_{1} T_{2} T_{1} \ket{\psi} \rVert + \lVert S_{2} S_{1} \ket{\psi} + T_{2} T_{1} \ket{\psi} \rVert .
    \end{split}
    \end{equation}
    Using \cref{rel:corr_single} and that the observables are both unitary and involutory, we have
    \begin{equation}
        S_{q} T_{q} \ket{\psi}
        \approxe{\eta} S_{q} S_{q} \ket{\psi}
        = \ket{\psi} .
    \end{equation}
    Since norms are preserved under unitary operations, we can then bound the norm of all the expressions
    \begin{equation}
    \label{eq:swap_vanish_norm_1}
    \begin{split}
        \lVert \ket{\psi} - S_{2} S_{1} T_{2} T_{1} \ket{\psi} \rVert
        & = \lVert S_{2} S_{1} \ket{\psi} - T_{1} T_{2} \ket{\psi} \rVert \\
        & = \lVert S_{1} T_{1} \ket{\psi} - S_{2} T_{2} \ket{\psi} \rVert \\
        & \leq 2 \eta .
    \end{split}
    \end{equation}
    Now we have bounded the first term of \cref{eq:swap_vanish_norm}.
    For the second term, using \cref{rel:acomm_single} gives
    \begin{equation}
        S_{2} S_{1} \ket{\psi} + T_{2} T_{1} \ket{\psi}
        \approxe{\eta} S_{2} S_{1} \ket{\psi} - T_{1} T_{2} \ket{\psi} ,
    \end{equation}
    for which we have already bounded the norm.
    Thus,
    \begin{equation}
    \label{eq:swap_vanish_norm_2}
        \lVert S_{2} S_{1} \ket{\psi} + T_{2} T_{1} \ket{\psi} \rVert
        \leq 3 \eta .
    \end{equation}
    Combining \cref{eq:swap_vanish_norm,eq:swap_vanish_norm_1,eq:swap_vanish_norm_2}, we arrive at a robust version of \cref{eq:swap_vanish_terms}.
    That is
    \begin{equation}
        (I \pm i S_{2} S_{1}) (I \pm i T_{2} T_{1}) \ket{\psi}
        \approxe{5 \eta} 0 .
    \end{equation}
    This immediately allows us to write a robust version of \cref{eq:swap_vanished}
    \begin{equation}
    \label{eq:swap_vanished_robust}
    \begin{split}
        W \ket{\psi}
        \approxe{5 \eta / 2} \frac{1}{4} [
        &   \ket{00}_{\hilb{A}^{\prime} \hilb{B}^{\prime}} \otimes (I - i S_{2} S_{1}) (I + i T_{2} T_{1}) \ket{\psi} \\
        & + \ket{11}_{\hilb{A}^{\prime} \hilb{B}^{\prime}} \otimes S_{1} (I + i S_{2} S_{1}) T_{1} (I - i T_{2} T_{1}) \ket{\psi}
        ] .
    \end{split}
    \end{equation}
    We now examine the second term of \cref{eq:swap_vanished_robust}.
    We can write \cref{eq:swap_clean} robustly as
    \begin{subequations}
    \begin{align}
        \label{eq:swap_clean_robust_alice}
        S_{1} (I + i S_{2} S_{1}) \ket{\psi}
        & \approxe{\eta} (I - i S_{2} S_{1}) S_{1} \ket{\psi} , \\
        \label{eq:swap_clean_robust_bob}
        T_{1} (I - i T_{2} T_{1}) \ket{\psi}
        & \approxe{\eta} (I + i T_{2} T_{1}) T_{1} \ket{\psi} .
    \end{align}
    \end{subequations}
    Now, since $\lVert S_{1} (I + i S_{2} S_{1}) \rVert \leq 2$, \cref{eq:swap_clean_robust_bob} implies that
    \begin{equation}
    \label{eq:swap_cleaned_robust_1}
    \begin{split}
        & S_{1} (I + i S_{2} S_{1}) T_{1} (I - i T_{2} T_{1}) \ket{\psi} \\
        & \quad \approxe{2 \eta} S_{1} (I + i S_{2} S_{1}) (I + i T_{2} T_{1}) T_{1} \ket{\psi} .
    \end{split}
    \end{equation}
    Similarly, since $\lVert (I + i T_{2} T_{1}) T_{1} \rVert \leq 2$, \cref{eq:swap_clean_robust_alice} implies that
    \begin{equation}
    \label{eq:swap_cleaned_robust_2}
    \begin{split}
        & S_{1} (I + i S_{2} S_{1}) (I + i T_{2} T_{1}) T_{1} \ket{\psi} \\
        & \quad \approxe{2 \eta} (I - i S_{2} S_{1}) S_{1} (I + i T_{2} T_{1}) T_{1} \ket{\psi} .
    \end{split}
    \end{equation}
    Finally, since $\lVert (I - i S_{2} S_{1}) (I + i T_{2} T_{1}) \rVert \leq 4$, \cref{rel:corr_single} implies that
    \begin{equation}
    \label{eq:swap_cleaned_robust_3}
        (I - i S_{2} S_{1}) S_{1} (I + i T_{2} T_{1}) T_{1} \ket{\psi}
        \approxe{4 \eta} (I - i S_{2} S_{1}) (I + i T_{2} T_{1}) \ket{\psi} .
    \end{equation}
    Combining \cref{eq:swap_cleaned_robust_1,eq:swap_cleaned_robust_2,eq:swap_cleaned_robust_3} through the triangle inequality gives
    \begin{equation}
        S_{1} (I + i S_{2} S_{1}) T_{1} (I - i T_{2} T_{1}) \ket{\psi}
        \approxe{8 \eta} (I - i S_{2} S_{1}) (I + i T_{2} T_{1}) \ket{\psi} .
    \end{equation}
    We can now apply this to \cref{eq:swap_vanished_robust} to estimate $W \ket{\psi}$ by
    \begin{equation}
    \label{eq:cleaned_robust}
        W \ket{\psi}
        \approxe{9 \eta / 2} \ket{\Phi^{+}}_{\hilb{A}^{\prime} \hilb{B}^{\prime}}
        \otimes \frac{1}{2 \sqrt{2}} (I - i S_{2} S_{1}) (I + i T_{2} T_{1}) \ket{\psi} .
    \end{equation}
    Since $\lVert I + i T_{2} T_{1} \rVert \leq 2$, we have
    \begin{equation}
    \begin{split}
        (I - i S_{2} S_{1}) (I + i T_{2} T_{1}) \ket{\psi}
        & \approxe{4 \eta} (I + i T_{2} T_{1}) (I - i T_{1} T_{2}) \ket{\psi} \\
        & = (2I + i T_{2} T_{1} - i T_{1} T_{2}) \ket{\psi} \\
        & \approxe{\eta} 2 (I + i T_{2} T_{1}) \ket{\psi} \\
        & = 2 \sqrt{2} \ket{\varphi} ,
    \end{split}
    \end{equation}
    where for the first line we used \cref{eq:swap_vanish_norm_1} and for the third line we used \cref{rel:acomm_single}.
    Thus, we can apply this to \cref{eq:cleaned_robust} to estimate $W \ket{\psi}$ by
    \begin{equation}
        \left\lVert W \ket{\psi}
        - \ket{\Phi^{+}}_{\hilb{A}^{\prime} \hilb{B}^{\prime}}
        \otimes \ket{\varphi} \right\rVert \leq \frac{1}{4} \mathopen{}\left( 18 + 5 \sqrt{2} \right)\mathclose{} \eta .
    \end{equation}
    
    To estimate $W S_{1} \ket{\psi}$, we note that
    \begin{subequations}
    \begin{align}
    \begin{split}
        & (I - i S_{2} S_{1}) ( I + i T_{2} T_{1}) S_{1} \ket{\psi} \\
        & \quad \approxe{2 \eta} S_{1} (I + i S_{2} S_{1}) ( I + i T_{2} T_{1}) \ket{\psi} ,
    \end{split} \\
    \begin{split}
        & (I - i S_{2} S_{1}) T_{1} (I - i T_{2} T_{1}) S_{1} \ket{\psi} \\
        & \quad \approxe{2 \eta} S_{1} (I + i S_{2} S_{1}) T_{1} (I - i T_{2} T_{1}) \ket{\psi} ,
    \end{split} \\
    \begin{split}
        & S_{1} (I + i S_{2} S_{1}) (I + i T_{2} T_{1}) S_{1} \ket{\psi} \\
        & \quad \approxe{2 \eta} (I - i S_{2} S_{1}) (I + i T_{2} T_{1}) \ket{\psi} ,
    \end{split} \\
    \begin{split}
        & S_{1} (I + i S_{2} S_{1}) T_{1} (I - i T_{2} T_{1}) S_{1} \ket{\psi} \\
        & \quad \approxe{2 \eta} (I - i S_{2} S_{1}) T_{1} (I - i T_{2} T_{1}) \ket{\psi} .
    \end{split}
    \end{align}
    \end{subequations}
    Therefore, we have
    \begin{equation}
        \begin{split}
            W S_{1} \ket{\psi}
            \approxe{2 \eta} \frac{1}{4} [
            &   \ket{00}_{\hilb{A}^{\prime} \hilb{B}^{\prime}} \otimes S_{1} (I + i S_{2} S_{1}) ( I + i T_{2} T_{1}) \ket{\psi} \\
            & + \ket{01}_{\hilb{A}^{\prime} \hilb{B}^{\prime}} \otimes S_{1} (I + i S_{2} S_{1}) T_{1} (I - i T_{2} T_{1}) \ket{\psi} \\
            & + \ket{10}_{\hilb{A}^{\prime} \hilb{B}^{\prime}} \otimes (I - i S_{2} S_{1}) (I + i T_{2} T_{1}) \ket{\psi} \\
            & + \ket{11}_{\hilb{A}^{\prime} \hilb{B}^{\prime}} \otimes (I - i S_{2} S_{1}) T_{1} (I - i T_{2} T_{1}) \ket{\psi}
            ] .
    \end{split}
    \end{equation}
    The right-hand side of this equation is just that of \cref{eq:swapped_k} but with different states in $\hilb{A}^{\prime} \otimes \hilb{B}^{\prime}$ identifying each term of the superposition.
    Considering this, we can use the same robustness arguments as before to deduce that
    \begin{equation}
        \left\lVert W S_{1} \ket{\psi}
        - \sigma_{\mathrm{x}}^{\hilb{B}^{\prime}} \ket{\Phi^{+}}_{\hilb{A}^{\prime} \hilb{B}^{\prime}}
        \otimes \ket{\varphi} \right\rVert \leq \frac{1}{4} \mathopen{}\left( 26 + 5 \sqrt{2} \right)\mathclose{} \eta .
    \end{equation}
    
    To estimate $W S_{2} \ket{\psi}$, we note that
    \begin{subequations}
    \begin{align}
    \begin{split}
        & (I - i S_{2} S_{1}) ( I + i T_{2} T_{1}) S_{2} \ket{\psi} \\
        & \quad \approxe{4 \eta} i S_{1} (I + i S_{2} S_{1}) ( I + i T_{2} T_{1}) \ket{\psi} ,
    \end{split} \\
    \begin{split}
        & (I - i S_{2} S_{1}) T_{1} (I - i T_{2} T_{1}) S_{2} \ket{\psi} \\
        & \quad \approxe{4 \eta} i S_{1} (I + i S_{2} S_{1}) T_{1} (I - i T_{2} T_{1}) \ket{\psi} ,
    \end{split} \\
    \begin{split}
        & S_{1} (I + i S_{2} S_{1}) (I + i T_{2} T_{1}) S_{2} \ket{\psi} \\
        & \quad \approxe{4 \eta} -i (I - i S_{2} S_{1}) ( I + i T_{2} T_{1}) \ket{\psi} ,
    \end{split} \\
    \begin{split}
        & S_{1} (I + i S_{2} S_{1}) T_{1} (I - i T_{2} T_{1}) S_{2} \ket{\psi} \\
        & \quad \approxe{4 \eta} -i (I - i S_{2} S_{1}) T_{1} (I - i T_{2} T_{1}) S_{2} \ket{\psi} .
    \end{split}
    \end{align}
    \end{subequations}
    Therefore, we have
    \begin{equation}
        \begin{split}
            W S_{2} \ket{\psi}
            \approxe{4 \eta} \frac{1}{4} [
            &   i \ket{00}_{\hilb{A}^{\prime} \hilb{B}^{\prime}} \otimes S_{1} (I + i S_{2} S_{1}) ( I + i T_{2} T_{1}) \ket{\psi} \\
            & + i \ket{01}_{\hilb{A}^{\prime} \hilb{B}^{\prime}} \otimes S_{1} (I + i S_{2} S_{1}) T_{1} (I - i T_{2} T_{1}) \ket{\psi} \\
            & - i \ket{10}_{\hilb{A}^{\prime} \hilb{B}^{\prime}} \otimes (I - i S_{2} S_{1}) (I + i T_{2} T_{1}) \ket{\psi} \\
            & - i \ket{11}_{\hilb{A}^{\prime} \hilb{B}^{\prime}} \otimes (I - i S_{2} S_{1}) T_{1} (I - i T_{2} T_{1}) \ket{\psi}
            ] .
    \end{split}
    \end{equation}
    Again, we can use the same robustness arguments as before to deduce that
    \begin{equation}
        \left\lVert W S_{2} \ket{\psi}
        - \sigma_{\mathrm{y}}^{\hilb{B}^{\prime}} \ket{\Phi^{+}}_{\hilb{A}^{\prime} \hilb{B}^{\prime}}
        \otimes \ket{\varphi} \right\rVert \leq \frac{1}{4} \mathopen{}\left( 34 + 5 \sqrt{2} \right)\mathclose{} \eta .
    \end{equation}
    
    The final case of estimating $W S_{3} \ket{\psi}$ requires a little more care.
    By repeatedly applying \cref{rel:acomm_single}, it can be shown that
    \begin{subequations}
    \begin{align}
    \begin{split}
        & (I - i S_{2} S_{1}) ( I + i T_{2} T_{1}) S_{3} \ket{\psi} \\
        & \quad \approxe{8 \eta} S_{3} (I - i S_{2} S_{1}) ( I + i T_{2} T_{1}) \ket{\psi} ,
    \end{split} \\
    \begin{split}
        & (I - i S_{2} S_{1}) T_{1} (I - i T_{2} T_{1}) S_{3} \ket{\psi} \\
        & \quad \approxe{8 \eta} S_{3} (I - i S_{2} S_{1}) T_{1} (I - i T_{2} T_{1}) \ket{\psi} ,
    \end{split} \\
    \begin{split}
        & S_{1} (I + i S_{2} S_{1}) (I + i T_{2} T_{1}) S_{3} \ket{\psi} \\
        & \quad \approxe{16 \eta} - S_{3} S_{1} (I + i S_{2} S_{1}) (I + i T_{2} T_{1}) \ket{\psi} ,
    \end{split} \\
    \begin{split}
        & S_{1} (I + i S_{2} S_{1}) T_{1} (I - i T_{2} T_{1}) S_{3} \ket{\psi} \\
        & \quad \approxe{16 \eta} - S_{3} S_{1} (I + i S_{2} S_{1}) T_{1} (I - i T_{2} T_{1}) \ket{\psi} .
    \end{split}
    \end{align}
    \end{subequations}
    Therefore, we have
    \begin{equation}
        \begin{split}
            W S_{3} \ket{\psi}
            \approxe{12 \eta} \frac{1}{4} [
            &   \ket{00}_{\hilb{A}^{\prime} \hilb{B}^{\prime}} \otimes S_{3} (I - i S_{2} S_{1}) ( I + i T_{2} T_{1}) \ket{\psi} \\
            & + \ket{01}_{\hilb{A}^{\prime} \hilb{B}^{\prime}} \otimes S_{3} (I - i S_{2} S_{1}) T_{1} (I - i T_{2} T_{1}) \ket{\psi} \\
            & - \ket{10}_{\hilb{A}^{\prime} \hilb{B}^{\prime}} \otimes S_{3} S_{1} (I + i S_{2} S_{1}) (I + i T_{2} T_{1}) \ket{\psi} \\
            & - \ket{11}_{\hilb{A}^{\prime} \hilb{B}^{\prime}} \otimes S_{3} S_{1} (I + i S_{2} S_{1}) T_{1} (I - i T_{2} T_{1}) \ket{\psi}
            ] .
    \end{split}
    \end{equation}
    We can use the same robustness arguments as before (but this time being careful to note that each estimate is unaffected by the presence of the observables $S_{3}$ as it is unitary) to deduce that
    \begin{equation}
        \left\lVert W S_{3} \ket{\psi}
        - \sigma_{\mathrm{z}}^{\hilb{B}^{\prime}} \ket{\Phi^{+}}_{\hilb{A}^{\prime} \hilb{B}^{\prime}}
        \otimes S_{3} \ket{\varphi} \right\rVert \leq \frac{1}{4} \mathopen{}\left( 66 + 5 \sqrt{2} \right)\mathclose{} \eta .
    \end{equation}
    
    For the robust phase kickback stage, first note that by \cref{rel:corr_single,rel:acomm_single} we have
    \begin{equation}
        S_{3} (I + i T_{2} T_{1}) \ket{\psi}
        \approxe{6 \eta} T_{3} (I + i T_{2} T_{1}) \ket{\psi} .
    \end{equation}
    Using also that $(I \pm T_{3}) (I \pm T_{3}) = 2 (I \pm T_{3})$ and $(I \pm T_{3}) (I \mp T_{3}) = 0$, we get
    \begin{subequations}
    \begin{gather}
    \begin{split}
        & (I \pm S_{3}) (I \pm T_{3}) (I + i T_{2} T_{1}) \ket{\psi} \\
        & \quad \approxe{12 \eta} 2 (I \pm T_{3}) (I + i T_{2} T_{1}) \ket{\psi} ,
    \end{split} \\
        (I \pm S_{3}) (I \mp T_{3}) (I + i T_{2} T_{1}) \ket{\psi}
        \approxe{12 \eta} 0 .
    \end{gather}
    \end{subequations}
    Therefore, in place of \cref{eq:kickback_single_simplified}, we have the robust version
    \begin{equation}
    \begin{split}
        K \ket{\varphi}
        \overset{6 \sqrt{2} \eta}&{\approx} \ket{0}_{\hilb{A}^{\prime\prime}} \ket{0}_{\hilb{B}^{\prime\prime}}
        \otimes \frac{I + T_{3}}{2} \ket{\varphi} \\
        & \phantom{{}\approx{}} + \ket{1}_{\hilb{A}^{\prime\prime}} \ket{1}_{\hilb{B}^{\prime\prime}}
        \otimes \frac{I - T_{3}}{2} \ket{\varphi} .
    \end{split}
    \end{equation}
    To estimate $K S_{3} \ket{\varphi}$, we use (as in the ideal case) that $(I + S_{3}) S_{3} = (I + S_{3})$ and $(I - S_{3}) S_{3} = - (I - S_{3})$.
    A robust version of \cref{eq:kickback_single_z_simplified} given by
    \begin{equation}
    \begin{split}
        K S_{3} \ket{\varphi}
        \overset{6 \sqrt{2} \eta}&{\approx} \ket{0}_{\hilb{A}^{\prime\prime}} \ket{0}_{\hilb{B}^{\prime\prime}}
        \otimes \frac{I + T_{3}}{2} \ket{\varphi} \\
        & \phantom{{}\approx{}} - \ket{1}_{\hilb{A}^{\prime\prime}} \ket{1}_{\hilb{B}^{\prime\prime}}
        \otimes \frac{I - T_{3}}{2} \ket{\varphi}
    \end{split}
    \end{equation}
    is then immediate by identical argument to before.
\end{proof}

\section{Proof of many-copy self-test}
\label{sec:parallel_test_proof}

Before exhibiting the proof of \cref{thm:single_observable_isometry}, we show the following lemma, which will be used repeatedly to cancel terms corresponding to correlated complex conjugation of reference measurements at only some (but not all) positions.
We note that a similar argument was also used in \cite[Appendix~E]{bowles2018self} in part of a proof of an analogous result.

\begin{lemma}
\label{lem:global_conj}
    Let $\ket{\psi} \in \hilb{A} \otimes \hilb{B}$.
    Suppose for each $q \in \{1, 2, 3\}$ and $j \in \{1, \dots, n\}$ that there exist $\pm 1$-outcome observables $S_{q}^{(j)}$ on $\hilb{A}$ and $T_{q}^{(j)}$ on $\hilb{B}$ satisfying (for some $\eta \geq 0$) the following relations:
    \begin{enumerate}
        \item $\left( S_{q}^{(j)} - T_{q}^{(j)} \right) \ket{\psi} \approxe{\eta} 0$ for all $q$ and $j$.
        \item $\left\{ S_{q}^{(j)}, S_{r}^{(j)} \right\} \ket{\psi} \approxe{\eta} 0$ and $\left\{ T_{q}^{(j)}, T_{r}^{(j)} \right\} \ket{\psi} \approxe{\eta} 0$ for all $q,r$ and $j$ such that $q \neq r$.
        \item $\left[ S_{q}^{(j)}, S_{r}^{(k)} \right] \ket{\psi} \approxe{\eta} 0$ and $\left[ T_{q}^{(j)}, T_{r}^{(k)} \right] \ket{\psi} \approxe{\eta} 0$ for all $q,r$ and $j,k$ such that $j \neq k$.
    \end{enumerate}
    For all $j < n$, if (emulating the conclusion of \cref{prop:conj_rel}) it also holds that
    \begin{equation}
    \label{eq:conj_rel}
        \left( I + S_{1}^{(j)} S_{1}^{(j+1)} S_{2}^{(j)} S_{2}^{(j+1)} S_{3}^{(j)} S_{3}^{(j+1)} \right) \ket{\psi} \approxe{\eta} 0 ,
    \end{equation}
    then we have
    \begin{equation}
    \label{eq:vanshing_conj_terms_bob}
    \begin{split}
        \left( I \pm T_{3}^{(j)} \right)\mathclose{}
        & \mathopen{}\left( I + i T_{2}^{(j)} T_{1}^{(j)} \right)\mathclose{}
        \mathopen{}\left( I \mp T_{3}^{(j+1)} \right)\mathclose{} \\
        & \mathopen{}\left( I + i T_{2}^{(j+1)} T_{1}^{(j+1)} \right)
        \ket{\psi} \approxe{118 \eta} 0 .
    \end{split}
    \end{equation}
\end{lemma}
\begin{proof}
    The left-hand side of \cref{eq:vanshing_conj_terms_bob} can be expanded and then rewritten by grouping pairs of terms as
    \begin{equation}
    \begin{split}
        & \left( I + T_{3}^{(j)} T_{2}^{(j)} T_{1}^{(j)} T_{3}^{(j+1)} T_{2}^{(j+1)} T_{1}^{(j+1)} \right) \ket{\psi} \\
        & \quad \mp \left( T_{3}^{(j+1)} + T_{3}^{(j)} T_{2}^{(j)} T_{1}^{(j)} T_{2}^{(j+1)} T_{1}^{(j+1)} \right) \ket{\psi} \\
        & \quad + i \left( T_{2}^{(j+1)} T_{1}^{(j+1)} - T_{3}^{(j)} T_{2}^{(j)} T_{1}^{(j)} T_{3}^{(j+1)} \right) \ket{\psi} \\
        & \quad \pm i \left( T_{3}^{(j)} T_{2}^{(j+1)} T_{1}^{(j+1)} - T_{2}^{(j)} T_{1}^{(j)} T_{3}^{(j+1)} \right) \ket{\psi} \\
        & \quad \mp i \left( T_{3}^{(j+1)} T_{2}^{(j+1)} T_{1}^{(j+1)} - T_{3}^{(j)} T_{2}^{(j)} T_{1}^{(j)} \right) \ket{\psi} \\
        & \quad - \left( T_{2}^{(j)} T_{1}^{(j)} T_{2}^{(j+1)} T_{1}^{(j+1)} + T_{3}^{(j)} T_{3}^{(j+1)} \right) \ket{\psi} \\
        & \quad - i \left( T_{3}^{(j)} T_{3}^{(j+1)} T_{2}^{(j+1)} T_{1}^{(j+1)} - T_{2}^{(j)} T_{1}^{(j)} \right) \ket{\psi} \\
        & \quad \pm \left( T_{2}^{(j)} T_{1}^{(j)} T_{3}^{(j+1)} T_{2}^{(j+1)} T_{1}^{(j+1)} + T_{3}^{(j)} \right) \ket{\psi} .
    \end{split}
    \end{equation}
    By applying the given relations to \cref{eq:conj_rel}, it can be seen that each of the eight resulting terms approximately vanishes.
    Specifically, we have
    \begin{subequations}
    \begin{align}
        \left( I + T_{3}^{(j)} T_{2}^{(j)} T_{1}^{(j)} T_{3}^{(j+1)} T_{2}^{(j+1)} T_{1}^{(j+1)} \right) \ket{\psi} & \approxe{15 \eta} 0 , \\
        \left( T_{3}^{(j+1)} + T_{3}^{(j)} T_{2}^{(j)} T_{1}^{(j)} T_{2}^{(j+1)} T_{1}^{(j+1)} \right) \ket{\psi} & \approxe{10 \eta} 0 , \\
        \left( T_{2}^{(j+1)} T_{1}^{(j+1)} - T_{3}^{(j)} T_{2}^{(j)} T_{1}^{(j)} T_{3}^{(j+1)} \right) \ket{\psi} & \approxe{17 \eta} 0 , \\
        \left( T_{3}^{(j)} T_{2}^{(j+1)} T_{1}^{(j+1)} - T_{2}^{(j)} T_{1}^{(j)} T_{3}^{(j+1)} \right) \ket{\psi} & \approxe{17 \eta} 0 , \\
        \left( T_{3}^{(j+1)} T_{2}^{(j+1)} T_{1}^{(j+1)} - T_{3}^{(j)} T_{2}^{(j)} T_{1}^{(j)} \right) \ket{\psi} & \approxe{17 \eta} 0 , \\
        \left( T_{2}^{(j)} T_{1}^{(j)} T_{2}^{(j+1)} T_{1}^{(j+1)} + T_{3}^{(j)} T_{3}^{(j+1)} \right) \ket{\psi} & \approxe{10 \eta} 0 , \\
        \left( T_{3}^{(j)} T_{3}^{(j+1)} T_{2}^{(j+1)} T_{1}^{(j+1)} - T_{2}^{(j)} T_{1}^{(j)} \right) \ket{\psi} & \approxe{17 \eta} 0 , \\
        \left( T_{2}^{(j)} T_{1}^{(j)} T_{3}^{(j+1)} T_{2}^{(j+1)} T_{1}^{(j+1)} + T_{3}^{(j)} \right) \ket{\psi} & \approxe{15 \eta} 0 .
    \end{align}
    \end{subequations}
    The triangle inequality then gives \cref{eq:vanshing_conj_terms_bob}, as required.
\end{proof}

We now proceed with the main body of proof for \cref{thm:single_observable_isometry}.

\begin{proof}[Proof of \cref{thm:single_observable_isometry}]
    Consider the isometry and the corresponding notation introduced in \cref{fig:swap_kickback_xy_copy}.
    We begin by considering, for any $j \in \{1, \dots, n\}$, the action of the isometry $V^{(j)}$.
    For this, \cref{prop:swap_and_kickback_isoms} shows that after the ``swap'' stage of the circuit, given by $W^{(j)}$, we have that
    \begin{subequations}
    \label{eq:swapped_k_simple}
    \begin{align}
        \label{eq:swapped_k_simple_state}
        W^{(j)} \ket{\psi}
        & = \ket{\Phi^{+}}_{\hilb{A}_{j}^{\prime} \hilb{B}_{j}^{\prime}} \otimes \frac{1}{\sqrt{2}} \mathopen{}\left( I + i T_{2}^{(j)} T_{1}^{(j)} \right) \ket{\psi} , \\
        \label{eq:swapped_k_simple_x}
        W^{(j)} S_{1}^{(j)} \ket{\psi}
        & = \sigma_{\mathrm{x}}^{\hilb{B}_{j}^{\prime}} \ket{\Phi^{+}}_{\hilb{A}_{j}^{\prime} \hilb{B}_{j}^{\prime}} \otimes \frac{1}{\sqrt{2}} \mathopen{}\left( I + i T_{2}^{(j)} T_{1}^{(j)} \right) \ket{\psi} , \\
        \label{eq:swapped_k_simple_y}
        W^{(j)} S_{2}^{(j)} \ket{\psi}
        & = \sigma_{\mathrm{y}}^{\hilb{B}_{j}^{\prime}} \ket{\Phi^{+}}_{\hilb{A}_{j}^{\prime} \hilb{B}_{j}^{\prime}} \otimes \frac{1}{\sqrt{2}} \mathopen{}\left( I + i T_{2}^{(j)} T_{1}^{(j)} \right) \ket{\psi} , \\
        \label{eq:swapped_k_simple_z}
        W^{(j)} S_{3}^{(j)} \ket{\psi}
        & = \sigma_{\mathrm{z}}^{\hilb{B}_{j}^{\prime}} \ket{\Phi^{+}}_{\hilb{A}_{j}^{\prime} \hilb{B}_{j}^{\prime}} \otimes \frac{1}{\sqrt{2}} S_{3}^{(j)} \mathopen{}\left( I + i T_{2}^{(j)} T_{1}^{(j)} \right) \ket{\psi} ,
    \end{align}
    \end{subequations}

    We now apply the ``phase kickback'' stage of the isometry, given by $K^{(j)}$, to the expressions of \cref{eq:swapped_k_simple}.
    To make clearer the resulting equations, let us suppress the extracted state of the primed ancillae in our notation for now, as it is entirely unaffected by the remainder of the isometry.
    For this purpose, let us define (similarly to \cref{prop:swap_and_kickback_isoms})
    \begin{equation}
        \ket{\varphi^{j}} = \frac{1}{\sqrt{2}} \mathopen{}\left( I + i T_{2}^{(j)} T_{1}^{(j)} \right) \ket{\psi} .
    \end{equation}
    For \cref{eq:swapped_k_simple_state,eq:swapped_k_simple_x,eq:swapped_k_simple_y}, \cref{prop:swap_and_kickback_isoms} gives the action of $K^{(j)}$ as
    \begin{equation}
    \label{eq:kick_j}
    \begin{split}
        K^{(j)} \ket{\varphi^{j}}
        ={} & \ket{0}_{\hilb{A}_{j}^{\prime\prime}} \ket{0}_{\hilb{B}_{j}^{\prime\prime}} \\
        & \quad \otimes \frac{1}{2 \sqrt{2}} \left( I + T_{3}^{(j)} \right)\mathclose{} \mathopen{}\left( I + i T_{2}^{(j)} T_{1}^{(j)} \right) \ket{\psi} \\
        & + \ket{1}_{\hilb{A}_{j}^{\prime\prime}} \ket{1}_{\hilb{B}_{j}^{\prime\prime}} \\
        & \quad \otimes \frac{1}{2 \sqrt{2}} \left( I - T_{3}^{(j)} \right)\mathclose{} \mathopen{}\left( I + i T_{2}^{(j)} T_{1}^{(j)} \right) \ket{\psi} .
    \end{split}
    \end{equation}
    For \cref{eq:swapped_k_simple_z}, \cref{prop:swap_and_kickback_isoms} gives the action of $K^{(j)}$ as
    \begin{equation}
    \label{eq:kick_j_z}
    \begin{split}
        K^{(j)} S_{3}^{(j)} \ket{\varphi^{j}}
        ={} & \ket{0}_{\hilb{A}_{j}^{\prime\prime}} \ket{0}_{\hilb{B}_{j}^{\prime\prime}} \\
        & \quad \otimes \frac{1}{2 \sqrt{2}} \left( I + T_{3}^{(j)} \right)\mathclose{} \mathopen{}\left( I + i T_{2}^{(j)} T_{1}^{(j)} \right) \ket{\psi} \\
        & - \ket{1}_{\hilb{A}_{j}^{\prime\prime}} \ket{1}_{\hilb{B}_{j}^{\prime\prime}} \\
        & \quad \otimes \frac{1}{2 \sqrt{2}} \left( I - T_{3}^{(j)} \right)\mathclose{} \mathopen{}\left( I + i T_{2}^{(j)} T_{1}^{(j)} \right) \ket{\psi} .
    \end{split}
    \end{equation}
    We have now finished examining the action of $V^{(j)}$ on $\ket{\psi}$ and $S_{q}^{(j)} \ket{\psi}$.

    Notice that \cref{eq:kick_j,eq:kick_j_z} have a form consistent with that required by the junk state of \cref{eq:junk_form}, however, we have only yet extracted a single copy of $\ket{\Phi^{+}}$.
    Suppressing the ancillae once again when convenient, we will examine the action of $V^{(j)}$ on two subnormalized states of a similar form to those contained in the two terms of \cref{eq:kick_j,eq:kick_j_z}.
    Specifically, let us define for all $k \in \{1, \dots, n\}$ the vectors
    \begin{equation}
    \label{eq:iter_junk_def}
    \begin{split}
        \ket{\xi_{\pm}^{k}}_{\hilb{A} \hilb{B}}
        & = \frac{1}{\left( 2 \sqrt{2} \right)^{k}} \prod_{j=1}^{k} \mathopen{}\left( I \pm T_{3}^{(j)} \right)\mathclose{} \mathopen{}\left( I + i T_{2}^{(j)} T_{1}^{(j)} \right) \ket{\psi}_{\hilb{A} \hilb{B}} \\
        & = \prod_{j=1}^{k} J_{\pm}^{(j)} \ket{\psi}_{\hilb{A} \hilb{B}} ,
    \end{split}
    \end{equation}
    where
    \begin{equation}
        J_{\pm}^{(j)} = \frac{1}{2 \sqrt{2}} \mathopen{}\left( I \pm T_{3}^{(j)} \right)\mathclose{} \mathopen{}\left( I + i T_{2}^{(j)} T_{1}^{(j)} \right)\mathclose{} .
    \end{equation}
    We note that, with this notation, we can now combine \cref{eq:swapped_k_simple_state,eq:swapped_k_simple_x,eq:swapped_k_simple_y} with \cref{eq:kick_j} and \cref{eq:swapped_k_simple_z} with \cref{eq:kick_j_z} to write
    \begin{subequations}
    \begin{align}
        \label{eq:single_isom_state}
        V^{(j)} \ket{\psi}
        & = \ket{\Phi^{+}}_{\hilb{A}_{j}^{\prime} \hilb{B}_{j}^{\prime}}
        \otimes \left( \ket{0}_{\hilb{A}_{j}^{\prime\prime}} \ket{0}_{\hilb{B}_{j}^{\prime\prime}}
        \otimes J_{+}^{(j)} \ket{\psi}
        + \ket{1}_{\hilb{A}_{j}^{\prime\prime}} \ket{1}_{\hilb{B}_{j}^{\prime\prime}}
        \otimes J_{-}^{(j)} \ket{\psi} \right)\mathclose{} , \\
        \label{eq:single_isom_obs}
        V^{(j)} S_{q}^{(j)} \ket{\psi}
        & = \sigma_{q}^{\hilb{B}_{j}^{\prime}} \ket{\Phi^{+}}_{\hilb{A}_{j}^{\prime} \hilb{B}_{j}^{\prime}}
        \otimes \left( \ket{0}_{\hilb{A}_{j}^{\prime\prime}} \ket{0}_{\hilb{B}_{j}^{\prime\prime}} \otimes J_{+}^{(j)} \ket{\psi}
        + (-1)^{[q=3]} \ket{1}_{\hilb{A}_{j}^{\prime\prime}} \ket{1}_{\hilb{B}_{j}^{\prime\prime}} \otimes J_{-}^{(j)} \ket{\psi} \right)\mathclose{} .
    \end{align}
    \end{subequations}
    In the special case of $j=1$, since $\ket{\xi_{\pm}^{1}} = J_{\pm}^{(1)} \ket{\psi}$, we recover
    \begin{subequations}
    \label{eq:single_isom_proof}
    \begin{align}
        \label{eq:single_isom_1_state_proof}
        V^{(1)} \ket{\psi}
        & = \ket{\Phi^{+}}_{\hilb{A}_{1}^{\prime} \hilb{B}_{1}^{\prime}}
        \otimes \left( \ket{0}_{\hilb{A}_{1}^{\prime\prime}} \ket{0}_{\hilb{B}_{1}^{\prime\prime}}
        \otimes \ket{\xi_{+}^{1}}
        + \ket{1}_{\hilb{A}_{1}^{\prime\prime}} \ket{1}_{\hilb{B}_{1}^{\prime\prime}}
        \otimes \ket{\xi_{-}^{1}} \right)\mathclose{} , \\
        \label{eq:single_isom_1_obs_proof}
        V^{(1)} S_{q}^{(1)} \ket{\psi}
        & = \sigma_{q}^{\hilb{B}_{1}^{\prime}} \ket{\Phi^{+}}_{\hilb{A}_{1}^{\prime} \hilb{B}_{1}^{\prime}}
        \otimes \left( \ket{0}_{\hilb{A}_{1}^{\prime\prime}} \ket{0}_{\hilb{B}_{1}^{\prime\prime}} \otimes \ket{\xi_{+}^{1}}
        + (-1)^{[q=3]} \ket{1}_{\hilb{A}_{1}^{\prime\prime}} \ket{1}_{\hilb{B}_{1}^{\prime\prime}} \otimes \ket{\xi_{-}^{1}} \right)\mathclose{} .
    \end{align}
    \end{subequations}

    We now examine $V^{(k)} \ket{\xi_{\pm}^{k-1}}_{\hilb{A} \hilb{B}}$, where $1 < k \leq n $.
    We begin by showing that
    \begin{equation}
    \label{eq:isom_bob_iter_comm}
        V_{\hilb{B}}^{(k)} \ket{\xi_{\pm}^{k-1}}
        = \left( \prod_{j=1}^{k-1} J_{\pm}^{(j)} \right)\mathclose{} V_{\hilb{B}}^{(k)} \ket{\psi} .
    \end{equation}
    To do this, notice by the definition of $V_{\hilb{B}}^{(k)} = K_{\hilb{B}}^{(k)} W_{\hilb{B}}^{(k)}$ given in \cref{fig:swap_kickback_xy_copy} that
    \begin{equation}
    \begin{split}
        V_{\hilb{B}}^{(k)} \ket{\psi}
        = \frac{1}{4} \Bigl[ &
        \ket{0}_{\hilb{B}_{k}^{\prime}} \otimes \ket{0}_{\hilb{B}_{k}^{\prime\prime}} \otimes \left( I + i T_{2}^{(k)} T_{1}^{(k)} \right) \ket{\psi}
        + \ket{0}_{\hilb{B}_{k}^{\prime}} \otimes \ket{1}_{\hilb{B}_{k}^{\prime\prime}} \otimes T_{3}^{(k)} \mathopen{}\left( I + i T_{2}^{(k)} T_{1}^{(k)} \right) \ket{\psi} \\
        & + \ket{1}_{\hilb{B}_{k}^{\prime}} \otimes \ket{0}_{\hilb{B}_{k}^{\prime\prime}} \otimes T_{1}^{(k)} \mathopen{}\left( I - i T_{2}^{(k)} T_{1}^{(k)} \right) \ket{\psi}
        + \ket{1}_{\hilb{B}_{k}^{\prime}} \otimes \ket{1}_{\hilb{B}_{k}^{\prime\prime}} \otimes T_{3}^{(k)} T_{1}^{(k)} \mathopen{}\left( I - i T_{2}^{(k)} T_{1}^{(k)} \right) \ket{\psi}
        \Bigr] .
    \end{split}
    \end{equation}
    After applying all the $J_{\pm}^{(j)}$ on the left and bringing operators with index $k$ past all operators with other indices to the front via a chain of state-dependent commutation and swapping operators between Alice and Bob (\cref{rel:corr,rel:comm}) we get
    \begin{equation}
    \begin{split}
        \left( \prod_{j=1}^{k-1} J_{\pm}^{(j)} \right)\mathclose{} V_{\hilb{B}}^{(k)} \ket{\psi}
        = \frac{1}{4} \Bigl[ &
        \ket{0}_{\hilb{B}_{k}^{\prime}} \otimes \ket{0}_{\hilb{B}_{k}^{\prime\prime}} \otimes \left( I + i T_{2}^{(k)} T_{1}^{(k)} \right) \prod_{j=1}^{k-1} J_{\pm}^{(j)} \ket{\psi} \\
        & + \ket{0}_{\hilb{B}_{k}^{\prime}} \otimes \ket{1}_{\hilb{B}_{k}^{\prime\prime}} \otimes T_{3}^{(k)} \mathopen{}\left( I + i T_{2}^{(k)} T_{1}^{(k)} \right) \prod_{j=1}^{k-1} J_{\pm}^{(j)} \ket{\psi} \\
        & + \ket{1}_{\hilb{B}_{k}^{\prime}} \otimes \ket{0}_{\hilb{B}_{k}^{\prime\prime}} \otimes T_{1}^{(k)} \mathopen{}\left( I - i T_{2}^{(k)} T_{1}^{(k)} \right) \prod_{j=1}^{k-1} J_{\pm}^{(j)} \ket{\psi} \\
        & + \ket{1}_{\hilb{B}_{k}^{\prime}} \otimes \ket{1}_{\hilb{B}_{k}^{\prime\prime}} \otimes T_{3}^{(k)} T_{1}^{(k)} \mathopen{}\left( I - i T_{2}^{(k)} T_{1}^{(k)} \right) \prod_{j=1}^{k-1} J_{\pm}^{(j)} \ket{\psi}
        \Bigr] .
    \end{split}
    \end{equation}
    By \cref{eq:iter_junk_def}, and again by the definition of $V_{\hilb{B}}^{(k)}$, the right-hand side is simply $V_{\hilb{B}}^{(k)} \ket{\xi_{\pm}^{k-1}}$.
    This is the desired \cref{eq:isom_bob_iter_comm}.
    Since all $J_{\pm}^{(j)}$ act on Bob's subsystem, they commute with $V_{\hilb{A}}^{(k)}$.
    We can then apply $V_{\hilb{A}}^{(k)}$ to both sides of \cref{eq:isom_bob_iter_comm} to get
    \begin{equation}
        V^{(k)} \ket{\xi_{\pm}^{k-1}}
        = \left( \prod_{j=1}^{k-1} J_{\pm}^{(j)} \right)\mathclose{} V^{(k)} \ket{\psi} .
    \end{equation}
    Substituting \cref{eq:single_isom_state} for $V^{(k)} \ket{\psi}$ then gives
    \begin{equation}
    \label{eq:isom_full_simple}
        V^{(k)} \ket{\xi_{\pm}^{k-1}}
        = \ket{\Phi^{+}}_{\hilb{A}_{k}^{\prime} \hilb{B}_{k}^{\prime}}
        \otimes \left( \prod_{j=1}^{k-1} J_{\pm}^{(j)} \right)\mathclose{}
        \mathopen{}\left(
        \ket{0}_{\hilb{A}_{k}^{\prime\prime}} \ket{0}_{\hilb{B}_{k}^{\prime\prime}}
        \otimes J_{+}^{(k)} \ket{\psi}
        + \ket{1}_{\hilb{A}_{k}^{\prime\prime}} \ket{1}_{\hilb{B}_{k}^{\prime\prime}}
        \otimes J_{-}^{(k)} \ket{\psi}
        \right)\mathclose{} .
    \end{equation}
    \Cref{lem:global_conj} implies
    \begin{equation}
    \begin{split}
        \left( I \pm T_{3}^{(k-1)} \right)\mathclose{}
        & \mathopen{}\left( I + i T_{2}^{(k-1)} T_{1}^{(k-1)} \right)\mathclose{} \\
        & \mathopen{}\left( I \mp T_{3}^{(k)} \right)\mathclose{}
        \mathopen{}\left( I + i T_{2}^{(k)} T_{1}^{(k)} \right) \ket{\psi} = 0
    \end{split}
    \end{equation}
    which, rewriting in terms of the $J_{\pm}^{(j)}$, then implies
    \begin{equation}
    \label{eq:conj_cancel}
        J_{\pm}^{(k-1)} J_{\mp}^{(k)} \ket{\psi} = 0 .
    \end{equation}
    Thus, using this to simplify \cref{eq:isom_full_simple}, we have
    \begin{subequations}
    \label{eq:isom_step_proof}
    \begin{align}
        V^{(k)} \ket{\xi_{+}^{k-1}}
        & = \ket{\Phi^{+}}_{\hilb{A}_{k}^{\prime} \hilb{B}_{k}^{\prime}}
        \otimes \ket{0}_{\hilb{A}_{k}^{\prime\prime}} \ket{0}_{\hilb{B}_{k}^{\prime\prime}}
        \otimes \ket{\xi_{+}^{k}} , \\
        V^{(k)} \ket{\xi_{-}^{k-1}}
        & = \ket{\Phi^{+}}_{\hilb{A}_{k}^{\prime} \hilb{B}_{k}^{\prime}}
        \otimes \ket{1}_{\hilb{A}_{k}^{\prime\prime}} \ket{1}_{\hilb{B}_{k}^{\prime\prime}}
        \otimes \ket{\xi_{-}^{k}} .
    \end{align}
    \end{subequations}

    By definition of $V_{\hilb{A}}^{(k)}$,
    \begin{equation}
    \begin{split}
        V_{\hilb{A}}^{(k)} \ket{\psi}
        = \frac{1}{4} \Bigl[ &
        \ket{0}_{\hilb{A}_{k}^{\prime}} \otimes \ket{0}_{\hilb{A}_{k}^{\prime\prime}} \otimes \left( I - i S_{2}^{(k)} S_{1}^{(k)} \right) \ket{\psi} \\
        & + \ket{0}_{\hilb{A}_{k}^{\prime}} \otimes \ket{1}_{\hilb{A}_{k}^{\prime\prime}} \otimes S_{3}^{(k)} \mathopen{}\left( I - i S_{2}^{(k)} S_{1}^{(k)} \right) \ket{\psi} \\
        & + \ket{1}_{\hilb{A}_{k}^{\prime}} \otimes \ket{0}_{\hilb{A}_{k}^{\prime\prime}} \otimes S_{1}^{(k)} \mathopen{}\left( I + i S_{2}^{(k)} S_{1}^{(k)} \right) \ket{\psi} \\
        & + \ket{1}_{\hilb{A}_{k}^{\prime}} \otimes \ket{1}_{\hilb{A}_{k}^{\prime\prime}} \otimes S_{3}^{(k)} S_{1}^{(k)} \mathopen{}\left( I + i S_{2}^{(k)} S_{1}^{(k)} \right) \ket{\psi}
        \Bigr] .
    \end{split}
    \end{equation}
    Therefore, using \cref{rel:corr,rel:comm}, whenever $j \neq k$
    \begin{equation}
        V_{\hilb{A}}^{(k)} S_{q}^{(j)} \ket{\psi}
        = S_{q}^{(j)} V_{\hilb{A}}^{(k)} \ket{\psi} .
    \end{equation}
    From this, it follows by noting $V_{\hilb{B}}^{k}$ commutes with $S_{q}^{(j)}$ that
    \begin{equation}
    \label{eq:single_isom_state_switch}
        V^{(k)} S_{q}^{(j)} \ket{\psi} = S_{q}^{(j)} V^{(k)} \ket{\psi} .
    \end{equation}
    A special case is (remembering that $1 < k \leq n$)
    \begin{equation}
    \label{eq:single_isom_1_state_switch_proof}
        V^{(1)} S_{q}^{(k)} \ket{\psi} = S_{q}^{(k)} V^{(1)} \ket{\psi} .
    \end{equation}
    Furthermore, again since $V_{\hilb{A}}^{k}$ and $S_{q}^{(j)}$ defined on Alice's side commute with $V_{\hilb{B}}^{(k)}$ and all $J_{\pm}^{(i)}$, and by applying \cref{eq:single_isom_state_switch},
    \begin{equation}
    \label{eq:single_isom_junk_switch_proof}
    \begin{split}
        V^{(k)} S_{q}^{(j)} \ket{\xi_{\pm}^{k-1}}
        & = \left( V_{\hilb{B}}^{(k)} \prod_{i=1}^{k-1} J_{\pm}^{(i)} \right)\mathclose{} V_{\hilb{A}}^{(k)} S_{q}^{(j)} \ket{\psi} \\
        & = \left( V_{\hilb{B}}^{(k)} \prod_{i=1}^{k-1} J_{\pm}^{(i)} \right)\mathclose{} S_{q}^{(j)} V_{\hilb{A}}^{(k)} \ket{\psi} \\
        & = S_{q}^{(j)} V^{(k)} \ket{\xi_{\pm}^{k-1}} .
    \end{split}
    \end{equation}

    Finally, it follows similarly to \cref{eq:isom_step_proof} that
    \begin{subequations}
    \label{eq:isom_step_obs_proof}
    \begin{align}
    \begin{split}
        V^{(k)} S_{q}^{(k)} \ket{\xi_{+}^{k-1}}
        & = \sigma_{q}^{\hilb{B}_{k}^{\prime}} \ket{\Phi^{+}}_{\hilb{A}_{k}^{\prime} \hilb{B}_{k}^{\prime}} \\
        & \phantom{{}={}} \otimes \ket{0}_{\hilb{A}_{k}^{\prime\prime}} \ket{0}_{\hilb{B}_{k}^{\prime\prime}}
        \otimes \ket{\xi_{+}^{k}} ,
    \end{split} \\
    \begin{split}
        V^{(k)} S_{q}^{(k)} \ket{\xi_{-}^{k-1}}
        & = (-1)^{[q=3]} \sigma_{q}^{\hilb{B}_{k}^{\prime}} \ket{\Phi^{+}}_{\hilb{A}_{k}^{\prime} \hilb{B}_{k}^{\prime}} \\
        & \phantom{{}={}} \otimes \ket{1}_{\hilb{A}_{k}^{\prime\prime}} \ket{1}_{\hilb{B}_{k}^{\prime\prime}}
        \otimes \ket{\xi_{-}^{k}} .
    \end{split}
    \end{align}
    \end{subequations}
    This is because acting with $S_{q}^{(k)}$ followed by $V_{\hilb{A}}^{(k)}$ on both sides of \cref{eq:isom_bob_iter_comm} gives
    \begin{equation}
        V^{(k)} S_{q}^{(k)} \ket{\xi_{\pm}^{k-1}}
        = \left( \prod_{j=1}^{k-1} J_{\pm}^{(j)} \right)\mathclose{} V^{(k)} S_{q}^{(k)} \ket{\psi} .
    \end{equation}
    Substituting \cref{eq:single_isom_obs} for $V^{(k)} S_{q}^{(k)} \ket{\psi}$ then gives
    \begin{equation}
        V^{(k)} S_{q}^{(k)} \ket{\xi_{\pm}^{k-1}}
        = \sigma_{q}^{\hilb{B}_{k}^{\prime}} \ket{\Phi^{+}}_{\hilb{A}_{k}^{\prime} \hilb{B}_{k}^{\prime}}
        \otimes \left( \prod_{j=1}^{k-1} J_{\pm}^{(j)} \right)\mathclose{}
        \mathopen{}\Bigl(
        \ket{0}_{\hilb{A}_{k}^{\prime\prime}} \ket{0}_{\hilb{B}_{k}^{\prime\prime}}
        \otimes J_{+}^{(k)} \ket{\psi}
        + (-1)^{[q=3]} \ket{1}_{\hilb{A}_{k}^{\prime\prime}} \ket{1}_{\hilb{B}_{k}^{\prime\prime}}
        \otimes J_{-}^{(k)} \ket{\psi}
        \Bigr)\mathclose{} .
    \end{equation}
    Thus, using \cref{eq:conj_cancel} to simplify this, we have the desired \cref{eq:isom_step_obs_proof}.

    After the full application of the isometry $V = V^{(n)} \dots V^{(1)}$, and defining
    \begin{subequations}
    \begin{align}
        \ket{0}_{\hilb{A^{\prime\prime}}} & = \ket{0 \dots 0}_{\hilb{A}^{\prime\prime}} , \\
        \ket{1}_{\hilb{A^{\prime\prime}}} & = \ket{1 \dots 1}_{\hilb{A}^{\prime\prime}} , \\
        \ket{0}_{\hilb{B^{\prime\prime}}} & = \ket{0 \dots 0}_{\hilb{B}^{\prime\prime}} , \\
        \ket{1}_{\hilb{B^{\prime\prime}}} & = \ket{1 \dots 1}_{\hilb{B}^{\prime\prime}} ,
    \end{align}
    \end{subequations}
    \cref{eq:single_isom_1_state_proof,eq:isom_step_proof} together give
    \begin{equation}
    \label{eq:isom_state_result_proof}
    \begin{split}
        V \ket{\psi}
        = \bigotimes_{j=1}^{n} \ket{\Phi^{+}}_{\hilb{A}_{j}^{\prime} \hilb{B}_{j}^{\prime}}
        \otimes \bigl( &
        \ket{0}_{\hilb{A}^{\prime\prime}} \ket{0}_{\hilb{B}^{\prime\prime}} \otimes \ket{\xi_{+}^{n}} \\
        & + \ket{1}_{\hilb{A}^{\prime\prime}} \ket{1}_{\hilb{B}^{\prime\prime}} \otimes \ket{\xi_{-}^{n}}
        \bigr)\mathclose{} .
    \end{split}
    \end{equation}
    Similarly, we have
    \begin{equation}
    \label{eq:isom_obs_keq1_proof}
    \begin{split}
        V S_{q}^{(1)} \ket{\psi}
        & = \sigma_{q}^{\hilb{B}_{1}^{\prime}} \ket{\Phi^{+}}_{\hilb{A}_{1}^{\prime} \hilb{B}_{1}^{\prime}} \otimes V^{(n)} \dots V^{(2)} \mathopen{}\left( \ket{0}_{\hilb{A}_{1}^{\prime\prime}} \ket{0}_{\hilb{B}_{1}^{\prime\prime}} \otimes \ket{\xi_{+}^{1}} + (-1)^{[q=3]} \ket{1}_{\hilb{A}_{1}^{\prime\prime}} \ket{1}_{\hilb{B}_{1}^{\prime\prime}} \otimes \ket{\xi_{-}^{1}} \right) \\
        & = \sigma_{q}^{\hilb{B}_{1}^{\prime}} \bigotimes_{j=1}^{n} \ket{\Phi^{+}}_{\hilb{A}_{j}^{\prime} \hilb{B}_{j}^{\prime}} \otimes \left( \ket{0}_{\hilb{A}^{\prime\prime}} \ket{0}_{\hilb{B}^{\prime\prime}} \otimes \ket{\xi_{+}^{n}} + (-1)^{[q=3]} \ket{1}_{\hilb{A}^{\prime\prime}} \ket{1}_{\hilb{B}^{\prime\prime}} \otimes \ket{\xi_{-}^{n}} \right) \\
        & = \sigma_{q}^{\hilb{B}_{1}^{\prime}} \bigotimes_{j=1}^{n} \ket{\Phi^{+}}_{\hilb{A}_{j}^{\prime} \hilb{B}_{j}^{\prime}} \otimes \sigma_{3[q=3]}^{\hilb{B}^{\prime\prime}} \mathopen{}\left( \ket{0}_{\hilb{A}^{\prime\prime}} \ket{0}_{\hilb{B}^{\prime\prime}} \otimes \ket{\xi_{+}^{n}} + \ket{1}_{\hilb{A}^{\prime\prime}} \ket{1}_{\hilb{B}^{\prime\prime}} \otimes \ket{\xi_{-}^{n}} \right)\mathclose{} .
    \end{split}
    \end{equation}
    The first equality follows from \cref{eq:single_isom_1_obs_proof} and the second equality from \cref{eq:isom_step_proof}.
    Furthermore, for $1 < k \leq n$, we can write
    \begin{equation}
    \label{eq:isom_obs_kgt1_proof}
    \begin{split}
        V S_{q}^{(k)} \ket{\psi}
        & = \ket{\Phi^{+}}_{\hilb{A}_{1}^{\prime} \hilb{B}_{1}^{\prime}} \otimes V^{(n)} \dots V^{(2)} S_{q}^{(k)} \mathopen{}\left( \ket{0}_{\hilb{A}_{1}^{\prime\prime}} \ket{0}_{\hilb{B}_{1}^{\prime\prime}} \otimes \ket{\xi_{+}^{1}} + \ket{1}_{\hilb{A}_{1}^{\prime\prime}} \ket{1}_{\hilb{B}_{1}^{\prime\prime}} \otimes \ket{\xi_{-}^{1}} \right) \\
        & = \bigotimes_{j=1}^{k-1} \ket{\Phi^{+}}_{\hilb{A}_{j}^{\prime} \hilb{B}_{j}^{\prime}} \otimes V^{(n)} \dots V^{(k)} S_{q}^{(k)} \mathopen{}\left( \ket{0 \dots 0} \ket{0 \dots 0} \otimes \ket{\xi_{+}^{k-1}} + \ket{1 \dots 1} \ket{1 \dots 1} \otimes \ket{\xi_{-}^{k-1}} \right) \\
        & = \sigma_{q}^{\hilb{B}_{k}^{\prime}} \bigotimes_{j=1}^{n} \ket{\Phi^{+}}_{\hilb{A}_{j}^{\prime} \hilb{B}_{j}^{\prime}} \otimes \left( \ket{0}_{\hilb{A}^{\prime\prime}} \ket{0}_{\hilb{B}^{\prime\prime}} \otimes \ket{\xi_{+}^{n}} + (-1)^{[q=3]} \ket{1}_{\hilb{A}^{\prime\prime}} \ket{1}_{\hilb{B}^{\prime\prime}} \otimes \ket{\xi_{-}^{n}} \right) \\
        & = \sigma_{q}^{\hilb{B}_{k}^{\prime}} \bigotimes_{j=1}^{n} \ket{\Phi^{+}}_{\hilb{A}_{j}^{\prime} \hilb{B}_{j}^{\prime}} \otimes \sigma_{3[q=3]}^{\hilb{B}^{\prime\prime}} \mathopen{}\left( \ket{0}_{\hilb{A}^{\prime\prime}} \ket{0}_{\hilb{B}^{\prime\prime}} \otimes \ket{\xi_{+}^{n}} + \ket{1}_{\hilb{A}^{\prime\prime}} \ket{1}_{\hilb{B}^{\prime\prime}} \otimes \ket{\xi_{-}^{n}} \right)\mathclose{} .
    \end{split}
    \end{equation}
    For the first equality we used \cref{eq:single_isom_1_state_switch_proof,eq:single_isom_1_state_proof}; the second equality used \cref{eq:single_isom_junk_switch_proof,eq:isom_step_proof}; and the third equality used \cref{eq:isom_step_obs_proof,eq:isom_step_proof}.
    Together, \cref{eq:isom_state_result_proof,eq:isom_obs_keq1_proof,eq:isom_obs_kgt1_proof} have the desired form by taking $\ket{\xi_{0}} = \ket{\xi_{+}^{n}}$ and $\ket{\xi_{1}} = \ket{\xi_{-}^{n}}$.
\end{proof}

\section{Action of many untrusted operators}
\label{sec:many_operator_action}

\begin{proof}[Proof of \cref{lem:single_to_multiple_observable}]
    Consider some unitary operators $U_{\tilde{\hilb{A}}} \colon \tilde{\hilb{A}} \to \tilde{\hilb{A}}$ and $U_{\tilde{\hilb{B}}} \colon \tilde{\hilb{B}} \to \tilde{\hilb{B}}$ which extend the isometries $V_{\hilb{A}}$ and $V_{\hilb{B}}$ to have domains $\tilde{\hilb{A}}$ and $\tilde{\hilb{B}}$, respectively.
    This can be achieved by extending orthonormal bases of the images of each isometry to orthonormal bases of each full space.
    Define the local unitary $U = U_{\tilde{\hilb{A}}} \otimes U_{\tilde{\hilb{B}}}$ on $\tilde{\hilb{A}} \otimes \tilde{\hilb{B}}$.
    We may also consider the trivial extension (by appropriate direct sums) of all operators $A_{j}$ to $\tilde{\hilb{A}}$ and the state $\ket{\psi}$ to $\tilde{\hilb{A}} \otimes \tilde{\hilb{B}}$, each with zero weight in their new components.
    This preserves the norms of the state and each operator.
    Reusing the notation of the original state and operators also for their trivial extensions, we can now write that for all $j$
    \begin{subequations}
    \begin{align}
        U \ket{\psi} & = V \ket{\psi} , \\
        U A_{j} \ket{\psi} & = V A_{j} \ket{\psi} .
    \end{align}
    \end{subequations}

    From the assumption of \cref{eq:isom_state} and that $\lVert \tilde{B}_{j} \rVert \leq 1$, it follows that
    \begin{equation}
    \label{eq:switch_state}
        \tilde{B}_{j} U \ket{\psi}
        \approxe{\delta} \tilde{B}_{j} \ket{\phi} .
    \end{equation}
    Using the assumptions of \cref{eq:ref_stab,eq:isom_obs}, we can write
    \begin{equation}
    \label{eq:switch_obs}
        \tilde{B}_{j} \ket{\phi}
        = \tilde{A}_{j} \ket{\phi}
        \approxe{\delta} U A_{j} \ket{\psi} .
    \end{equation}
    Thus, combining \cref{eq:switch_state,eq:switch_obs} using the triangle inequality yields
    \begin{equation}
    \label{eq:switch_obs_on_ref_state}
        \tilde{B}_{j} U \ket{\psi}
        \approxe{2 \delta} U A_{j} \ket{\psi} .
    \end{equation}
    Since $V_{\hilb{A}}$ is an isometry, $V_{\hilb{A}}^{\dagger} V_{\hilb{A}} = I_{\hilb{A}}$, where $I_{\hilb{A}}$ is the identity operator on $\hilb{A}$.
    Thus,
    \begin{equation}
    \begin{split}
        V A_{j}
        & = V_{\hilb{A}} A_{j} \otimes V_{\hilb{B}} \\
        & = V_{\hilb{A}} A_{j} V_{\hilb{A}}^{\dagger} V_{\hilb{A}} \otimes V_{\hilb{B}} \\
        & = V_{\hilb{A}} A_{j} V_{\hilb{A}}^{\dagger} V .
    \end{split}
    \end{equation}
    Similarly, $U A_{j} = U_{\tilde{\hilb{A}}} A_{j} U_{\tilde{\hilb{A}}}^{\dagger} U$.
    We can therefore rewrite \cref{eq:switch_obs_on_ref_state} as
    \begin{equation}
    \label{eq:switch_all}
        \tilde{B}_{j} U \ket{\psi}
        \approxe{2 \delta} \left( U_{\tilde{\hilb{A}}} A_{j} U_{\tilde{\hilb{A}}}^{\dagger} \right)\mathclose{} U \ket{\psi} .
    \end{equation}
    We now use the properties just exhibited to examine the state $\tilde{A}_{1} \dots \tilde{A}_{m} \ket{\phi}$.
    Repeated use of \cref{eq:ref_stab} and the fact that operators defined on $\tilde{\hilb{A}}$ commute with those defined on $\tilde{\hilb{B}}$ gives
    \begin{equation}
        \tilde{A}_{1} \dots \tilde{A}_{m} \ket{\phi}
        = \tilde{B}_{m} \dots \tilde{B}_{1} \ket{\phi} .
    \end{equation}
    Using \cref{eq:isom_state} one time, and again that $\lVert \tilde{B}_{j} \rVert \leq 1$, we can then write
    \begin{equation}
        \tilde{A}_{1} \dots \tilde{A}_{m} \ket{\phi}
        \approxe{\delta} \left( \tilde{B}_{m} \dots \tilde{B}_{1} \right)\mathclose{} U \ket{\psi} .
    \end{equation}
    Repeated use of \cref{eq:switch_all}, noting that $\lVert U_{\tilde{\hilb{A}}} A_{j} U_{\tilde{\hilb{A}}}^{\dagger} \rVert \leq \lVert A_{j} \rVert \leq 1$ for all $j$ since $U_{\tilde{\hilb{A}}}$ is unitary, gives via the triangle inequality
    \begin{equation}
        \left( \tilde{B}_{m} \dots \tilde{B}_{1} \right)\mathclose{} U \ket{\psi}
        \approxe{2 m \delta} \left( U_{\tilde{\hilb{A}}} A_{1} U_{\tilde{\hilb{A}}}^{\dagger} \right) \dots \left( U_{\tilde{\hilb{A}}} A_{m} U_{\tilde{\hilb{A}}}^{\dagger} \right)\mathclose{} U \ket{\psi} .
    \end{equation}
    Therefore,
    \begin{equation}
    \label{eq:expanded_isom}
        \tilde{A}_{1} \dots \tilde{A}_{m} \ket{\phi}
        \approxe{(2m + 1) \delta} \left( U_{\tilde{\hilb{A}}} A_{1} U_{\tilde{\hilb{A}}}^{\dagger} \right) \dots \left( U_{\tilde{\hilb{A}}} A_{m} U_{\tilde{\hilb{A}}}^{\dagger} \right)\mathclose{} U \ket{\psi} .
    \end{equation}
    Again using the fact that $V_{\hilb{A}}^{\dagger} V_{\hilb{A}} = I_{\hilb{A}}$, we have
    \begin{equation}
    \begin{split}
        \left( V_{\hilb{A}} A_{1} V_{\hilb{A}}^{\dagger} \right) \dots \left( V_{\hilb{A}} A_{m} V_{\hilb{A}}^{\dagger} \right)\mathclose{} V
        & = V_{\hilb{A}} (A_{1} \dots A_{m}) V_{\hilb{A}}^{\dagger} V \\
        & = V_{\hilb{A}} (A_{1} \dots A_{m}) V_{\hilb{A}}^{\dagger} V_{\hilb{A}} \otimes V_{\hilb{B}} \\
        & = V_{\hilb{A}} (A_{1} \dots A_{m}) \otimes V_{\hilb{B}} \\
        & = V (A_{1} \dots A_{m}) ,
    \end{split}
    \end{equation}
    and similarly
    \begin{equation}
        \left( U_{\tilde{\hilb{A}}} A_{1} U_{\tilde{\hilb{A}}}^{\dagger} \right) \dots \left( U_{\tilde{\hilb{A}}} A_{m} U_{\tilde{\hilb{A}}}^{\dagger} \right)\mathclose{} U
        = U (A_{1} \dots A_{m}) .
    \end{equation}
    Thus, \cref{eq:expanded_isom} becomes
    \begin{equation}
        \tilde{A}_{1} \dots \tilde{A}_{m} \ket{\phi}
        \approxe{(2m + 1) \delta} U (A_{1} \dots A_{m}) \ket{\psi} .
    \end{equation}
    Due to the construction of $U$ and the extended versions of the operators $A_{j}$ and the state $\ket{\psi}$, we have
    $U (A_{1} \dots A_{m}) \ket{\psi} = V (A_{1} \dots A_{m}) \ket{\psi}$.
    Therefore,
    \begin{equation}
        \tilde{A}_{1} \dots \tilde{A}_{m} \ket{\phi}
        \approxe{(2m + 1) \delta} V (A_{1} \dots A_{m}) \ket{\psi}
    \end{equation}
    as required.
\end{proof}

\section{Proof of state preparation}
\label{sec:state_preparation_proof}

\begin{proof}[Proof of \cref{thm:state_preparation}]
    Denote the state
    \begin{equation}
        \ket{\psi^{\prime}}
        = \bigotimes_{j=1}^{n} \ket{\Phi^{+}}_{\hilb{A}_{j}^{\prime} \hilb{B}_{j}^{\prime}} \otimes \ket{\xi}
    \end{equation}
    and projective measurement operators
    \begin{equation}
    \begin{split}
        \hat{N}_{\bm a \mid \bm \chi}
        & = \bigotimes_{j=1}^{n} \proj{\sigma_{\chi_{j}}^{a_{j}}}_{\hilb{B}_{j}^{\prime}} \otimes \proj{0}_{\hilb{B}^{\prime\prime}} \\
        & \phantom{{}={}} + \bigotimes_{j=1}^{n} \proj{\sigma_{\chi_{j}}^{a_{j}}}_{\hilb{B}_{j}^{\prime}}^{*} \otimes \proj{1}_{\hilb{B}^{\prime\prime}} \\
        & = \bigotimes_{j=1}^{n} \proj{\sigma_{\chi_{j}}^{a_{j}}}_{\hilb{B}_{j}^{\prime}} \otimes \proj{0}_{\hilb{B}^{\prime\prime}} \\
        & \phantom{{}={}} + \bigotimes_{j=1}^{n} \proj*{\sigma_{\chi_{j}}^{a_{j} (-1)^{[\chi_{j} = \mathrm{z}]}}}_{\hilb{B}_{j}^{\prime}} \otimes \proj{1}_{\hilb{B}^{\prime\prime}} .
    \end{split}
    \end{equation}
    For any $\chi \in \{1, \dots, 5\}$, the Bell state $\ket{\Phi^{+}}$ can be written in the form
    \begin{equation}
        \ket{\Phi^{+}}
        = \frac{1}{\sqrt{2}} \mathopen{}\left( \ket{\sigma_{\chi}^{+}}^{\star} \otimes \ket{\sigma_{\chi}^{+}}
        + \ket{\sigma_{\chi}^{-}}^{\star} \otimes \ket{\sigma_{\chi}^{-}} \right)\mathclose{} ,
    \end{equation}
    where the superscript $\star$ (as opposed to the usual $*$) denotes complex conjugation performed in the computational basis.
    This is such that
    \begin{subequations}
    \begin{gather}
        \ket{\sigma_{\mathrm{x}}^{\pm}}^{\star} = \ket{\sigma_{\mathrm{x}}^{\pm}} , \quad
        \ket{\sigma_{\mathrm{y}}^{\pm}}^{\star} = \ket{\sigma_{\mathrm{y}}^{\mp}} , \quad
        \ket{\sigma_{\mathrm{z}}^{\pm}}^{\star} = \ket{\sigma_{\mathrm{z}}^{\pm}} , \\
        \ket{\sigma_{\mathrm{x} \pm \mathrm{y}}^{+}}^{\star} = \ket{\sigma_{\mathrm{x} \mp \mathrm{y}}^{+}} , \quad
        \ket{\sigma_{\mathrm{x} \pm \mathrm{y}}^{-}}^{\star} = \ket{\sigma_{\mathrm{x} \mp \mathrm{y}}^{-}} .
    \end{gather}
    \end{subequations}
    We then have
    \begin{equation}
    \begin{split}
        \frac{\hat{N}_{\bm a \mid \bm \chi} \ket{\psi^{\prime}}}{\sqrt{\bra{\psi^{\prime}} \hat{N}_{\bm a \mid \bm \chi} \ket{\psi^{\prime}}}}
        ={} & \bigotimes_{j=1}^{n} \ket{\sigma_{\chi_{j}}^{a_{j}}}_{\hilb{A}_{j}^{\prime}}^{\star}
        \otimes \bigotimes_{j=1}^{n} \ket{\sigma_{\chi_{j}}^{a_{j}}}_{\hilb{B}_{j}^{\prime}}
        \otimes \ket{0}_{\hilb{A}^{\prime\prime}} \ket{0}_{\hilb{B}^{\prime\prime}}
        \otimes \ket{\xi_{0}}_{\hilb{A} \hilb{B}} \\
        & + \bigotimes_{j=1}^{n} \ket*{\sigma_{\chi_{j}}^{a_{j} (-1)^{[\chi_{j} = \mathrm{z}]}}}_{\hilb{A}_{j}^{\prime}}^{\star}
        \otimes \bigotimes_{j=1}^{n} \ket*{\sigma_{\chi_{j}}^{a_{j} (-1)^{[\chi_{j} = \mathrm{z}]}}}_{\hilb{B}_{j}^{\prime}}
        \otimes \ket{1}_{\hilb{A}^{\prime\prime}} \ket{1}_{\hilb{B}^{\prime\prime}}
        \otimes \ket{\xi_{1}}_{\hilb{A} \hilb{B}} .
    \end{split}
    \end{equation}
    Tracing out $\hilb{A}$, $\hilb{A}^{\prime}$, and $\hilb{A}^{\prime\prime}$ gives
    \begin{equation}
    \label{eq:alice_traced_trusted}
    \begin{split}
        \tr_{\hilb{A} \hilb{A}^{\prime} \hilb{A}^{\prime\prime}} \mathopen{}\left(
        \frac{\hat{N}_{\bm a \mid \bm \chi} \proj{\psi^{\prime}} \hat{N}_{\bm a \mid \bm \chi}}{\bra{\psi^{\prime}} \hat{N}_{\bm a \mid \bm \chi} \ket{\psi^{\prime}}}
        \right)
        ={} & \bigotimes_{j=1}^{n} \proj{\sigma_{\chi_{j}}^{a_{j}}}_{\hilb{B}_{j}^{\prime}}
        \otimes \proj{0}_{\hilb{B}^{\prime\prime}}
        \otimes \tr_{\hilb{A}}(\proj{\xi_{0}}) \\
        & + \bigotimes_{j=1}^{n} \proj*{\sigma_{\chi_{j}}^{a_{j} (-1)^{[\chi_{j} = \mathrm{z}]}}}_{\hilb{B}_{j}^{\prime}}
        \otimes \proj{1}_{\hilb{B}^{\prime\prime}}
        \otimes \tr_{\hilb{A}}(\proj{\xi_{1}}) \\
        ={} & \proj{e_{\bm a \mid \bm \chi}}_{\hilb{B}^{\prime}}
        \otimes \proj{0}_{\hilb{B}^{\prime\prime}} \otimes \beta_{0}
        + \proj{e_{\bm a \mid \bm \chi}^{*}}_{\hilb{B}^{\prime}}
        \otimes \proj{1}_{\hilb{B}^{\prime\prime}} \otimes \beta_{1} ,
    \end{split}
    \end{equation}
    where $\beta_{0} = \tr_{\hilb{A}}(\proj{\xi_{0}})$ and $\beta_{1} = \tr_{\hilb{A}}(\proj{\xi_{1}})$ and we have $\tr(\beta_{0}) + \tr(\beta_{1}) = 1$.
    Using properties of the partial trace, we also have
    \begin{equation}
    \label{eq:alice_traced_untrusted}
        \tr_{\hilb{A} \hilb{A}^{\prime} \hilb{A}^{\prime\prime}} \mathopen{}\left(
        \frac{V^{\bm \chi} \Pi_{\bm a \mid \bm \chi}^{\hilb{A}} \proj{\psi} \Pi_{\bm a \mid \bm \chi}^{\hilb{A}} {V^{\bm \chi}}^{\dagger}}{\bra{\psi} \Pi_{\bm a \mid \bm \chi}^{\hilb{A}} \ket{\psi}}
        \right)
        = V_{\hilb{B}} \rho_{\hilb{B}}^{\bm a \mid \bm \chi} V_{\hilb{B}}^{\dagger} .
    \end{equation}
    Using the linearity of the partial trace to combine \cref{eq:alice_traced_trusted,eq:alice_traced_untrusted}, and since the trace class norm is decreasing under the partial trace, we have
    \begin{multline}
    \label{eq:trace_norm_partial_trace_compact}
        \left\lVert
        V_{\hilb{B}} \rho_{\hilb{B}}^{\bm a \mid \bm \chi} V_{\hilb{B}}^{\dagger}
        - \left( \proj{e_{\bm a \mid \bm \chi}} \otimes \proj{0} \otimes \beta_{0}
        + \proj{e_{\bm a \mid \bm \chi}^{*}} \otimes \proj{1} \otimes \beta_{1} \right)\mathclose{}
        \right\rVert_{1} \\
        \leq \left\lVert
        \frac{V^{\bm \chi} \Pi_{\bm a \mid \bm \chi}^{\hilb{A}} \proj{\psi} \Pi_{\bm a \mid \bm \chi}^{\hilb{A}} {V^{\bm \chi}}^{\dagger}}{\bra{\psi} \Pi_{\bm a \mid \bm \chi}^{\hilb{A}} \ket{\psi}}
        - \frac{\hat{N}_{\bm a} \proj{\psi^{\prime}} \hat{N}_{\bm a}}{\bra{\psi^{\prime}} \hat{N}_{\bm a} \ket{\psi^{\prime}}}
        \right\rVert_{1} .
    \end{multline}

    Let us introduce a bijection $u \colon \{+, -\}^{n} \to \{0, 1\}^{n}$ which converts between representations of binary strings by taking every entry $+$ to $0$ and every entry $-$ to $1$.
    Let $\bm s \in \{0, 1\}^{n}$ be any string.
    \Cref{eq:local_proj_alice,eq:model_observables_alice} together imply that
    \begin{equation}
    \label{eq:alice_proj_to_obs_untrusted}
        A_{\bm \chi}^{\bm s}
        = \sum_{\bm a \in \{+, -\}^{n}} (-1)^{u(\bm a) \cdot \bm s} \Pi_{\bm a \mid \bm \chi}^{\hilb{A}} .
    \end{equation}
    Defining $\hat{B}_{\bm \chi}^{\bm s}$ by
    \begin{equation}
    \label{eq:alice_proj_to_obs_trusted}
    \begin{split}
        \hat{B}_{\bm \chi}^{\bm s}
        & = \sum_{\bm a \in \{+, -\}^{n}} (-1)^{u(\bm a) \cdot \bm s} \hat{N}_{\bm a \mid \bm \chi} \\
        & = \bigotimes_{j=1}^{n} \mathopen{}\left( \sigma_{\chi_{j}}^{\hilb{B}_{j}^{\prime}} \right)^{s_{j}} \otimes \proj{0}_{\hilb{B}^{\prime\prime}}
        + \bigotimes_{j=1}^{n} \mathopen{}\left( {\sigma_{\chi_{j}}^{\hilb{B}_{j}^{\prime}}}^{*} \right)^{s_{j}} \otimes \proj{1}_{\hilb{B}^{\prime\prime}} ,
    \end{split}
    \end{equation}
    the result of \cref{cor:protocol_isometry} in this notation is that
    \begin{equation}
    \label{eq:protocol_isometry_obs_split_compact}
        \left\lVert
        V^{\bm \chi} A_{\bm \chi}^{\bm s} \ket{\psi}
        - \hat{B}_{\bm \chi}^{\bm s} \ket{\psi^{\prime}}
        \right\rVert
        \leq \gamma(\varepsilon, n) .
    \end{equation}
    Due to \cref{eq:alice_proj_to_obs_untrusted,eq:alice_proj_to_obs_trusted,eq:protocol_isometry_obs_split_compact}, we can now apply \cref{thm:robust_prob} for each $\bm \chi \in \mathcal{S}$.
    This gives, with probability at least $1 - 4 \gamma(\varepsilon, n)^{2/3}$ over all $\bm a \in \{+, -\}^{n}$ given $\bm \chi$, that one half multiplied by the right-hand side of \cref{eq:trace_norm_partial_trace_compact} is bounded above as
    \begin{equation}
        \frac{1}{2} \left\lVert
        \frac{V^{\bm \chi} \Pi_{\bm a \mid \bm \chi}^{\hilb{A}} \proj{\psi} \Pi_{\bm a \mid \bm \chi}^{\hilb{A}} {V^{\bm \chi}}^{\dagger}}{\bra{\psi} \Pi_{\bm a \mid \bm \chi}^{\hilb{A}} \ket{\psi}}
        - \frac{\hat{N}_{\bm a \mid \bm \chi} \proj{\psi^{\prime}} \hat{N}_{\bm a \mid \bm \chi}}{\bra{\psi^{\prime}} \hat{N}_{\bm a \mid \bm \chi} \ket{\psi^{\prime}}}
        \right\rVert_{1}
        \leq \gamma(\varepsilon, n)^{2/3} .
    \end{equation}
    Therefore, with probability at least $1 - 4 \tau(\varepsilon, n)$ over all $\bm a \in \{+, -\}^{n}$ given $\bm \chi$, we have
    \begin{equation}
    \begin{split}
        \frac{1}{2} \Bigl\lVert
        V_{\hilb{B}} \rho_{\hilb{B}}^{\bm a \mid \bm \chi} V_{\hilb{B}}^{\dagger}
        - \Bigl( & \proj{e_{\bm a \mid \bm \chi}} \otimes \proj{0} \otimes \beta_{0} \\
        & + \proj{e_{\bm a \mid \bm \chi}^{*}} \otimes \proj{1} \otimes \beta_{1} \Bigr)\mathclose{}
        \Bigr\rVert_{1}
        \leq \tau(\varepsilon, n) ,
    \end{split}
    \end{equation}
    where we define $\tau(\varepsilon, n) = \gamma(\varepsilon, n)^{2/3}$.
\end{proof}

\printbibliography

@article{clauser1969proposed,
    title={Proposed experiment to test local hidden-variable theories},
    author={Clauser, John F. and Horne, Michael A. and Shimony, Abner and Holt, Richard A.},
    journal={Phys. Rev. Lett.},
    volume={23},
    number={15},
    pages={880--884},
    year={1969},
    month={10},
    publisher={American Physical Society},
    doi={10.1103/PhysRevLett.23.880},
    issn={1079-7114}
}

@inproceedings{mayers1998quantum,
    author={Mayers, Dominic and Yao, Andrew},
    booktitle={Proceedings 39th Annual Symposium on Foundations of Computer Science (Cat. No. 98CB36280)}, 
    title={Quantum cryptography with imperfect apparatus}, 
    year={1998},
    month={11},
    pages={503--509},
    publisher={IEEE},
    venue={Palo Alto, CA, USA},
    issn={0272-5428},
    doi={10.1109/SFCS.1998.743501}
}

@article{raussendorf2001one,
    title={A one-way quantum computer},
    author={Raussendorf, Robert and Briegel, Hans J.},
    journal={Phys. Rev. Lett.},
    volume={86},
    number={22},
    pages={5188--5191},
    year={2001},
    month={5},
    publisher={American Physical Society},
    doi={10.1103/PhysRevLett.86.5188},
    issn={1079-7114}
}

@article{mayers2004self,
    author={Mayers, Dominic and Yao, Andrew},
    title={Self testing quantum apparatus},
    year={2004},
    month={7},
    volume={4},
    number={4},
    journal={Quantum Inf. Comput.},
    pages={273--286},
    publisher={Rinton Press, Incorporated},
    issn={1533-7146},
    doi={10.26421/QIC4.4-3}
}

@article{childs2005unified,
    title={Unified derivations of measurement-based schemes for quantum computation},
    author={Childs, Andrew M. and Leung, Debbie W. and Nielsen, Michael A.},
    journal={Phys. Rev. A},
    volume={71},
    number={3},
    pages={032318},
    year={2005},
    month={3},
    publisher={American Physical Society},
    doi={10.1103/PhysRevA.71.032318},
    issn={2469-9934}
}

@article{navascues2007bounding,
    title={Bounding the set of quantum correlations},
    author={Navascu{\'e}s, Miguel and Pironio, Stefano and Ac{\'i}n, Antonio},
    journal={Phys. Rev. Lett.},
    volume={98},
    number={1},
    pages={010401},
    year={2007},
    month={1},
    publisher={American Physical Society},
    doi={10.1103/PhysRevLett.98.010401},
    issn={1079-7114}
}

@article{danos2007measurement,
    author={Danos, Vincent and Kashefi, Elham and Panangaden, Prakash},
    title={The measurement calculus},
    year={2007},
    month={4},
    journal={J. ACM},
    publisher={Association for Computing Machinery},
    volume={54},
    number={2},
    pages={8–-es},
    issn={0004-5411},
    doi={10.1145/1219092.1219096}
}

@article{navascues2008convergent,
    title={A convergent hierarchy of semidefinite programs characterizing the set of quantum correlations},
    author={Navascu{\'e}s, Miguel and Pironio, Stefano and Ac{\'i}n, Antonio},
    journal={New J. Phys.},
    volume={10},
    number={7},
    pages={073013},
    year={2008},
    month={7},
    publisher={{IOP} Publishing},
    issn={1367-2630},
    doi={10.1088/1367-2630/10/7/073013}
}

@article{broadbent2009parallelizing,
    title={Parallelizing quantum circuits},
    author={Anne Broadbent and Elham Kashefi},
    journal={Theor. Comput. Sci.},
    volume={410},
    number={26},
    pages={2489--2510},
    year={2009},
    month={1},
    publisher={Elsevier},
    issn={0304-3975},
    doi={https://doi.org/10.1016/j.tcs.2008.12.046}
}

@inproceedings{broadbent2009universal,
    author={Broadbent, Anne and Fitzsimons, Joseph and Kashefi, Elham},
    booktitle={2009 50th Annual IEEE Symposium on Foundations of Computer Science}, 
    title={Universal blind quantum computation}, 
    year={2009},
    month={10},
    pages={517--526},
    doi={10.1109/FOCS.2009.36},
    issn={0272-5428}
}

@inproceedings{mckague2011generalized,
    title={Generalized self-testing and the security of the 6-state protocol},
    booktitle={Theory of Quantum Computation, Communication, and Cryptography},
    isbn={978-3-642-18073-6},
    issn={1611-3349},
    doi={10.1007/978-3-642-18073-6_10},
    publisher={Springer Berlin Heidelberg},
    address={Berlin, Heidelberg},
    author={McKague, Matthew and Mosca, Michele},
    editor={van Dam, Wim and Kendon, Vivien M and Severini, Simone},
    year={2011},
    pages={113–-130}
}

@article{colbeck2011private,
    title={Private randomness expansion with untrusted devices},
    author={Colbeck, Roger and Kent, Adrian},
    journal={J. Phys. A Math. Theor.},
    volume={44},
    number={9},
    pages={095305},
    year={2011},
    month={2},
    publisher={{IOP} Publishing},
    issn={1751-8121},
    doi={10.1088/1751-8113/44/9/095305}
}

@article{barz2012demonstration,
    title={Demonstration of blind quantum computing},
    author={Barz, Stefanie and Kashefi, Elham and Broadbent, Anne and Fitzsimons, Joseph F and Zeilinger, Anton and Walther, Philip},
    journal={Science},
    volume={335},
    number={6066},
    pages={303--308},
    year={2012},
    month={1},
    publisher={American Association for the Advancement of Science},
    doi={10.1126/science.1214707}
}

@article{rastegin2012relations,
    title={Relations for certain symmetric norms and anti-norms before and after partial trace},
    author={Rastegin, Alexey E},
    journal={J. Stat. Phys.},
    volume={148},
    number={6},
    pages={1040--1053},
    year={2012},
    month={8},
    publisher={Springer},
    doi={10.1007/s10955-012-0569-8},
    issn={1572-9613}
}

@article{mckague2012robust,
	doi={10.1088/1751-8113/45/45/455304},
	year={2012},
	month={10},
	publisher={{IOP} Publishing},
	volume={45},
	number={45},
	pages={455304},
	author={McKague, Matthew and Yang, Tzyh Haur and Scarani, Valerio},
	title={Robust self-testing of the singlet},
	journal={J. Phys. A Math. Theor.},
	issn={1751-8121}
}

@article{barrett2013memory,
    title={Memory attacks on device-independent quantum cryptography},
    author={Barrett, Jonathan and Colbeck, Roger and Kent, Adrian},
    journal={Phys. Rev. Lett.},
    volume={110},
    number={1},
    pages={010503},
    year={2013},
    month={1},
    publisher={American Physical Society},
    doi={10.1103/PhysRevLett.110.010503},
    issn={1079-7114}
}

@inproceedings{reichardt2013classical1,
    author={Reichardt, Ben W. and Unger, Falk and Vazirani, Umesh},
    title={A classical leash for a quantum system: Command of quantum systems via rigidity of {CHSH} games},
    year={2013},
    month={1},
    isbn={9781450318594},
    publisher={Association for Computing Machinery},
    address={New York, NY, USA},
    doi={10.1145/2422436.2422473},
    booktitle={Proceedings of the 4th Conference on Innovations in Theoretical Computer Science},
    pages={321--322},
    venue={Berkeley, California, USA},
    series={ITCS '13}
}

@article{reichardt2013classical2,
    title={Classical command of quantum systems},
    author={Reichardt, Ben W and Unger, Falk and Vazirani, Umesh},
    journal={Nature},
    volume={496},
    number={7446},
    pages={456--460},
    year={2013},
    month={4},
    publisher={Nature Publishing Group},
    doi={10.1038/nature12035},
    issn={1476-4687}
}

@article{yang2013robust,
    title={Robust self-testing of unknown quantum systems into any entangled two-qubit states},
    author={Yang, Tzyh Haur and Navascu{\'e}s, Miguel},
    journal={Phys. Rev. A},
    volume={87},
    number={5},
    pages={050102},
    year={2013},
    month={5},
    publisher={American Physical Society},
    doi={10.1103/PhysRevA.87.050102},
    issn={2469-9934}
}

@article{barz2013experimental,
    title={Experimental verification of quantum computation},
    author={Barz, Stefanie and Fitzsimons, Joseph F and Kashefi, Elham and Walther, Philip},
    journal={Nat. Phys.},
    volume={9},
    number={11},
    pages={727--731},
    year={2013},
    month={9},
    publisher={Nature Publishing Group},
    issn={1745-2481},
    doi={10.1038/nphys2763}
}

@article{morimae2014verification,
    title={Verification for measurement-only blind quantum computing},
    author={Morimae, Tomoyuki},
    journal={Phys. Rev. A},
    volume={89},
    number={6},
    pages={060302},
    year={2014},
    month={6},
    publisher={American Physical Society},
    doi={10.1103/PhysRevA.89.060302},
    issn={2469-9934}
}

@article{yang2014robust,
    title={Robust and versatile black-box certification of quantum devices},
    author={Yang, Tzyh Haur and V{\'e}rtesi, Tam{\'a}s and Bancal, Jean-Daniel and Scarani, Valerio and Navascu{\'e}s, Miguel},
    journal={Phys. Rev. Lett.},
    volume={113},
    number={4},
    pages={040401},
    year={2014},
    month={7},
    publisher={American Physical Society},
    doi={10.1103/PhysRevLett.113.040401},
    issn={1079-7114}
}

@article{wu2014robust,
    title={Robust self-testing of the three-qubit $W$ state},
    author={Wu, Xingyao and Cai, Yu and Yang, Tzyh Haur and Le, Huy Nguyen and Bancal, Jean-Daniel and Scarani, Valerio},
    journal={Phys. Rev. A},
    volume={90},
    number={4},
    pages={042339},
    year={2014},
    month={10},
    publisher={American Physical Society},
    doi={10.1103/PhysRevA.90.042339},
    issn={2469-9934}
}

@inbook{mancinska2014maximally,
    author={Man{\v{c}}inska, Laura},
    editor={Calude, Cristian S. and Freivalds, R{\=u}si{\c{n}}{\v{s}} and Kazuo, Iwama},
    title={Maximally entangled state in pseudo-telepathy games},
    booktitle={Computing with New Resources},
    year={2014},
    month={12},
    publisher={Springer International Publishing},
    address={Cham},
    pages={200--207},
    isbn={978-3-319-13350-8},
    doi={10.1007/978-3-319-13350-8_15}
}

@inproceedings{dunjko2014composable,
    author={Dunjko, Vedran and Fitzsimons, Joseph F and Portmann, Christopher and Renner, Renato},
    title={Composable security of delegated quantum computation},
    editor={Sarkar, Palash and Iwata, Tetsu},
    booktitle={Advances in Cryptology -- ASIACRYPT 2014},
    year={2014},
    month={12},
    publisher={Springer Berlin Heidelberg},
    address={Berlin, Heidelberg},
    venue={Kaohsiung, Taiwan},
    pages={406--425},
    doi={10.1007/978-3-662-45608-8_22},
    isbn={978-3-662-45608-8}
}

@article{bancal2015physical,
    title={Physical characterization of quantum devices from nonlocal correlations},
    author={Bancal, Jean-Daniel and Navascu{\'e}s, Miguel and Scarani, Valerio and V{\'e}rtesi, Tam{\'a}s and Yang, Tzyh Haur},
    journal={Phys. Rev. A},
    volume={91},
    number={2},
    pages={022115},
    year={2015},
    month={2},
    publisher={American Physical Society},
    doi={10.1103/PhysRevA.91.022115},
    issn={1094-1622}
}

@article{navascues2015almost,
    title={Almost quantum correlations},
    author={Navascu{\'e}s, Miguel and Guryanova, Yelena and Hoban, Matty J and Ac{\'i}n, Antonio},
    journal={Nat. Commun.},
    volume={6},
    number={1},
    pages={1--7},
    year={2015},
    month={2},
    publisher={Nature Publishing Group},
    issn={2041-1723},
    doi={10.1038/ncomms7288}
}

@article{bamps2015sum,
    title={Sum-of-squares decompositions for a family of {C}lauser--{H}orne--{S}himony--{H}olt-like inequalities and their application to self-testing},
    author={Bamps, C{\'e}dric and Pironio, Stefano},
    journal={Phys. Rev. A},
    volume={91},
    number={5},
    pages={052111},
    year={2015},
    month={5},
    publisher={American Physical Society},
    doi={10.1103/PhysRevA.91.052111},
    issn={2469-9934}
}

@article{gheorghiu2015robustness,
    doi={10.1088/1367-2630/17/8/083040},
    year={2015},
    month={8},
    publisher={{IOP} Publishing},
    volume={17},
    number={8},
    pages={083040},
    author={Alexandru Gheorghiu and Elham Kashefi and Petros Wallden},
    title={Robustness and device independence of verifiable blind quantum computing},
    journal={New J. Phys.},
    issn={1367-2630}
}

@article{hayashi2015verifiable,
    title={Verifiable measurement-only blind quantum computing with stabilizer testing},
    author={Hayashi, Masahito and Morimae, Tomoyuki},
    journal={Phys. Rev. Lett.},
    volume={115},
    number={22},
    pages={220502},
    year={2015},
    month={11},
    publisher={American Physical Society},
    doi={10.1103/PhysRevLett.115.220502},
    issn={1079-7114}
}

@misc{hajdusek2015device,
    title={Device-independent verifiable blind quantum computation},
    author={Hajdu{\v{s}}ek, Michal and P{\'e}rez-Delgado, Carlos A and Fitzsimons, Joseph F},
    year={2015},
    month={12},
    eprint={1502.02563},
    archivePrefix={arXiv},
    primaryClass={quant-ph},
    doi={10.48550/ARXIV.1502.02563}
}

@article{wang2016all,
    doi={10.1088/1367-2630/18/2/025021},
    year={2016},
    month={2},
    publisher={{IOP} Publishing},
    volume={18},
    number={2},
    pages={025021},
    author={Yukun Wang and Xingyao Wu and Valerio Scarani},
    title={All the self-testings of the singlet for two binary measurements},
    journal={New J. Phys.},
    issn={1367-2630}
}

@misc{dunjko2016blind,
    title={Blind quantum computing with two almost identical states},
    author={Dunjko, Vedran and Kashefi, Elham},
    year={2016},
    month={4},
    eprint={1604.01586},
    archivePrefix={arXiv},
    primaryClass={quant-ph},
    doi={10.48550/ARXIV.1604.01586}
}

@article{mckague2016self,
    title={Self-testing in parallel},
    volume={18},
    number={4},
    journal={New J. Phys.},
    publisher={{IOP} Publishing},
    author={Matthew McKague},
    year={2016},
    month={4},
    pages={045013},
    issn={1367-2630},
    doi={10.1088/1367-2630/18/4/045013}
}

@article{acin2016optimal,
    title={Optimal randomness certification from one entangled bit},
    author={Ac{\'i}n, Antonio and Pironio, Stefano and V{\'e}rtesi, Tam{\'a}s and Wittek, Peter},
    journal={Phys. Rev. A},
    volume={93},
    number={4},
    pages={040102},
    year={2016},
    month={4},
    publisher={American Physical Society},
    doi={10.1103/PhysRevA.93.040102},
    issn={2469-9934}
}

@article{mckague2016interactive,
    author={McKague, Matthew},
    title={Interactive proofs for $\mathsf{BQP}$ via self-tested graph states},
    year={2016},
    month={6},
    pages={1--42},
    doi={10.4086/toc.2016.v012a003},
    publisher={Theory of Computing},
    journal={Theory Comput.},
    volume={12},
    number={3},
    issn={1557-2862}
}

@article{wu2016device,
    title={Device-independent parallel self-testing of two singlets},
    author={Wu, Xingyao and Bancal, Jean-Daniel and McKague, Matthew and Scarani, Valerio},
    journal={Phys. Rev. A},
    volume={93},
    number={6},
    pages={062121},
    year={2016},
    month={6},
    publisher={American Physical Society},
    doi={10.1103/PhysRevA.93.062121},
    issn={2469-9934}
}

@article{alsina2016experimental,
    title={Experimental test of {Mermin} inequalities on a five-qubit quantum computer},
    author={Alsina, Daniel and Latorre, Jos{\'e} Ignacio},
    journal={Phys. Rev. A},
    volume={94},
    number={1},
    pages={012314},
    year={2016},
    month={7},
    publisher={American Physical Society},
    doi={10.1103/PhysRevA.94.012314},
    issn={2469-9934}
}

@article{kaniewski2016analytic,
    title={Analytic and nearly optimal self-testing bounds for the {Clauser}-{Horne}-{Shimony}-{Holt} and {Mermin} inequalities},
    author={Kaniewski, J{\k{e}}drzej},
    journal={Phys. Rev. Lett.},
    volume={117},
    number={7},
    pages={070402},
    year={2016},
    month={8},
    publisher={American Physical Society},
    doi={10.1103/PhysRevLett.117.070402},
    issn={1079-7114}
}

@article{devitt2016performing,
    title={Performing quantum computing experiments in the cloud},
    author={Devitt, Simon J.},
    journal={Phys. Rev. A},
    volume={94},
    number={3},
    pages={032329},
    year={2016},
    month={9},
    publisher={American Physical Society},
    doi={10.1103/PhysRevA.94.032329},
    issn={2469-9934}
}

@misc{coudron2016parallel,
    title={The parallel-repeated magic square game is rigid}, 
    author={Matthew Coudron and Anand Natarajan},
    year={2016},
    month={9},
    eprint={1609.06306},
    archivePrefix={arXiv},
    primaryClass={quant-ph},
    doi = {10.48550/ARXIV.1609.06306}
}

@article{ostrev2016structure,
    author={Ostrev, Dimiter},
    title={The structure of nearly-optimal quantum strategies for the non-local {XOR} games},
    year={2016},
    month={10},
    publisher={Rinton Press, Incorporated},
    volume={16},
    number={13--14},
    journal={Quantum Inf. Comput.},
    pages={1191--1211},
    doi={10.26421/QIC16.13-14-6},
    issn={1533-7146}
}

@article{gheorghiu2017rigidity,
    author={Alexandru Gheorghiu and Petros Wallden and Elham Kashefi},
    title={Rigidity of quantum steering and one-sided device-independent verifiable quantum computation},
    year={2017},
    month={2},
    journal={New J. Phys.},
    volume={19},
    number={2},
    pages={023043},
    publisher={{IOP} Publishing},
    doi={10.1088/1367-2630/aa5cff},
    issn={1367-2630}
}

@article{castelvecchi2017ibm,
    title={{IBM}'s quantum cloud computer goes commercial},
    author={Castelvecchi, Davide},
    journal={Nature},
    volume={543},
    number={7644},
    pages={159},
    year={2017},
    month={3},
    publisher={Nature Publishing Group},
    doi={10.1038/nature.2017.21585},
    issn={1476-4687}
}

@article{kashefi2017optimised,
    doi={10.1088/1751-8121/aa5dac},
    year={2017},
    month={3},
    publisher={{IOP} Publishing},
    volume={50},
    number={14},
    pages={145306},
    author={Elham Kashefi and Petros Wallden},
    title={Optimised resource construction for verifiable quantum computation},
    journal={J. Phys. A Math. Theor.},
    issn={1751-8121}
}

@misc{aharonov2017interactive,
    title={Interactive proofs for quantum computations},
    author={Aharonov, Dorit and Ben-Or, Michael and Eban, Elad and Mahadev, Urmila},
    year={2017},
    month={4},
    eprint={1704.04487},
    archivePrefix={arXiv},
    primaryClass={quant-ph},
    doi={10.48550/ARXIV.1704.04487}
}

@article{mckague2017self,
    doi={10.22331/q-2017-04-25-1},
    title={Self-testing in parallel with {CHSH}},
    author={McKague, Matthew},
    journal={Quantum},
    issn={2521-327X},
    publisher={Verein zur F{\"{o}}rderung des Open Access Publizierens in den Quantenwissenschaften},
    volume={1},
    pages={1},
    month={4},
    year={2017}
}

@article{coladangelo2017all,
    title={All pure bipartite entangled states can be self-tested},
    author={Coladangelo, Andrea and Goh, Koon Tong and Scarani, Valerio},
    journal={Nature Communications},
    volume={8},
    number={1},
    pages={1--5},
    year={2017},
    month={5},
    publisher={Nature Publishing Group},
    doi={10.1038/ncomms15485},
    issn={2041-1723}
}

@article{kaniewski2017self,
    title={Self-testing of binary observables based on commutation},
    author={Kaniewski, J{\k{e}}drzej},
    journal={Phys. Rev. A},
    volume={95},
    number={6},
    pages={062323},
    year={2017},
    month={6},
    publisher={American Physical Society},
    doi={10.1103/PhysRevA.95.062323},
    issn={2469-9934}
}

@article{fitzsimons2017private,
    title={Private quantum computation: an introduction to blind quantum computing and related protocols},
    author={Fitzsimons, Joseph F},
    journal={npj Quantum Information},
    volume={3},
    number={1},
    pages={1--11},
    year={2017},
    month={6},
    publisher={Nature Publishing Group},
    doi={10.1038/s41534-017-0025-3},
    issn={2056-6387}
}

@inproceedings{natarajan2017quantum,
    author={Natarajan, Anand and Vidick, Thomas},
    title={A quantum linearity test for robustly verifying entanglement},
    year={2017},
    month={6},
    booktitle={Proceedings of the 49th Annual ACM SIGACT Symposium on Theory of Computing},
    pages={1003--1015},
    publisher={Association for Computing Machinery},
    address={New York, NY, USA},
    series={STOC 2017},
    venue={Montreal, Canada},
    doi={10.1145/3055399.3055468}
}

@article{fitzsimons2017unconditionally,
    title={Unconditionally verifiable blind quantum computation},
    author={Fitzsimons, Joseph F. and Kashefi, Elham},
    journal={Phys. Rev. A},
    volume={96},
    number={1},
    pages={012303},
    year={2017},
    month={7},
    publisher={American Physical Society},
    doi={10.1103/PhysRevA.96.012303},
    issn={2469-9934}
}

@article{coladangelo2017parallel,
    author={Coladangelo, Andrea},
    title={Parallel self-testing of (tilted) {EPR} pairs via copies of (tilted) {CHSH} and the magic square game},
    year={2017},
    publisher={Rinton Press, Incorporated},
    volume={17},
    number={9--10},
    journal={Quantum Inf. Comput.},
    month={8},
    pages={831--865},
    doi={10.26421/QIC17.9-10-6},
    issn={1533-7146}
}

@article{fujii2017verifiable,
    title={Verifiable fault tolerance in measurement-based quantum computation},
    author={Fujii, Keisuke and Hayashi, Masahito},
    journal={Phys. Rev. A},
    volume={96},
    number={3},
    pages={030301},
    year={2017},
    month={9},
    publisher={American Physical Society},
    doi={10.1103/PhysRevA.96.030301},
    issn={2469-9934}
}

@article{morimae2017verification,
    title={Verification of hypergraph states},
    author={Morimae, Tomoyuki and Takeuchi, Yuki and Hayashi, Masahito},
    journal={Phys. Rev. A},
    volume={96},
    number={6},
    pages={062321},
    year={2017},
    month={12},
    publisher={American Physical Society},
    doi={10.1103/PhysRevA.96.062321},
    issn={2469-9934}
}

@article{gowers2017inverse,
    author={W. T. Gowers and O. Hatami},
    title={Inverse and stability theorems for approximate representations of finite groups},
    year={2017},
    month={12},
    publisher={London Mathematical Society, Turpion Ltd and the Russian Academy of Sciences},
    volume={208},
    number={12},
    pages={1784},
    journal={Sb. Math.},
    doi={10.1070/SM8872},
    issn={1064-5616}
}

@article{fitzsimons2018post,
    title={\textit{Post hoc} verification of quantum computation},
    author={Fitzsimons, Joseph F. and Hajdu{\v{s}}ek, Michal and Morimae, Tomoyuki},
    journal={Phys. Rev. Lett.},
    volume={120},
    number={4},
    pages={040501},
    year={2018},
    month={1},
    publisher={American Physical Society},
    doi={10.1103/PhysRevLett.120.040501},
    issn={1079-7114}
}

@article{lin2018device,
    title={Device-independent point estimation from finite data and its application to device-independent property estimation},
    author={Lin, Pei-Sheng and Rosset, Denis and Zhang, Yanbao and Bancal, Jean-Daniel and Liang, Yeong-Cherng},
    journal={Phys. Rev. A},
    volume={97},
    number={3},
    pages={032309},
    year={2018},
    month={3},
    publisher={American Physical Society},
    doi = {10.1103/PhysRevA.97.032309},
    issn={2469-9934}
}

@article{hayashi2018self,
    title={Self-guaranteed measurement-based quantum computation},
    author={Hayashi, Masahito and Hajdu{\v{s}}ek, Michal},
    journal={Phys. Rev. A},
    volume={97},
    number={5},
    pages={052308},
    year={2018},
    month={5},
    publisher={American Physical Society},
    doi={10.1103/PhysRevA.97.052308},
    issn={2469-9934}
}

@article{broadbent2018verify,
    author={Broadbent, Anne},
    title={How to verify a quantum computation},
    year={2018},
    month={6},
    pages={1--37},
    doi={10.4086/toc.2018.v014a011},
    publisher={Theory of Computing},
    journal={Theory Comput.},
    volume={14},
    number={11},
    issn={1557-2862}
}

@inproceedings{arnon2018noise,
    author={Rotem Arnon-Friedman and Henry Yuen},
    title={Noise-tolerant testing of high entanglement of formation},
    booktitle={45th International Colloquium on Automata, Languages, and  Programming (ICALP 2018)},
    pages={11:1--11:12},
    series={Leibniz International Proceedings in Informatics (LIPIcs)},
    venue={Prague, Czech Republic},
    year={2018},
    month={7},
    volume={107},
    editor={Ioannis Chatzigiannakis and Christos Kaklamanis and D{\'a}niel Marx and Donald Sannella},
    publisher={Schloss Dagstuhl -- Leibniz-Zentrum f{\"u}er Informatik},
    address={Dagstuhl, Germany},
    doi={10.4230/LIPIcs.ICALP.2018.11},
    isbn={978-3-95977-076-7},
    issn={1868-8969}
}

@article{chao2018test,
    title={Test for a large amount of entanglement, using few measurements},
    author={Chao, Rui and Reichardt, Ben W and Sutherland, Chris and Vidick, Thomas},
    journal={Quantum},
    volume={2},
    pages={92},
    year={2018},
    month={9},
    publisher={Verein zur F{\"{o}}rderung des Open Access Publizierens in den Quantenwissenschaften},
    doi={10.22331/q-2018-09-03-92},
    issn={2521-327X}
}

@inproceedings{natarajan2018low,
    author={Natarajan, Anand and Vidick, Thomas},
    booktitle={2018 IEEE 59th Annual Symposium on Foundations of Computer Science (FOCS)},
    title={Low-degree testing for quantum states, and a quantum entangled games {PCP} for {QMA}},
    venue={Paris, France},
    year={2018},
    month={10},
    pages={731--742},
    publisher={IEEE},
    doi={10.1109/FOCS.2018.00075},
    issn={2575-8454}
}

@inproceedings{mahadev2018classical,
    author={Mahadev, Urmila},
    booktitle={2018 IEEE 59th Annual Symposium on Foundations of Computer Science (FOCS)},
    title={Classical verification of quantum computations},
    year={2018},
    month={10},
    pages={259--267},
    doi={10.1109/FOCS.2018.00033},
    issn={2575-8454}
}

@article{bowles2018self,
    title={Self-testing of Pauli observables for device-independent entanglement certification},
    author={Bowles, Joseph and {\v{S}}upi{\'c}, Ivan and Cavalcanti, Daniel and Ac{\'i}n, Antonio},
    journal={Phys. Rev. A},
    volume={98},
    number={4},
    pages={042336},
    year={2018},
    month={10},
    publisher={American Physical Society},
    doi={10.1103/PhysRevA.98.042336},
    issn={2469-9934}
}

@article{coladangelo2018generalization,
    title={Generalization of the {Clauser}-{Horne}-{Shimony}-{Holt} inequality self-testing maximally entangled states of any local dimension},
    author={Coladangelo, Andrea},
    journal={Phys. Rev. A},
    volume={98},
    number={5},
    pages={052115},
    year={2018},
    month={11},
    publisher={American Physical Society},
    doi={10.1103/PhysRevA.98.052115},
    issn={2469-9934}
}

@article{sekatski2018certifying,
    title={Certifying the building blocks of quantum computers from {B}ell's theorem},
    author={Sekatski, Pavel and Bancal, Jean-Daniel and Wagner, Sebastian and Sangouard, Nicolas},
    journal={Phys. Rev. Lett.},
    volume={121},
    number={18},
    pages={180505},
    year={2018},
    month={11},
    publisher={American Physical Society},
    doi={10.1103/PhysRevLett.121.180505},
    issn={1079-7114}
}

@article{li2019analytic,
    author={Xinhui Li and Yukun Wang and Yunguang Han and Sujuan Qin and Fei Gao and Qiaoyan Wen},
    title={Analytic robustness bound for self-testing of the singlet with two binary measurements},
    journal={J. Opt. Soc. Am. B},
    volume={36},
    number={2},
    pages={457--463},
    year={2019},
    month={2},
    publisher = {Optica Publishing Group},
    doi={10.1364/JOSAB.36.000457},
    issn={1520-8540}
}

@article{arnon2019device,
    author={Rotem Arnon-Friedman and Jean-Daniel Bancal},
    title={Device-independent certification of one-shot distillable entanglement},
    journal={New J. Phys.},
    volume={21},
    number={3},
    pages={033010},
    publisher={IOP Publishing},
    year={2019},
    month={3},
    doi={10.1088/1367-2630/aafef6},
    issn={1367-2630}
}

@inproceedings{coladangelo2019verifier,
    author={Coladangelo, Andrea and Grilo, Alex B and Jeffery, Stacey and Vidick, Thomas},
    editor={Ishai, Yuval and Rijmen, Vincent},
    title={Verifier-on-a-leash: New schemes for verifiable delegated quantum computation, with quasilinear resources},
    booktitle={Advances in Cryptology -- {EUROCRYPT} 2019},
    venue={Darmstadt, Germany},
    year={2019},
    month={4},
    publisher={Springer International Publishing},
    address={Cham},
    pages={247--277},
    doi={10.1007/978-3-030-17659-4_9},
    isbn={978-3-030-17659-4}
}

@inproceedings{gheorghiu2019computationally,
    author={Gheorghiu, Alexandru and Vidick, Thomas},
    booktitle={2019 IEEE 60th Annual Symposium on Foundations of Computer Science (FOCS)},
    title={Computationally-secure and composable remote state preparation},
    year={2019},
    month={11},
    pages={1024--1033},
    doi={10.1109/FOCS.2019.00066},
    issn={2575-8454}
}

@inproceedings{natarajan2019neexp,
    author={Natarajan, Anand and Wright, John},
    booktitle={2019 IEEE 60th Annual Symposium on Foundations of Computer Science (FOCS)},
    title={$\mathsf{NEEXP}$ is contained in $\mathsf{MIP}^{*}$},
    year={2019},
    month={11},
    pages={510--518},
    publisher={IEEE},
    doi={10.1109/FOCS.2019.00039},
    issn={2575-8454}
}

@article{xu2020improved,
    author={Xu, Qingshan and Tan, Xiaoqing and Huang, Rui},
    title={Improved resource state for verifiable blind quantum computation},
    journal={Entropy},
    publisher={MDPI},
    volume={22},
    year={2020},
    month={9},
    number={9},
    pages={996},
    issn={1099-4300},
    doi={10.3390/e22090996}
}

@article{supic2020self,
    title={Self-testing of quantum systems: A review},
    author={{\v{S}}upi{\'c}, Ivan and Bowles, Joseph},
    journal={Quantum},
    volume={4},
    pages={337},
    year={2020},
    month={9},
    publisher={Verein zur F{\"{o}}rderung des Open Access Publizierens in den Quantenwissenschaften},
    issn={2521-327X},
    doi={10.22331/q-2020-09-30-337}
}

@article{adamson2020quantum,
    title={Quantum magic rectangles: Characterization and application to certified randomness expansion},
    author={Adamson, Sean A. and Wallden, Petros},
    journal={Phys. Rev. Research},
    volume={2},
    number={4},
    pages={043317},
    year={2020},
    month={12},
    publisher={American Physical Society},
    doi={10.1103/PhysRevResearch.2.043317},
    issn={2643-1564}
}

@misc{mancinska2021constant,
    title={Constant-sized robust self-tests for states and measurements of unbounded dimension}, 
    author={Man{\v{c}}inska, Laura and Prakash, Jitendra and Schafhauser, Christopher},
    year={2021},
    month={3},
    eprint={2103.01729},
    archivePrefix={arXiv},
    primaryClass={quant-ph},
    doi = {10.48550/ARXIV.2103.01729}
}

@article{supic2021device,
    doi={10.22331/q-2021-03-23-418},
    title={Device-independent certification of tensor products of quantum states using single-copy self-testing protocols},
    author={{\v{S}}upi{\'c}, Ivan and Cavalcanti, Daniel and Bowles, Joseph},
    journal={Quantum},
    issn={2521-327X},
    publisher={Verein zur F{\"{o}}rderung des Open Access Publizierens in den Quantenwissenschaften},
    volume={5},
    pages={418},
    month={3},
    year={2021}
}

@article{agresti2021experimental,
    title={Experimental Robust Self-Testing of the State Generated by a Quantum Network},
    author={Agresti, Iris and Polacchi, Beatrice and Poderini, Davide and Polino, Emanuele and Suprano, Alessia and {\v{S}}upi{\'c}, Ivan and Bowles, Joseph and Carvacho, Gonzalo and Cavalcanti, Daniel and Sciarrino, Fabio},
    journal={PRX Quantum},
    volume={2},
    number={2},
    pages={020346},
    year={2021},
    month={6},
    publisher={American Physical Society},
    doi={10.1103/PRXQuantum.2.020346},
    issn={2691-3399}
}

@article{metger2021self,
    author={Metger, Tony and Vidick, Thomas},
    title={Self-testing of a single quantum device under computational assumptions},
    journal={Quantum},
    volume={5},
    pages={544},
    year={2021},
    month={9},
    publisher={Verein zur F{\"o}rderung des Open Access Publizierens in den Quantenwissenschaften},
    doi={10.22331/q-2021-09-16-544},
    issn={2521-327X}
}

@article{sarkar2021self,
    title={Self-testing quantum systems of arbitrary local dimension with minimal number of measurements},
    author={Sarkar, Shubhayan and Saha, Debashis and Kaniewski, J{\k{e}}drzej and Augusiak, Remigiusz},
    journal={{npj} Quantum Information},
    volume={7},
    number={1},
    pages={151},
    year={2021},
    month={10},
    publisher={Nature Publishing Group},
    issn={2056-6387},
    doi={10.1038/s41534-021-00490-3}
}

@article{kashefi2021securing,
    title={Verifying {BQP} computations on noisy devices with minimal overhead},
    author={Leichtle, Dominik and Music, Luka and Kashefi, Elham and Ollivier, Harold},
    journal={PRX Quantum},
    volume={2},
    number={4},
    pages={040302},
    year={2021},
    month={10},
    publisher={American Physical Society},
    doi={10.1103/PRXQuantum.2.040302},
    issn={2691-3399}
}

@article{renou2021quantum,
    title={Quantum theory based on real numbers can be experimentally falsified},
    author={Renou, Marc-Olivier and Trillo, David and Weilenmann, Mirjam and Le, Thinh P and Tavakoli, Armin and Gisin, Nicolas and Ac{\'i}n, Antonio and Navascu{\'e}s, Miguel},
    journal={Nature},
    volume={600},
    number={7890},
    pages={625--629},
    year={2021},
    month={12},
    publisher={Nature Publishing Group},
    doi={10.1038/s41586-021-04160-4},
    issn={1476-4687}
}

@article{fu2022constant,
    doi={10.22331/q-2022-01-03-614},
    title={Constant-sized correlations are sufficient to self-test maximally entangled states with unbounded dimension},
    author={Fu, Honghao},
    journal={Quantum},
    issn={2521-327X},
    publisher={Verein zur F{\"{o}}rderung des Open Access Publizierens in den Quantenwissenschaften},
    volume={6},
    pages={614},
    month={1},
    year={2022}
}

@article{adamson2022practical,
    title={Practical parallel self-testing of {Bell} states via magic rectangles},
    author={Adamson, Sean A. and Wallden, Petros},
    journal={Phys. Rev. A},
    volume={105},
    number={3},
    pages={032456},
    year={2022},
    month={3},
    publisher={American Physical Society},
    doi={10.1103/PhysRevA.105.032456},
    issn={2469-9934}
}

@article{portmann2022security,
    title={Security in quantum cryptography},
    author={Portmann, Christopher and Renner, Renato},
    journal={Rev. Mod. Phys.},
    volume={94},
    number={2},
    pages={025008},
    year={2022},
    month={6},
    publisher={American Physical Society},
    doi={10.1103/RevModPhys.94.025008}
}

@article{nadlinger2022experimental,
    title={Experimental quantum key distribution certified by {Bell}'s theorem},
    author={Nadlinger, David P. and Drmota, Peter and Nichol, Bethan C. and Araneda, Gabriel and Main, Dougal and Srinivas, Raghavendra and Lucas, David M. and Ballance, Christopher J. and Ivanov, Kirill and Tan, Ernest Y.-Z. and Sekatski, Pavel and Urbanke, R{\"u}diger Leo and Renner, Renato and Sangouard, Nicolas and Bancal, Jean-Daniel},
    journal={Nature},
    volume={607},
    number={7920},
    pages={682--686},
    year={2022},
    month={7},
    publisher={Nature Publishing Group},
    doi={10.1038/s41586-022-04941-5},
    issn={1476-4687}
}

@inproceedings{gheorghiu2023quantum,
    author={Gheorghiu, Alexandru and Metger, Tony and Poremba, Alexander},
    title={Quantum cryptography with classical communication: Parallel remote state preparation for copy-protection, verification, and more},
    booktitle={50th International Colloquium on Automata, Languages, and Programming (ICALP 2023)},
    pages={67:1--67:17},
    series={Leibniz International Proceedings in Informatics (LIPIcs)},
    venue={Paderborn, Germany},
    isbn={978-3-95977-278-5},
    issn={1868-8969},
    year={2023},
    month={7},
    volume={261},
    editor={Etessami, Kousha and Feige, Uriel and Puppis, Gabriele},
    publisher={Schloss Dagstuhl -- Leibniz-Zentrum f{\"u}r Informatik},
    address={Dagstuhl, Germany},
    doi={10.4230/LIPIcs.ICALP.2023.67}
}

@article{supic2023quantum,
    title={Quantum networks self-test all entangled states},
    author={{\v{S}}upi{\'c}, Ivan and Bowles, Joseph and Renou, Marc-Olivier and Ac{'\i}n, Antonio and Hoban, Matty J},
    journal={Nat. Phys.},
    volume={19},
    number={5},
    pages={670--675},
    year={2023},
    month={2},
    publisher={Nature Publishing Group},
    issn={1745-2481},
    doi={10.1038/s41567-023-01945-4}
}

@misc{ibmquantum,
    author={{IBM Research}},
    title={{IBM Quantum}},
    url={https://quantum-computing.ibm.com/},
    howpublished={Online (visited on 2022-12-08)}
}

@misc{amazonbraket,
    author={{Amazon Web Services}},
    title={{Amazon Braket}},
    url={https://aws.amazon.com/braket/},
    howpublished={Online (visited on 2022-12-08)}
}

@misc{azurequantum,
    author={{Microsoft}},
    title={{Azure Quantum}},
    url={https://azure.microsoft.com/products/quantum/},
    howpublished={Online (visited on 2022-12-08)}
}

\end{document}